\documentclass[a4paper]{article} 

\pdfoutput=1 

\usepackage{amsmath,amssymb}
\usepackage{fullpage}
\usepackage{hyperref}
\usepackage{natbib,multicol} 
\usepackage{amsfonts,color} 
\usepackage{amsthm} 

\usepackage{sidecap}   
\sidecaptionvpos{figure}{c}
\usepackage{wrapfig}   
\usepackage{graphicx} 
\DeclareGraphicsExtensions{.png,.jpg,.pdf}

\newcommand{\dimension}[1]{$#1$\nobreakdash-dimension}
\newcommand{\ndim}[1]{$#1$\nobreakdash-dimensional}	
\newcommand{\dimensional}[1]{\ndim{#1}}
\newcommand{\xaxis}[1]{$#1$\nobreakdash-axis}				
\newcommand{\manifold}[1][4]{$#1$\nobreakdash-manifold}	

\newcommand{\coords}[1]{$#1$\nobreakdash-coordinates}
\newcommand{\nvector}[1][4]{$#1$\nobreakdash-vector}							
\newcommand{\position}[1][4]{$#1$\nobreakdash-position}					
\newcommand{\direction}[1]{$#1$\nobreakdash-direction}
\newcommand{\coordinate}[1]{$#1$\nobreakdash-coordinate}
\newcommand{\component}[1]{$#1$\nobreakdash-component}

\newcommand{\momentum}[1][4]{$#1$\nobreakdash-momentum}
\newcommand{\velocity}[1][4]{$#1$\nobreakdash-velocity}

\newcommand{\acceleration}[1][4]{$#1$\nobreakdash-acceleration}

\newcommand{\sphere}[1]{$#1$\nobreakdash-sphere}

\newcommand{\fourvel}{\velocity[4]}                 
\newcommand{\fvec}[1]{\mathbf{#1}}								 

\newcommand{\Schw}[1][r]{\left(1-\frac{2M}{#1}\right)}	                 
\newcommand{\sqrtSchw}[1][r]{\sqrt{1-\frac{2M}{#1}}}
\newcommand{\Schwroot}{\sqrt\frac{2M}{r}}												 

\newcommand{\fish}{\left(1-\frac{2M}{r}+K^2\right)}		
\newcommand{\sqrtfish}{\sqrt{1-\frac{2M}{r}+K^2}}		
\newcommand{\ffish}{\Schw^{-1}\sqrtfish}				

\DeclareMathOperator{\sech}{sech}                                

\DeclareMathOperator{\diag}{diag}											
\newcommand{\asin}{\sin^{-1}}			
\newcommand{\atanh}{\tanh^{-1}}		
\usepackage[usenames]{xcolor}

\newtheorem{thrm}{Theorem}[section]

\newtheorem{lemma}[thrm]{Lemma}

\newcommand{\abs}[1]{\left\lvert#1\right\rvert}              

\newcommand{\diff}[2]{\frac{d#1}{d#2}}                       
\newcommand{\pdiff}[2]{\frac{\partial#1}{\partial#2}}           
\newcommand{\pdiffnpower}[3]{\left(\pdiff{#2}{#3}\right)^{#1}}
\newcommand{\pdiffsquared}[2]{\pdiffnpower{2}{#1}{#2}}
\newcommand{\eqn}[1]{\begin{equation}#1\end{equation}}
\newcommand{\eqnboxed}[1]{\begin{equation}\boxed{#1}\end{equation}}

\newcommand{\casesifthree}[6]{\begin{cases}#1&\text{if }#2\\#3&\text{if }#4\\#5&\text{if }#6\end{cases}}

\begin{document}
\include{jdf} 


\titlepage

  \thispagestyle{empty}
    
    \begin{figure}[h]
	\begin{center}
	\includegraphics[scale=1]{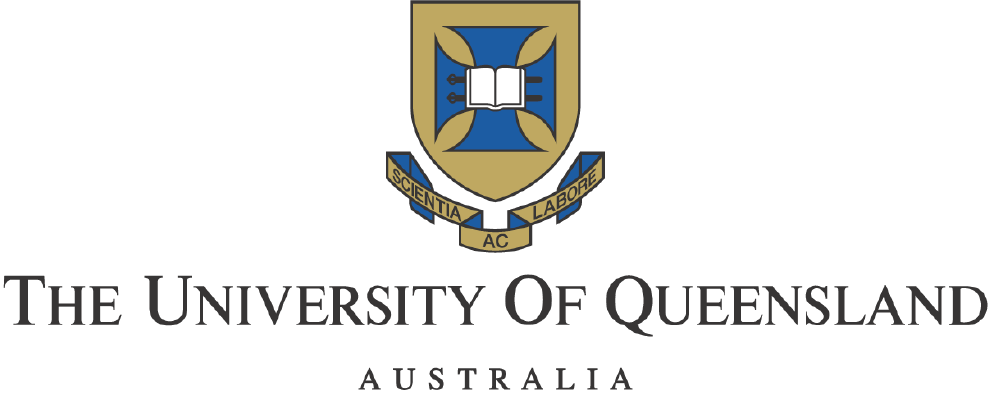}
	\end{center}
    \end{figure}
	\vfill
	
    \begin{center}
\huge\sc{Expanding Space, Redshifts, and Rigidity}\\
\vspace{0.1in}
\LARGE\sc\expandafter{Conceptual issues in cosmology}
    \end{center}
    \vfill
    
    \begin{center}
        \LARGE\sc{Colin MacLaurin}\\
        \large{BMath (Hons I), Grad Dip Theol}
    \end{center}
    	\vspace{0.2in}
	 \begin{center}
        \Large{Under the supervision of}\\
        \Large{Prof. Tamara Davis}
    \end{center}
    \vfill
    
    \begin{center}
        \large\sc{A thesis submitted to the University of Queensland\\
in partial fulfilment of the degree of Master of Science\\
Faculty of Science\\
School of Mathematics and Physics\\
June 2015}
\end{center}

\newpage
\pagenumbering{roman}
\renewcommand{\thepage}{- \roman{page} -}


\begin{changemargin}{2cm}
\section*{Abstract}
I examine the interpretation of photon redshifts in curved spacetime, as being gravitational or Doppler in origin. In Friedmann-Lema\^itre-Robertson-Walker spacetime, redshifts between comoving observers are often attributed to ``expanding space'', whereas in Schwarzschild spacetime, redshifts between static observers are attributed to ``gravitational'' causes. Yet various authors have suggested a freely falling observer congruence would interpret any redshift as Doppler, whereas a rigid congruence must interpret it as gravitational since there is no relative motion. I realise this proposal by explicitly constructing coordinate systems for rigid motion in the above spacetimes. This includes an extensive analysis of observer-dependent distance measurement in curved spacetime. I also introduce Rindler acceleration, the Milne model, and Newtonian cosmology.
\end{changemargin}

\vspace{1cm}

\emph{Note:} This version is from 2015, updated two weeks after formally submitting the thesis. It has been available on my personal website since January 2017. The above abstract is new, along with cosmetic improvements to the references. \qquad --- CM, 13th Nov. 2019\footnote{Email: colin.maclaurin@uqconnect.edu.au}

\vspace{1cm}

\section*{Disclaimer}
\begin{center}
The work presented in this thesis is,\\
to the best of my knowledge and belief original,\\
except as acknowledged in the text,\\
and has not been submitted either in whole or in part,\\
for a degree at this or any other university.\\
\vspace{0.5in}
\rule{35mm}{.1pt}\\
Colin MacLaurin
\end{center}

\vspace{0cm}

\section*{Dedication}

Tamara has been an ideal supervisor, for technical guidance, communication, organisation, and cheerful enthusiasm.

My parents Ann and Normand have very generously supported me through finance and home-cooked meals.

Finally to all the good people I've ever met: exciting as relativity is, it is you that makes life worth living.

\qquad\qquad\qquad ``No one makes it alone.'' \quad --- Malcolm Gladwell

\newpage

\section*{Original abstract (2015)}
The expansion of the universe is usually described as an expansion of space itself, but some have argued it is motion through space. A related question is what causes the cosmological redshift. This turns out to be closely related to the concept of rigidity. I critically examine these areas, and derive numerous original results.

I overview two flat spacetime analogues to the standard general relativistic Friedmann-Lema\^itre-Robertson-Walker (FLRW) model of the universe. The Milne model describes the empty universe case within Minkowski space. Newtonian gravity replicates the Friedmann equations which describe the rate of expansion. In both models, the expansion would normally be interpreted as motion.

I show redshifts to have a flexible interpretation, by inciting the equivalence principle for contrived motions of observer families. It is straightforward to induce a Doppler redshift interpretation, but the gravitational redshift case is difficult and coincides with the problem of determining rigid motions. This in turn requires a solid conceptual foundation of distance measurements in relativity, especially distance as measured in the local inertial or proper frame of an observer.

In the Schwarzschild geometry, the distance which receives the most attention in sources is the radial ``proper distance'' interval $\Schw^{-1/2}dr$. In the FLRW case, the usual focus is on the radial distance interval $Rd\chi$. I show these are indeed the most natural or canonical choices. However sources which imply an absolute nature for these quantities are misleading, because most observers would measure a different distance depending on their motion. Though the proper distance is a coordinate-independent quantity for a given worldline, the common choice of constant coordinate time slice used to define the worldline in the first place is coordinate-dependent. It is only observers stationary in the radial direction, that is $\diff{r}{t}=0$ and $\diff{\chi}{t}=0$ respectively, who measure these distances.

I demonstrate how to compute length-contraction within general relativity. This effect is computed in local inertial frames, and occurs with the same Lorentz factor $\gamma$ from special relativity. One approach uses local Lorentz transformations to obtain contraction by $\gamma$. An alternate but complementary derivation is based on the spatial metric, and given an expansion by $\gamma$, interpreted as due to measurement by length-contracted rulers. These results aid in determining proper-frame distances. I defend the concept of rigid rulers for distance measurement, though photon radar is now in vogue. The metric does not yield ruler distances in general, because these also depend on the motion of the rulers. A natural time slicing is along orthogonal hypersurfaces to the flow, which exist when the flow has zero helicity.

In flat spacetime, redshifts may be interpreted as purely gravitational by matching the velocities of observers to the Rindler coordinates, which describe accelerated rigid motion. I derive the kinematics of a Schwarzschild ``fishing line'', a \dimensional{1} rigid object reeled at a constant rate. Near a black hole, such a line has infinite proper length. This description yields a redshift interpretation with different gravitational and Doppler contributions from the usual interpretation. I also derive the kinematics of a \dimensional{3} revolving object in Schwarzschild spacetime. This demonstrates accelerated \dimensional{3} rigid motions do indeed exist in curved spacetime. It may also be used to flexibly interpret redshifts. I also quantify the behaviour of a long rigid cable extended from a comoving galaxy in an FLRW universe, by presenting a system of differential equations. In the case of a universe with only dark energy, such a cable would have proper length $\frac{\pi}{2}H_0^{-1}$ before reaching the Hubble sphere at (Hubble flow) proper distance $H_0^{-1}$, due to length-contraction. This rigid system also leads to an alternate decomposition of the Doppler and gravitational contributions of redshift.

The offset of the Milky Way from the Hubble flow implies length-contraction by $\approx 0.0002\%$ in the direction of motion, relative to the measurement of Hubble comovers. For Earth's orbit around the Sun, length-contraction of its rulers leads to a measured orbit length increase of $9$km, as measured by ``stationary'' observers. A potential application not analysed is of an astronaut plummeting towards a black hole, who would experience length-contraction and hence a different experience of tidal forces than usual analyses describe.

Overall, the interpretation of redshifts is highly flexible. Nonetheless the standard interpretation in the Schwarzschild and FLRW situations is the most natural. Expansion of space is indeed a valid metaphor for the expansion of the universe, though it has benefited from critique.

\newpage
\tableofcontents


\newpage
\pagenumbering{arabic}
\renewcommand{\thepage}{- \arabic{page} -}

\section{Introduction}

\subsection{Expansion}
The expansion of the universe is one of the key concepts upon which cosmology is derived. The primary evidence for this expansion is the redshift of light and other electromagnetic radiation, meaning an increase in wavelength between emission and reception. More distant galaxies are more redshifted, indicating they are receding from us faster than nearby galaxies. This expansion is modelled in general relativity by the Friedmann-Lema\^itre-Robertson-Walker (FLRW) spacetime, which describes a homogeneous and isotropic matter distribution. The rate of expansion is then derived from the Friedmann equations, which follow from Einstein's field equations.

But could the expansion be equally well explained by other coordinate systems? This has already been done for special cases such as the empty universe, whose flat spacetime is also described by the Minkowski metric. Interestingly, when applying a different coordinate system, our interpretation and language used to describe that system change. When using the FLRW metric this coordinate description suggests the interpretation that space is expanding, because of the scale factor term $R(t)$. On the other hand, when using the Minkowski metric one usually describes galaxies as moving through a fixed space. But of course, deriving coordinate-independent properties is paramount. I investigate what is invariant, the flexibility of coordinate-dependent interpretations as well as their limits, as well as the intermediate concept of ``natural'' or ``canonical'' --- where a symmetry suggests the choice of a particular coordinate system or in which the laws of physics might be expressed in a more convenient form.

As a more down-to-earth example, consider a sports field, which would usually be described using a flat coordinate system. One may describe a player who runs 40 metres from their goal line directly down the field. In other coordinate systems, this could be described as running so many degrees longitude, or as 40 metres east. There may be coordinate systems which are more simple and familiar than others, but that doesn't mean they are the only possibilities. In my thesis I explore various possible choices of coordinate system for describing our expanding universe. Through analysing these options, I put to the test some of the assumptions about our universe, and try to assess what is a necessary part of physics and what is a contingent aspect resulting merely from a choice of coordinate system.

For example, what does it mean to say space is expanding, or that space is stretching? These common phrases are used to intuitively describe the expansion of the universe, but they do have some limitations. 

In ``A diatribe on expanding space'', \citet[\S2]{peacock2001} asserts `` `expanding space' is in general a dangerously flawed way of thinking about an expanding universe.'' In contrast \citet[\S3.1]{davislineweaver2004} explain the Hubble recession velocity as due to the ``expansion of space, not movement through space''. One argument or test case about expanding space is that local inhomogeneities (such as a gravitationally bound galaxy of stars) do not show a Hubble expansion. Another test is the ``tethered galaxy'' scenario of a particle initially kept at constant ``proper distance'', which when released will not join the local Hubble flow, but in the case of a cosmological constant dominated universe the galaxy will actually end up on the far side of the origin \citep{davis+2003} \citep{peacock2001}!

\subsection{Redshift}
One of the potentially coordinate-dependent aspects of the expanding universe is what causes the cosmological redshift. This is commonly described as being due to the expansion of space. I investigate if it can be interpreted as a flexible combination of a Doppler shift and a gravitational shift, depending on the coordinate system chosen. This is done by appealing to the equivalence principle.

Cosmological redshift can be understood as an accumulation of Doppler shifts, as argued by the major influences on this thesis \citep{peacock2001} \citep{bunnhogg2009} \cite{davis2004}. Bunn and Hogg go further to argue that even the overall redshift can be understood as a Doppler shift. Though the Hubble recession velocity does not fit the special relativistic Doppler shift formula, a different measure of velocity is found by parallel propagating the galaxy \velocity{} to the observer's location, then the resulting speed difference does in fact fit this formula. They propose the use of the equivalence principle to set up interpretations of redshifts as being purely Doppler or purely gravitational in origin. The concept is to have a line of observers along a photon path, or equivalently a coordinate system. Then a purely Doppler interpretation would be achieved if the observers are all in freefall, because the absence of acceleration would seem the same as absence of gravity, so the redshift could not seem gravitational. On the other hand a purely gravitational interpretation is achieved if the observers maintain a constant distance from one another, since by perceiving no motion they naturally interpret the redshift as gravitational. 

\citet{chodorowski2005} argues that flat space models are less inferior to the FLRW model than one might expect. He is nonetheless clear that these models are of merely ``pedagogical'' value and are not serious contenders. In particular, the Milne and FLRW models share redshift interpretations. \citet{chodorowski2011} also derives a purely kinematic interpretation of redshifts. \citet[\S4]{gronelgaroy2007} argue Newtonian and special relativistic models of the expansion fail, and suggest a decomposition of cosmological redshift into Doppler and gravitational effects but with strict limits. \citet{peacock2001} proposes a redshift formula with combined Doppler and gravitational contributions.

It turns out the gravitational interpretation requires a system of observers at constant distance from one-another, which is the same property as rigidity.

\subsection{Rigidity}
A rigid body is one which maintains its shape over time, which in the strictest sense is Born-rigidity \citep{born1909} in which the constituent particles maintain exactly constant distance from one-another. Hence the kinematics of rigid bodies are identical to the kinematics required to interpret redshifts as gravitational. The background theory has at times been better developed in the continuum mechanics literature, so one could import concepts and terminology from that discipline.

The distance measured in the proper frame of a moving observer is given by the spatial projector $P_{\mu\nu}$, a quantity found in any advanced textbook on general relativity. This tensor may be expressed in any coordinate system, but a particularly nice choice is coordinates which are comoving for a given particle flow. This choice amounts to tracking particles over time, and is called the ``proper metric'' by one author. This little-known quantity is due to Souriau in 1958 and \citet{maugin1971}.

\citet{brotas2006} claimed to describe the rigid motion of a \dimensional{1} line in the Schwarzschild geometry. Such a discovery would enable flexible interpretation of redshifts. However his results are demonstrably wrong on several points. In particular, the new coordinate system he proposes fails to describe Schwarzschild spacetime as intended.

\citet{bunnhogg2009} propose fitting a rigid chain of observers between two endpoints corresponding to the emission and reception of a photon. They require the velocities at each endpoint to match up, which would seem a contradictory requirement for example if the endpoints are pulling apart. However they point out there is an analogy in the Rindler coordinates, in which the endpoints of a rigid rod have different speeds in almost all frames. Their paper outlines intriguing conceptual proposals but includes almost no calculation.


\section{Background part I --- Relativity, spaces, cosmology}
Cosmology is the subset of astronomy and astrophysics that describes the large-scale structure of the universe, including large distances and time scales. It is modeled using general relativity, which superseded Newton's theory of gravity.

\subsection{Models of space, time, and gravity}
\label{sec:models}
In \emph{Newtonian gravity}, the force $F$ between two particles of masses $m$ and $M$ is
\eqn{F=\frac{GmM}{r^2},}
where $r$ is their separation and $G$ is Newton's gravitational constant. Space and time are absolute and modelled as Euclidean $\mathbb R^3\times\mathbb R$ with Cartesian or other coordinates. Newtonian gravity is still used frequently for its computational simplicity and instructive purposes.

In relativity, space and time are no longer completely distinct, but merge into \emph{spacetime}. \emph{Special relativity} ignores gravity but describes the effects of high (relative) velocities. An object moving at speed $V$ relative to some observer has measured \emph{length-contraction} and \emph{time-dilation} by the \emph{Lorentz factor}
\eqn{\gamma(V)\equiv(1-V^2)^{-1/2}.}

\begin{SCfigure}
\centering
  \includegraphics[width=0.4\textwidth]{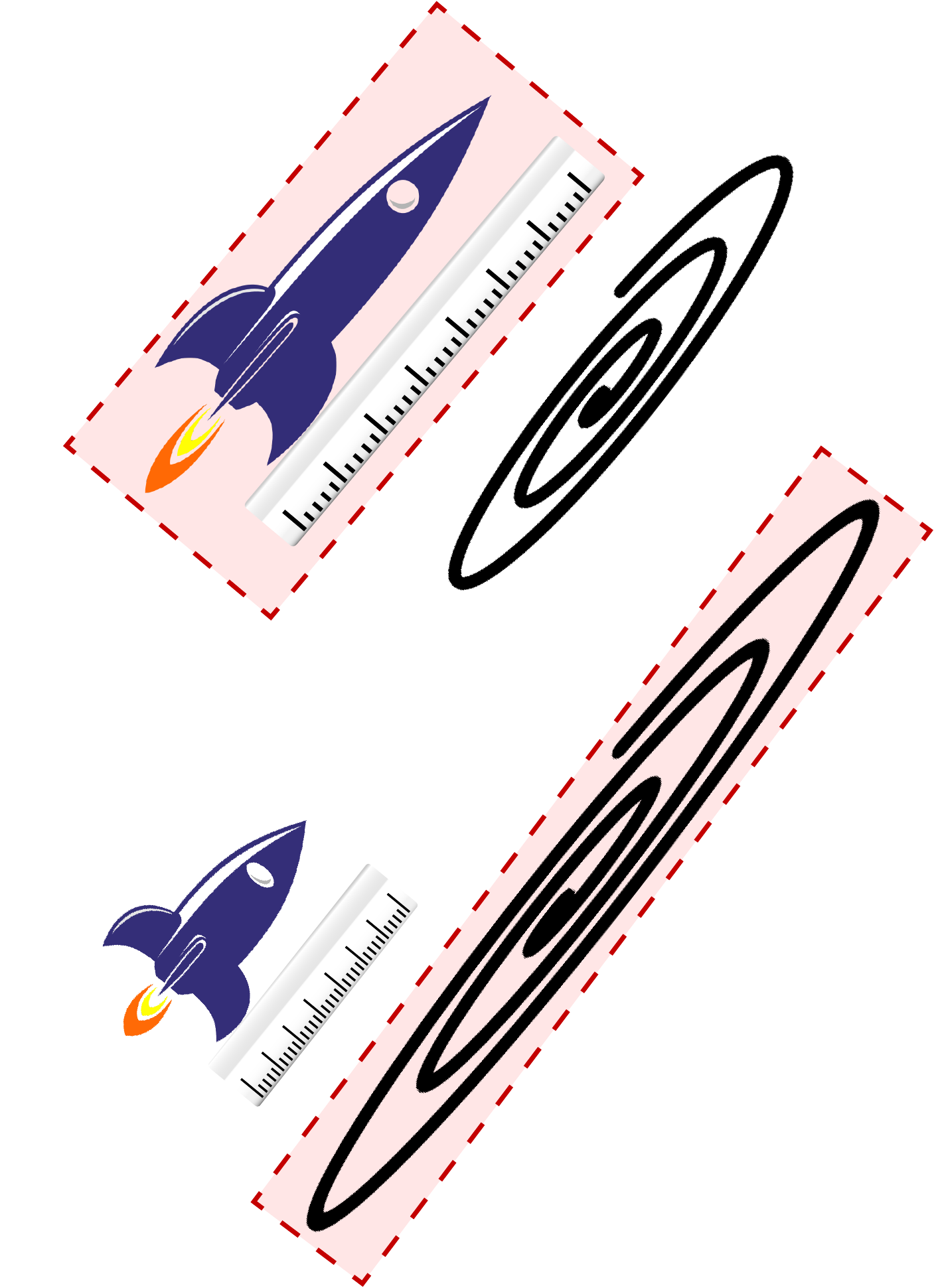}
  \caption{Length-contraction in flat spacetime. Suppose a fast rocket passes a galaxy. In the rocket's frame (top), the galaxy is length-contracted. In the galaxy frame (bottom) it is the rocket which is shortened. This is the usual line of thinking, but the following conceptual approach is equally valid. Suppose a rigid ruler is attached to the rocket, or painted on its side. Then in the rocket frame (top), the rocket ruler measures a shortened galaxy, as before. But in the galaxy frame (bottom), the rocket ruler is contracted, and so it measures a longer galaxy. The ``true'' (proper) length of the galaxy is the geometric mean of these measurements.}
\label{fig:lengthcontraction}
\end{SCfigure}

A \emph{Lorentz boost} changes coordinates between frames in standard configuration, and is the counterpart to Galilean transformations in Newtonian physics. For a coordinate system $(t,x,\ldots)$, a boost in the \direction{x} by speed $V$ is described by the new coordinates $(t',x',\ldots)$ where
\begin{align}
\label{eqn:Lorentzboosttime}t' &= \gamma\left(t-\frac{Vx}{c^2}\right)\\
x' &= \gamma\left(x-Vt\right),
\end{align}
where $c$ is the speed of light. The Lorentz velocity addition law describes how to combine collinear velocities. The \emph{rapidity} $\phi\equiv\tanh^{-1}V$ is a measure of speed for which collinear rapidities $\phi_1$ and $\phi_2$ add as simply $\phi_1+\phi_2$.

In both special and general relativity, spacetime is modelled by a \dimensional{4} geometry which describes the curvature, technically a pseudo-Riemannian \manifold{4} with quadratic form metric $g_{\mu\nu}$ \citep[p32]{hobson+2006}. A point in spacetime is an \emph{event}, described by coordinates $x^\mu=(x^0,x^1,x^2,x^3)$, and here using the sign convention $-+++$ so the first coordinate is time. The spacetime \emph{interval} $s$ is given by
\eqn{ds^2=g_{\mu\nu}dx^\mu dx^\nu.}
(For special relativity in particular, spacetime is modelled by \emph{Minkowski space} which has metric $\eta_{\mu\nu}\equiv\diag(-1,1,1,1)$.)

\emph{Geometric units} are used throughout, where $G=c=1$ and all units are measured in length or an exponent thereof. To convert from a usual unit system, keep lengths unchanged, multiply mass by $Gc^{-2}$, time by $c$, velocity by $c^{-1}$, and so on \cite[Appendix F]{wald1984}.

A \nvector $\fvec u$ has norm or magnitude
\eqn{g(\fvec u,\fvec u)=g_{\mu\nu}u^\mu u^\nu=u^\mu u_\mu,}
and is called timelike/null/spacelike if its norm is respectively negative/zero/positive. A vector is \emph{normalised} if its magnitude is $\pm 1$. For a particle with rest mass, the \velocity{} is $u^\mu\equiv\diff{x^\mu}{\tau}$ where $\tau$ is the proper time of the particle as described later, and $\fvec u$ is timelike and normalised. A photon worldline is null. A \emph{geodesic} worldline corresponds to freefall motion, and is the generalisation of a straight line in flat space.

The \emph{Einstein field equations} are described by the tensors
\eqn{G_{\mu\nu}+\Lambda g_{\mu\nu}=\frac{8\pi G}{c^4}T_{\mu\nu},}
which determine the spacetime geometry based on a given stress-energy-momentum distribution $T_{\mu\nu}$.

The \emph{covariant derivative} $\nabla\fvec u$ of a vector field $\fvec u$ is
\eqn{\nabla_\alpha u^\beta=\pdiff{u^\beta}{x^\alpha}+\Gamma^\beta_{\alpha\gamma}u^\gamma,}
in a coordinate basis. This describes how to \emph{parallel transport} or \emph{parallel propagate} vectors between tangent spaces \cite[\S20.4]{hartle2003}.

The \emph{symmetrisation} of a tensor $T$ is its symmetric part, and is denoted by parentheses around two or more indices, for instance $T_{(\alpha\beta)}$. Similarly, the antisymmetrisation is denoted by brackets around indices, such as $T_{[\alpha\beta]}$. For the rank $2$ tensor $T_{\alpha\beta}$,
\eqn{T_{(\alpha\beta)}=\frac{1}{2}(T_{\alpha\beta}+T_{\beta\alpha}), \qquad\qquad\qquad T_{[\alpha\beta]}=\frac{1}{2}(T_{\alpha\beta}-T_{\beta\alpha}).}

A \emph {Killing vector (field)} $\fvec\xi$ satisfies Killing's equation
\eqn{\nabla_{(\alpha}\xi_{\beta)}\equiv\frac{1}{2}\left(\nabla_\alpha \xi_\beta+\nabla_\beta \xi_\alpha\right)=0.}
Then for a \velocity{} $\fvec u$, $\xi_\alpha u^\alpha$ is constant along a geodesic \cite[p183]{schutz2009} \cite[\S3.8]{carroll2004}.

A \emph{photon} is a quanta of light or other electromagnetic radiation, and is represented in relativity as a point object. \emph{Redshift} $z$ is the lengthening of the wavelength $\lambda$, defined by $z=\frac{\delta\lambda}{\lambda}$, so $1+z=\frac{\lambda_{\rm observed}}{\lambda_{\rm emitted}}$. It is commonly described as originating from three causes --- Doppler/kinematic, gravitational, and cosmological. The first is due to relative motion between an emitter and receiver, the second due to a gravitational field, and the third due to the stretching of space. In special relativity, the only possible cause of redshift is the former, and for a relative speed $V$:
\eqn{\label{eqn:DopplerSR}1+z=\sqrt\frac{1+V}{1-V}.}

For a diagonal metric, the conversion of a contravariant vector to covariant form is given by
\eqn{\label{eqn:contratocov}u_\mu=g_{\mu\nu}u^\nu=g_{\mu\mu}u^\mu.}

\subsection{Schwarzschild spacetime --- black hole geometry}
The \emph{Schwarzschild geometry} describes the spacetime outside a spherically symmetric mass distribution, which is neither electrically charged nor rotating \citep[\S8]{griffithspodolsky2009}. In \emph{Schwarzschild coordinates} $(t,r,\theta,\phi)$ the metric takes form
\eqn{ds^2=-\Schw dt^2+\Schw^{-1}dr^2+r^2(d\theta^2+\sin^2\theta d\phi^2). \qquad\qquad (G=c=1)}
The coordinates resemble spherical coordinates for flat spacetime, where the angular coordinates $\theta$ and $\phi$ have the same meaning. $r$ is a radial coordinate, although not equal to the radial ``proper distance'' which is greater. $t$ is sometimes called ``far-away time'', because it is the proper time for a stationary observer at ``infinite'' $r$. Observers stationary in these coordinates are called \emph{Schwarzschild observers}. The terms ``geometry'' and ``spacetime'' emphasise the underlying invariant structure, in contrast to the arbitrariness of a given coordinate description.

The radial ``proper distance'' $dR$ is the metric $ds$ after setting $dt=d\theta=d\phi=0$:
\eqn{dR=\Schw^{-1/2}dr.}
A closed-form expression for $R$ is given in Appendix B, although it is easier to work in differential form.

A \emph{black hole} has all its matter contained within the sphere $r=2M$; this radius then becomes an \emph{event horizon}.

A \emph{raindrop} is a small particle in free fall, which initially started from rest at infinity (in the limit). These have purely radial motion.

A spinning black hole is described by the Kerr spacetime, in the electrically neutral case. These have the interesting feature of the \emph{ergosphere}, an oblate spheroid just outside the event horizon, where all matter cannot resist being swept around --- one might describe that ``space'' revolves there faster than light.

\subsection{FLRW spacetime --- universe geometry and expanding space}
The \emph{Friedmann-Lema\^itre-Robertson-Walker} (FLRW) spacetime models an homogeneous and isotropic universe \cite[\S6]{griffithspodolsky2009}. With coordinates $(t,\chi,\theta,\phi)$, the metric is
\eqn{ds^2 = -dt^2 + R(t)^2\left[d\chi^2 + S_k(\chi)^2(d\theta^2+sin^2\theta d\phi^2)\right]. \qquad\qquad (c=1)}
The $R(t)$ term is termed the \emph{scale factor}, and describes the relative expansion of the universe. $\chi$ is a radial coordinate called the \emph{comoving distance}. $t$ is called \emph{cosmic time}, and is the proper time for comoving particles as described below. Spatial homogeneity and isotropy imply constant spatial curvature, which allows only three possible geometries (up to scaling, and with certain additional requirements) --- negative/flat/positive curvature, giving spherical/flat/hyperbolic (i.e. saddle-shaped) geometries respectively, corresponding to $k=-1$, $0$, or $1$. $\chi$ is limited to $\chi\in[0,\pi]$ in the $k=1$ case, otherwise $[0,\infty)$. By definition,
\eqn{S_k(\chi)\equiv\casesifthree{\sinh\chi}{k=-1}{\chi}{k=0}{\sin\chi}{k=1}}

The ``proper distance'' $D$ is the distance to the origin, as determined by the above from metric with $dt=d\theta=d\phi=0$:
\eqn{D=R(t)\chi.}
Objects at constant $(\chi,\theta,\phi)$ are called \emph{comoving}, because they follow the general dynamics of the expansion. This dynamic is known as the \emph{Hubble flow}, and implies the relative speed between any two comoving points is proportional to the distance between them. This is the \emph{Hubble law} $v=HD$, where $v\equiv\diff{D}{t}$ is like a speed. For small redshift $z$ this can be approximated $v\approx cz$ \cite[p374]{rindler2006}. Cosmological redshift is due to the expansion of the universe, and is simply the ratio of the scale factors between emission and observation:
\eqn{1+z=\frac{R_{\rm observed}}{R_{\rm emitted}}.}

The \emph{Friedmann equations} determine the scale factor $R(t)$ and hence the expansion of the universe, over time. They follow from Einstein's field equations, and are given in one variation by:
\begin{align}
\dot\rho &= -3\frac{\dot a}{a}(\rho+\frac{p}{c^2})\\
\frac{\ddot a}{a} &= -\frac{4\pi G}{3}(\rho+\frac{3p}{c^2})+\frac{\Lambda c^2}{3}
\end{align}
Here $\rho$ is density, $p$ is pressure, and $a\equiv\frac{R(t)}{R_0}$ is the normalised scale factor, where $R_0$ is the scale factor at the present time.

Various metaphors are used to illustrate the universe and its expansion. The \emph{balloon model} depicted in figure~\ref{fig:balloonmodel} is suggestive of stationary galaxies (at grid corners for instance) and an expanding space between them (grid lines). It also illustrates the lack of a unique center. Though the strongest analogy is with a positively-curved universe in which case the FLRW metric is exactly the form for a \sphere{3} embedded in $\mathbb R^4$ --- it is analogous to all curvature cases \citep[p361, 363]{hobson+2006} \citep[p374]{rindler2006}. Another model is the \emph{rubber sheet} which illustrates the effect of matter on warping spacetime, by a heavy bowling ball which causes an indentation on the sheet and diverts the path of rolling marbles.

\begin{figure}[h]
\centering
\includegraphics[width=0.7\textwidth]{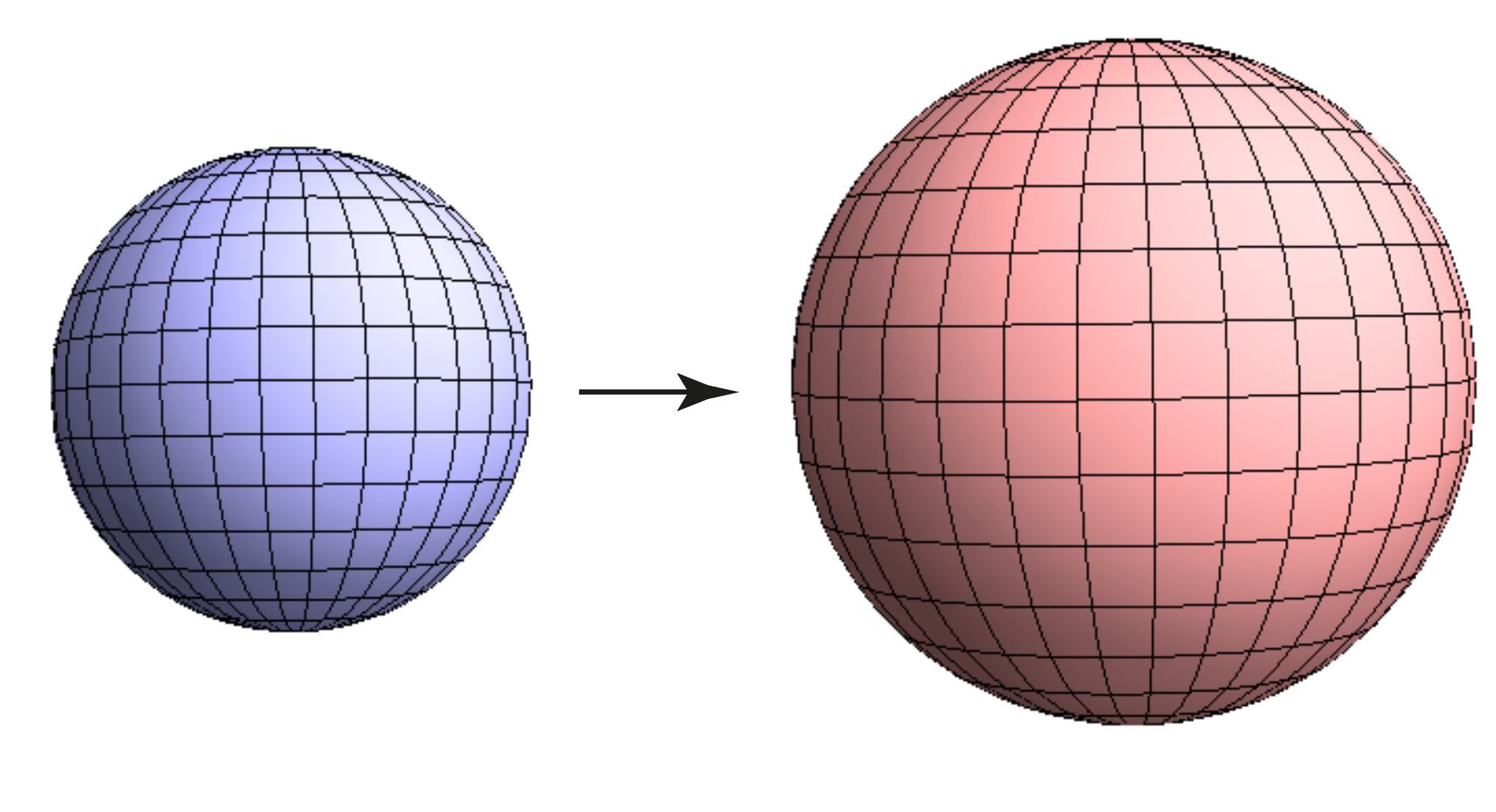}
\caption{Balloon model of the universe, which is expanding over time. Photons emitted at an earlier time have a higher frequency (blue), relative to those same photons at a later time (red), as measured by Hubble observers.}
\label{fig:balloonmodel}
\end{figure}

\subsection{Distance measures in relativity}
\label{sec:distanceintro}
There various distance measures used in cosmology could be categorised into observational and theoretical approaches. These are introduced here, and critiqued in section~\ref{sec:distancecritique} and following.

Observational distances are used by astronomers and include \emph{luminosity distance} $d_L=R_0S_k(\chi)(1+z)$, \emph{angular diameter distance} $d_A=\frac{R_0S_k(\chi)}{1+z}$, and cosmological redshift $z$ utilised as a distance indicator \citep[p345--351]{schutz2009} \citep[p371--374]{hobson+2006}. For close objects there is \emph{parallax} or \emph{proper motion} distance \citep{hogg1999}. These are equivalent in the limit $z\ll 1$ \citep[\S14.4]{weinberg1972}. For Solar System objects \emph{radar distance} is the primary measure.

It is important to connect observables with theory, and the following distances are ostensibly determined by the metric $g_{\mu\nu}$.

For a diagonal metric, one can read straight off a single metric component, for the interval in that spacetime direction. The general definition of \emph{proper distance} is: \citep[p147]{schutz2009} \citep[p11]{carroll2004}
\eqn{\Delta s=\int_P\sqrt{g_{\mu\nu}dx^\mu dx^\nu}=\int_P\sqrt{g_{\mu\nu}\diff{x^\mu}{\lambda}\diff{x^\nu}{\lambda}}d\lambda,}
where $P$ is a spacelike path parametrised by $\lambda$. For a timelike path the \emph{proper time} $\tau$ is given by the same expression with a minus: $\int\sqrt{-g_{\mu\nu}\cdots}$. For a null interval this expression gives $0$. These quantities depend not just on the endpoint events but on the path chosen. Paths which change from timelike to spacelike have no defined length \citep[p33--34]{wald1984}, but for a geodesic more specifically, its type cannot change \citep[p321]{misner+1973}.

\emph{Proper length} is the spatial distance of an object as measured in its own frame(s). It is an \emph{invariant}, at least for a perfectly rigid object, in the following sense: though other frames will in general measure a different value due to length-contraction, all agree that after compensating for their relative motion the result would be the same. Note proper length contrasts with ``proper distance'' which allow an arbitrary interval of time.

The Landau-Lifshitz \emph{radar distance} $\gamma^{\rm LL}_{ij}$ uses hypothetical light signals to determine distance \citep[\S9.2]{schutz2009}. Given a spacetime metric $g_{\mu\nu}$, suppose an observer sits at constant spatial coordinates, and emits a photon which reflects off an infinitesimally close object. The radar distance is defined as half the elapsed proper time for the observer, $\frac{1}{2}d\tau$, since speed ($c$) is distance over time, interpreted here within the particle's frame. It follows,
\eqn{\label{eqn:radardist}\gamma^{\rm LL}_{ij}=g_{ij}-\frac{g_{0i}g_{0j}}{g_{00}},}
where $i,j=1,2,3$. \citet[\S84]{landaulifshitz1971} term it ``proper distance'' and ``the metric of real space'', claiming it determines ``the geometric properties of the space''. They concede that only infinitesimal distances can be measured this way since in general the metric will change over time, and there are limits to simultaneity defined from this method.

Other distance measures include the spatial projector $P_{\mu\nu}$ and proper metric $\gamma_{ij}$ which are defined later.

\subsection{Natural and canonical}
\label{sec:naturalintro}
The terms \emph{natural} and \emph{canonical} describe a quantity or reference frame which stands out within the context of a specific situation. For instance Schwarzschild observers and Hubble comovers are natural in the sense of having various geometrical symmetries, even though physics can be done in any frame. They describe a choice which is especially simple, well-suited, or otherwise stands out within a given situation. This does not deny the principles of relativity or covariance.

Canonical quantities might be described as midway between the concepts ``relative'' and ``invariant''. Relativistic invariants include the number of particles, electric charge, speed of light, rest mass, proper time, proper distance \citep[\S5]{fayngold2010}, proper length, and proper acceleration.

The concept is imprecise; also neither the existence nor uniqueness of a natural choice should be assumed. For instance, for a rotating black hole one can list three natural families of observers \citep[\S1]{bini2014}. A situation with no obvious natural choice is the location of the winch or reel for one variant of the Schwarzschild fishing line described in section~\ref{sec:fishingline}, thus the choice lacks a certain aesthetic quality.

Coordinate systems may be singled out by having some of the following convenient or otherwise ``nice'' properties. A frame is termed \emph{adapted} to a field of observers $\fvec u$ if frame's timelike basis vector is $\fvec u$, and the spacelike basis vectors are orthogonal to it \citep[\S5]{bini2014}. In \emph{comoving coordinates} the material in question remains at constant spatial coordinates. \emph{Synchronous coordinates} have the form $ds^2=-dt^2+g_{ij}dx^idx^j$, named because the \coordinate{t} measures proper time for comoving observers and is orthogonal to the space coordinates, also the hypersurfaces of constant $t$ locally define simultaneity for comoving observers \citep[\S27.4]{misner+1973}. The existence of Killing vectors is associated with symmetries. Finally the existence of an \emph{orthogonal hypersurface} to a given flow yields a canonical measure of distance, as described in section~\ref{sec:orthogonalhypersurface}.

\subsection{Local inertial frames and the equivalence principle}
\label{sec:lif}
The property that curved spacetime is ``locally flat'' or Minkowski is expressed with mathematical precision in \emph{locally inertial frames} (LIFs) and coordinates. At a spacetime event $p$, locally inertial coordinates are defined by the following properties: \citep[p314]{misner+1973} \citep[p73--74]{carroll2004}
\begin{itemize}
\item $g_{\mu\nu}(p)=\eta_{\mu\nu}$ 
\item $\partial_\alpha g_{\mu\nu}(p)=0.$
\end{itemize}
In other words, at $p$ the metric is in canonical form and its first derivatives vanish. Hence the Christoffel symbols are $\Gamma^\alpha_{\beta\gamma}=0$ at the point, however the second derivatives are $\partial_\alpha \partial_\beta g_{\mu\nu}\ne 0$ in general, and since these determine the curvature quantities the curvature has not been removed, which is a reminder the flatness is only local \citep[p2-31 to p2-34]{taylorwheeler2000}.

The orthonormal basis vectors form a \emph{local Lorentz frame}, called a \emph{tetrad} or \emph{vierbien} because there are four of them. Though there is a unique tangent space at every spacetime event, one can choose the basis vectors differently to correspond to different velocities or frames. In general, LIFs hold only at a point, not in the neighbourhood around it. It is also possible to perform local Lorentz transformations.

More specific types of coordinate system approximate flat spacetime even more closely. \emph{Riemann normal coordinates}, known as Gaussian normal coordinates in differential geometry, are defined by geodesic segments \citep[p112--113]{carroll2004}. \emph{Fermi normal coordinates} or freely falling coordinates are closer still, and these fall and rotate with the motion of particles.

A related concept is the \emph{equivalence principle} which comes in several variants but essentially states:
\begin{quote}
[T]here is no experiment that can distinguish a uniform acceleration from a uniform gravitational field. \citep[p113]{hartle2003}
\end{quote}

\subsubsection{Relative speed}
Local inertial frames are ``unbelievably useful'', as in the following example where simply knowing an LIF exists is sufficient for a proof, without even requiring its construction \citep[p75--76]{carroll2004} \citep[p154, 199--200]{hartle2003}.

Suppose a particle with \velocity{} $\fvec u$ passes an observer with \velocity{} $\fvec u_{\rm obs}$, at the same place. Then some local inertial frame exists in which $u_{\rm obs}^\mu=(1,0,0,0)$ and $u^\mu=(\gamma,V\gamma,0,0)$ at the instant of passing, where $V$ is the ordinary relative speed between them --- the magnitude of the \velocity[3] --- and $\gamma$ is the Lorentz factor. In this frame,
\eqn{\label{eqn:relativespeed}\gamma=-\fvec u_{\rm obs}\cdot \fvec u, \qquad\qquad V=\sqrt{1-\gamma^{-2}},}
where the latter equation is listed for reference only and follows directly from the definition of $\gamma$. Since the former equation is covariant (tensorial), it must hold in any frame, whether inertial or not.

\subsection{Rindler coordinates --- the constantly accelerated rigid rod}

\begin{SCfigure}\centering
  \includegraphics[width=0.7\textwidth]{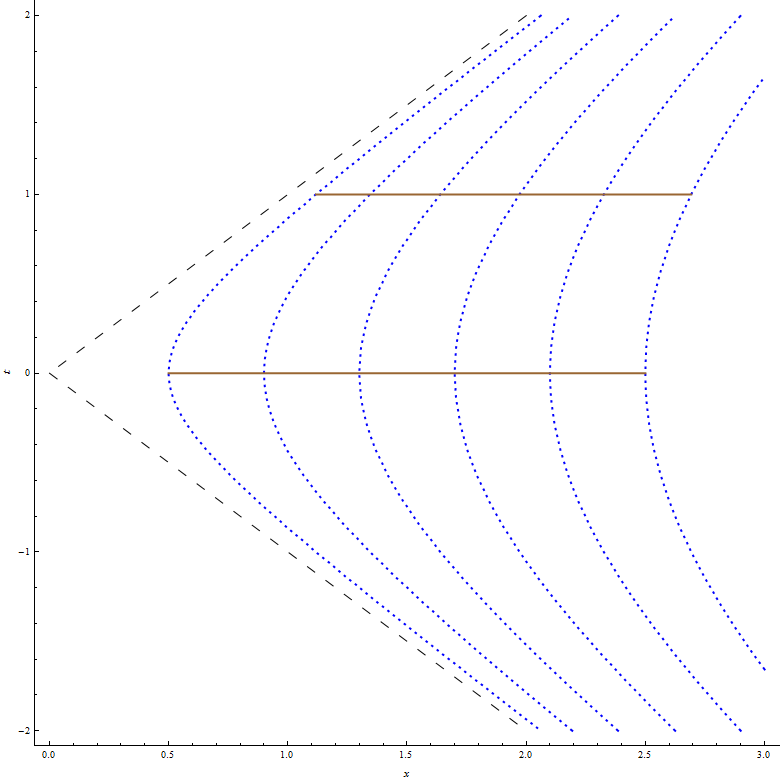}
\caption{The Rindler chart from an inertial observer's frame. The ``Rindler wedge'' is bounded by $0<x<\infty$, $-x<t<x$, and the asymptotes $t=\pm x$ are photon paths. Each hyperbola is the world line of a single particle with constant proper acceleration. For $t<0$, the rod is moving to the left, and for $t>0$ it is moving right; while for all $t$ it is accelerating to the right. Length-contraction effects are apparent in the changing length of the rod (brown line) segments, from $t=0$ to $t=1$.
A real-world rocket would appear far to the right in the Rindler chart than the pictured particles which have extreme accelerations. An acceleration of $1g$ is in geometric units $\frac{1}{X}=1g\cdot c^{-2}\approx 10^{-16}m^{-1}$. So at $t=0$, $x=X\approx 10^{16}$m, or about $1$ light-year.}
\label{fig:Rindlerchart}
\end{SCfigure}

The \emph{Rindler chart} or \emph{coordinates} describe a rigidly accelerating object in Minkowski space \cite[\S3.8, \S12.4]{rindler2006}. The neighbouring particles making up the object maintain constant distance from one another as measured locally in their instantaneous inertial frames. The system has counterintuitive properties which I use as a key pedagogical example of rigidity, to train the intuition for more general rigid motions.

Define a Minkowski frame $(t,x,y,z)$, and assume all motion and acceleration is restricted to the \direction{x}. As a first step, the simplest case is of a single particle. For constant proper acceleration $\alpha$, its worldline satisfies
\eqn{x^2-c^2t^2=\frac{c^4}{\alpha^2},}
which naturally enough is known as hyperbolic motion \citep[\S3.7]{rindler2006}. Note proper acceleration means the acceleration as measured in the instantaneous inertial frame of the particle, which is distinct from the coordinate acceleration measured by the inertial observer.

This generalises to Rindler coordinates which describe a delicately synchronised host of such particles, as in figure~\ref{fig:Rindlerchart}. Suppose the rigid object lies along the \xaxis{x}, and accelerates also along the \xaxis{x}. Each particle undergoes constant proper acceleration. The $y$\nobreakdash- and \coords{z} remain fixed for any given particle, and are suppressed in spacetime diagrams. Thus the object may be conceived of as a $1$D rod or a $3$D block.

Introduce new coordinates $X$ and $T$ by which the motion may be parametrised:
\eqn{\label{eqn:rindlercoords}t=X\sinh T, \qquad\qquad x=X\cosh T.}
This uses geometric units; $T$ is a time coordinate, and $X$ is a spatial coordinate fixed for a given particle which then has constant proper acceleration $\frac{1}{X}$. The $y$ and \coordinate{z}s are unchanged. The transformations are
\eqn{X^2=x^2-t^2,}
so when $t=0$, $x=X$ \citep[p71--73,76]{rindler2006}. Dividing the $t$ and $x$ in equation~\ref{eqn:rindlercoords} yields
\eqn{\tanh T=\frac{t}{x}.}
(This is preferable to Rindler's choice of the multiplicative inverse --- $\coth T=\frac{x}{t}$ --- because $T=\coth^{-1}\left(\frac{x}{t}\right)$ fails for $t=0$ where $\frac{x}{t}$ is undefined. But $T=\tanh^{-1}\left(\frac{t}{x}\right)$ is valid for all events because $x\ne 0$ for all particles, and $\tanh^{-1}(0)=0$ is well-defined even when $t=0$.)

\begin{wrapfigure}{l}{0.5\textwidth}
\includegraphics[width=0.48\textwidth]{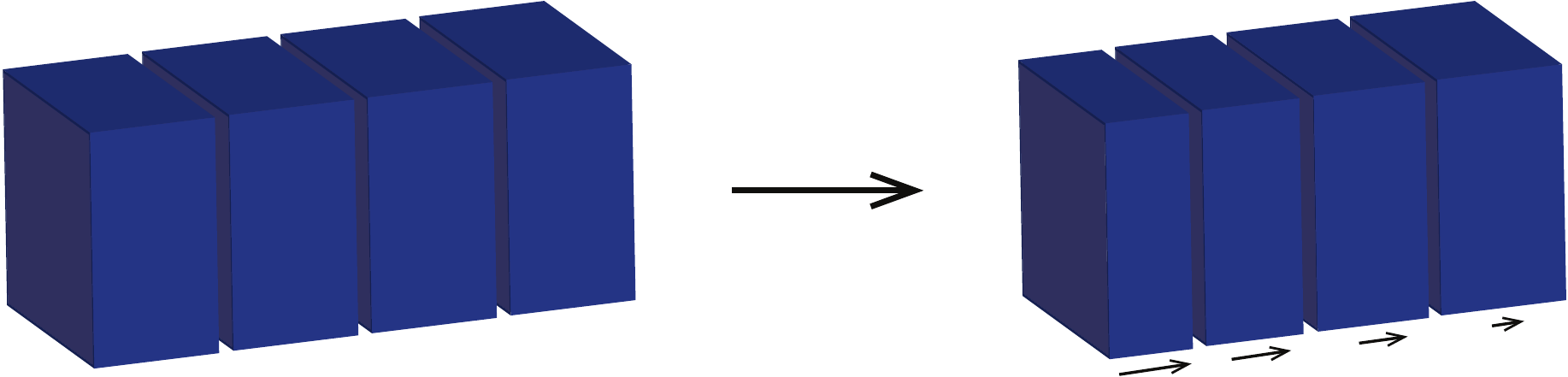}
\caption{A more intuitive picture of the Rindler system as a $3$D block. At $t=0$ the connected segments are of equal length, because the block is stationary (left diagram). After accelerating to the right, the trailing segments are moving the fastest and are thus the most length-contracted as measured by the inertial observer (right diagram). Other optical effects of high speed motion are not depicted. But from the perspective of the block itself, the left diagram applies for all $t$, because an instantaneous inertial frame for one particle is an inertial frame for the whole block, because in this simultaneity convention all particles have the same speed at any given instant.}
\label{fig:Rindlerblock}
\end{wrapfigure}

Note the following properties. For $t>0$ the rod is continually shortening, as measured by the observer, as it picks up speed. Hence, the trailing (more leftwards) particles must move faster than the leading (more rightwards) particles, in order to compensate for the steadily increasing length-contraction. The rod has limited extent in the negative \direction{x}, because the required acceleration increases without limit as the wedge $t=\pm x$ is approach. Thus it could be said the rod ``ends in a photon'' \citep[p51]{rindler1977}. This horizon has analogies with the event horizon of a black hole. In the positive \direction{x} by contrast, there are no limits on the rod's length because the required accelerations diminish.

The coordinate speed of a particle is $v=\diff{x}{t}=\frac{t}{x}$, noting the reversing of $t$ and $x$ \citep[p76]{rindler2006}. In $(t,x)$ coordinates (suppressing $y$ and $z$ which are shared by all coordinate systems used here), the \position{} of a particle may also be expressed
\eqn{\label{eqn:rindler4pos}x^\mu = \left(X\sinh T,X\cosh T,y,z\right).}

In $(T,X)$ coordinates the metric is: \citep[p268]{rindler2006}
\eqn{\label{eqn:rindlerTXmetric}ds^2 = -X^2dT^2+dX^2+dy^2+dz^2.}

\subsubsection{Bell's spaceship paradox}
\emph{Bell's spaceship paradox} is conceptually similar to the Rindler system. It involves a pair of spaceships with a string tied between them, which are initially stationary relative to an inertial observer with coordinates $(t,x,y,z)$. At a given time $t$ both rockets being accelerating in the \direction{x} with constant acceleration. The question is, will the string break?

In the observer's frame, the rockets maintain constant distance, inviting the conclusion that the string will not break. However, due to length-contraction this constant coordinate distance is actually a stretching. Alternatively, from the instantaneous frame of either rocket, the motion is not simultaneous, and the other rocket recedes. Hence the string will break when it exceeds its elastic limit, as deduced by all frames of reference.

This setup demonstrates that constant acceleration by collinear particles does \emph{not} form a rigid system, which is a complementary statement to the Rindler situation. John Bell, who popularised the scenario, interpreted it as demonstrating that length-contraction is ``real'', not just apparent. See \citet[\S5]{natario2014} for spacetime diagrams for versions of the paradox where waves propagate on an elastic string.

\section{Background part II --- Rigid bodies and the proper metric}
A \emph{rigid} body is one which maintains its shape over time, so its constituent particles maintain constant distance from one-another, at least in the strictest definition of rigidity. This is the same property (constant distance) required of a chain of observers to interpret a redshift as gravitational. Thus the kinematics of rigid bodies and gravitational redshift observer families are identical. These distances are measured in the proper frame.

\subsection{Born-rigidity --- the ideal}
\label{sec:bornrigid}
Rigidity is an under-developed concept within relativity, for various reasons. Most astronomical objects are well modelled by fluid or point particles, although exceptions include neutron star crusts and bar detectors for gravitational waves \citep[\S1]{beigschmidt2003}. Rigidity has been more fully developed in continuum mechanics than relativity \citep[p271, 277]{maugin2013} \citep[p1]{beig2004} \citep[\S1]{brotasfernandes2003}, although interest in the measurement process did surge since the 1990s \citep[\S9]{bini2014}. Also the restrictive nature of Born-rigidity, as described next, discouraged further research.

\citet{born1909} defined a rigid object as one with constant distance between the constituent particles, as measured locally in the instantaneous proper-frame of either, in the context of special relativity. Though a straightforward choice of definition, Herglotz and Noether both showed in 1910 that such a body has only three degrees of freedom --- spatial translations. The motion of a single particle determines the motion of the entire body. (There is additionally a class of allowed motions \citet[p26--27]{dewitt2011} calls ``superhelical'' which have no degrees of freedom at all \citep{guilini2010}.) In classical mechanics a rigid body has six degrees of freedom, consisting of rotations as well as translations. But in relativity, any change in rotation speed would cause length-contraction in the angular direction of rotation, and thus is not allowed for Born-rigid motion. This is the source of the ``Ehrenfest paradox'' concerning motion of a rotating disc.

Hypothetically, if one were to bump a Born-rigid object, the motion would propagate instantaneously throughout the body in order to maintain constant distance between all particles. But this would be an impossible contradiction with relativity because no signals may travel faster than light. Furthermore, due to the relativity of simultaneity, this ``instantaneous'' propagation would be backwards in time in some frames, so the ``effect'' (the motion) would precede the ''cause'' (the bump) \citep{lyle2010}. This would appear to be the death knell for Born-rigidity, but the resolution to this dilemma is simply that a bumped object will not display Born-rigid dynamics. Instead, an elastic theory will apply, for instance one in which disturbances propagate at $c$.

Applied correctly, Born-rigidity is indeed logically-consistent \cite[\S1]{eriksen+1982}. A body cannot be \emph{intrinsically} or passively Born-rigid, in the sense that pulling on one point would not result in Born-rigid motion. Nonetheless it is entirely consistent for Born-rigidity to ``emerge'' from a system, for instance where a force is applied on every particle individually as in an electric field, or with a fleet of rockets on pre-programmed flight paths. It is ``[o]nly by prearrangement'' \citep[p119]{taylorwheeler1992}; or, there are no ``rigid bodies'' but only ``rigid motion'' \citep[\S3.3]{guilini2010}.

It is helpful to remember that even ``in pre-relativistic mechanics, rigidity was always an ideal concept, at best a convenient approximation that no one would really have expected to be possible'' \citep[\S1.8]{lyle2010}. Rigid coordinate systems are also useful, as \citet[p2]{brotas2006} describes of the black hole fishing line, ``[w]e cannot use an undeformable line to angle a fish but we may use it to define a coordinate system.'' In summary, Born-rigididy is useful as:
\begin{itemize}
\item a base reference for elasticity, strain, etc.
\item an approximation for more realistic deformations
\item rigid coordinate systems
\item a new distance measure, ``rigidity distance'', described later.
\end{itemize}

\subsection{Spatial projector tensor}
\label{sec:projector}

The \emph{spatial projection} (or \emph{projector}) \emph{tensor} $P_{\alpha\beta}$ describes spatial distance as measured in the local proper frame of an observer. For a given timelike vector field $\fvec u$ corresponding to observer motion, it is defined by $\fvec P \equiv g+\fvec u\otimes\fvec u$, which in covariant coordinate form is: \citep[Appendix F]{carroll2004} 
\eqn{\label{eqn:projectorcoord}P_{\alpha\beta} = g_{\alpha\beta}+u_\alpha u_\beta.}
The projector maps tensors into the \emph{local rest space} of the motion, which is the set of vectors orthogonal to $\fvec u$ --- a subspace of the tangent space at each event \citep[p36]{guilini2010}. In other words, it gives a splitting of the tangent space \citep[\S1]{bini2014}. The distance is the length of the geodesic striking orthogonally to the velocity \cite[\S9]{defeliceclarke1990}. Note this orthogonal subspace is canonically determined for a given flow $\fvec u$ --- see also section~\ref{sec:orthogonalhypersurface}.

$P$ acts on a vector $\fvec v$ by contracting with it: \citep[p79]{hawkingellis1973}
\eqn{P^\alpha_\beta v^\beta.}
$P$ is also called the \emph{spatial metric} on the orthogonal subspace \citep[p217]{wald1984} \citep[p61]{carterquintana1972}, since
\eqn{P_{\alpha\beta}v^\alpha w^\beta=g_{\alpha\beta}v^\alpha w^\beta \qquad\qquad {\rm for}\enskip \fvec v,\fvec w \perp \fvec u.}

The projector has properties including: \citep[p272]{maugin2013} \citep[p23]{dewitt2011} 
\begin{itemize}
\item related expressions for contravariant and covariant
\item $P_{\alpha\beta}=P_{\beta\alpha}$ (symmetry)
\item $P^\alpha_\beta u_\alpha=0$ (an orthogonality property)
\item $P^\alpha_\beta P^\beta_\gamma=P^\alpha_\gamma$ (idempotence, meaning applying it a second time has no additional effect, that is $P^2=P$)
\item it can be used to raise or lower indices of purely spatial quantities.
\end{itemize}

Historically, Catt\`aneo was the first to give a ``systematic study'' according to \citep[p556]{massa1974}, although \citet[p321]{cattaneo1958} himself states the term ``space norm'' was already well known. The projector tensor is a standard quantity, found in any advanced book on general relativity. However, I will often work with a little-known adaptation called the \emph{proper metric}, a Lagrangian version which follows the motion of particles, and is described shortly.

Note $P$ actually gives a longer distance than the metric, at least in the forthcoming situations. This is counter-intuitive when used to thinking of length-\emph{contraction}. But it follows directly that if a ruler is contracted, then it will measure a \emph{greater} length, as illustrated in figure~\ref{fig:lengthcontraction}. Compare \citet[\S14]{gron2004} who illustrates the same idea in the context of the rotating disc.

\subsubsection{Example: Projector in the Rindler chart}

The \velocity{} of particles may be determined as follows (I have done this independently of any known result). The proper time is obtained from the metric, which is straightforward in $(T,X)$ coordinates; from equation~\ref{eqn:rindlerTXmetric},
\eqn{d\tau^2 \equiv -ds^2 = X^2dT^2,}
since $dX=dy=dz=0$. So $d\tau=XdT$, noting $X$ is always positive and that the time increments must have the same sign. (Alternatively, this may be obtained from normalisation of the \velocity{} $u^\alpha u_\alpha=-1$). But $X$ is just a constant for a given particle, so after integrating,
\eqn{\tau=XT,}
upon choosing the initial condition $\tau=0$ when $T=0$. (So $T$ is not a synchronous coordinate, but it turns out to be the \emph{rapidity} $\phi$.)

\begin{SCfigure}
  \includegraphics[width=0.6\textwidth]{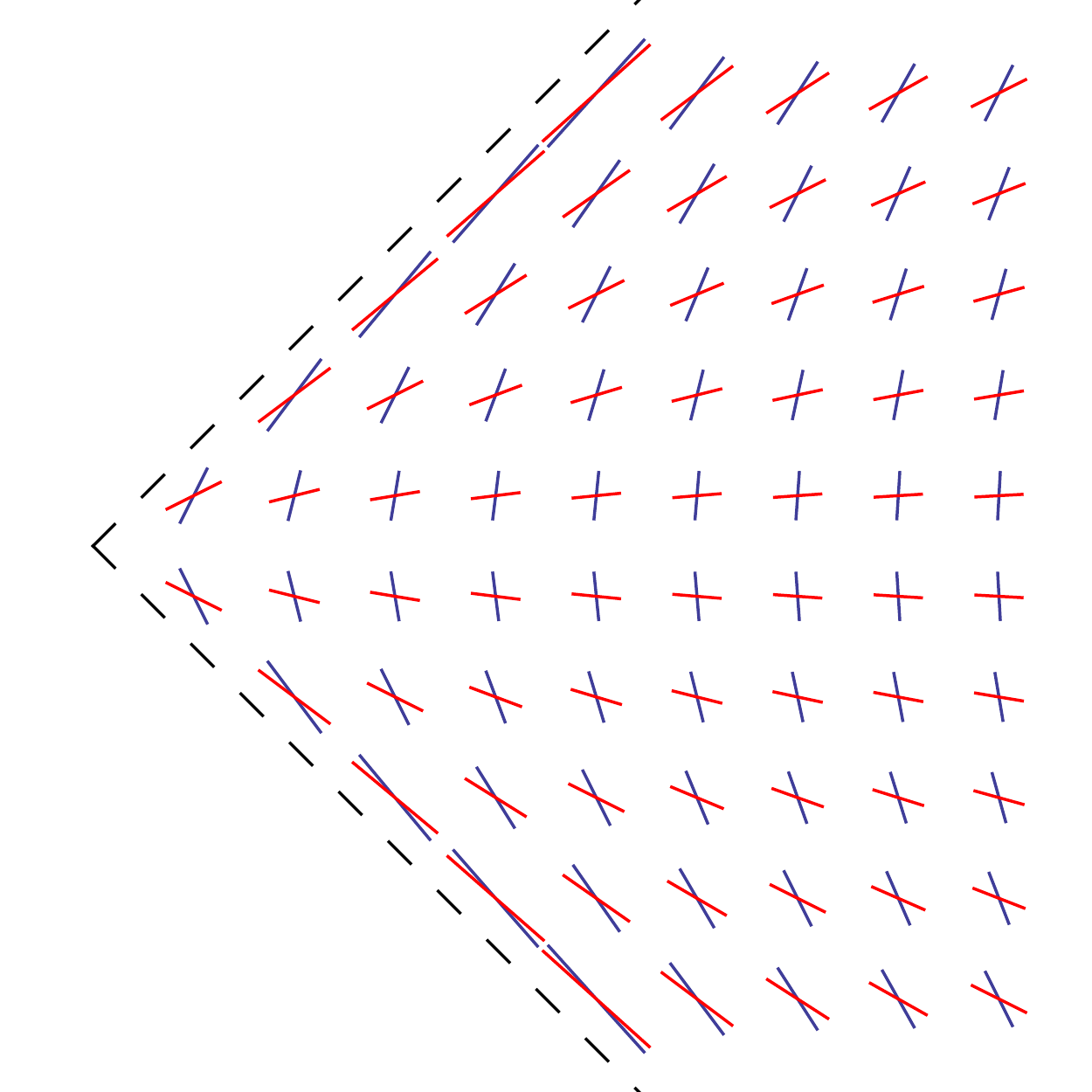}
  \caption{Motion and the orthogonal subspaces for the Rindler chart. Blue vectors represent the velocity $\fvec u$. Red vectors represent the space orthogonal to the motion. These spaces are $3$-dimensional, but the $y$ and $z$-components are suppressed. Recall $4$-orthogonality in Minkowski space does not generally correspond to $90^\circ$ angles on a Minkowski diagram. The flow lines of the orthogonal vectors integrate to form radial lines from the origin of constant $T$, corresponding to the unique instantaneous hyperplanes of simultaneity for the particles. See also section~\ref{sec:orthogonalhypersurface}.}
\label{fig:Rindlerorthogonal}
\end{SCfigure}

Working now in the original coordinates $(t,x,y,z)$, the \position{} is, from equation~\ref{eqn:rindler4pos},
\eqn{x^\mu = (t,x,y,z) = \left(X\sinh\left(\frac{\tau}{X}\right), X\cosh\left(\frac{\tau}{X}\right), y, z\right).}
Thus the \velocity{} is
\eqn{\label{eqn:rindler4velocity}u^\mu=\pdiff{x^\mu}{\tau} = \left(\cosh\left(\frac{\tau}{X}\right), \sinh\left(\frac{\tau}{X}\right), 0, 0\right).}
Clearly $u^\mu u_\mu=-1$ as required. The covariant form follows from equation~\ref{eqn:contratocov}:
\eqn{u_\mu=\left(-\cosh\left(\frac{\tau}{X}\right),\sinh\left(\frac{\tau}{X}\right),0,0\right).}

Thus
\eqn{u_\mu u_\nu = \begin{pmatrix}
\cosh^2\left(\frac{\tau}{X}\right) & -\sinh\left(\frac{\tau}{X}\right)\cosh\left(\frac{\tau}{X}\right) & 0 & 0 \\
-\sinh\left(\frac{\tau}{X}\right)\cosh\left(\frac{\tau}{X}\right) & \sinh^2\left(\frac{\tau}{X}\right) & 0 & 0 \\
0 & 0 & 0 & 0 \\
0 & 0 & 0 & 0 \\
\end{pmatrix}}

Thus, after simplifying using the identities $\cosh^2(x)-\sinh^2(x)=1$ and $\sinh(x)\cosh(x)=\frac{1}{2}\sinh(2x)$,
\eqn{P_{\mu\nu} = g_{\mu\nu}+u_\mu u_\nu = \begin{pmatrix}
\sinh^2\left(\frac{\tau}{X}\right) & -\frac{1}{2}\sinh\left(\frac{2\tau}{X}\right) & 0 & 0 \\
-\frac{1}{2}\sinh\left(\frac{2\tau}{X}\right) & \cosh^2\left(\frac{\tau}{X}\right) & 0 & 0 \\
0 & 0 & 1 & 0 \\
0 & 0 & 0 & 1 \\
\end{pmatrix}.}
This is the spatial metric for the Rindler particles; it is the background spacetime metric $g$ as distorted by the motion $\fvec u$. But the situation is much clearer in the proper metric.

\subsection{Material manifold}
The proper metric requires the tracking of individual particles, which is achieved by assigning a permanent label to each particle. These form a coordinate system $X^i$, $i=1,2,3$, for a \dimensional{3} \emph{material manifold} or reference configuration $\mathcal M^3$, corresponding to snapshots in time. Then the material coordinates along with the proper time of particles can be used to parametrise the motion, via $\fvec x=\fvec x(\vec X,\tau)$. At times I will bundle $\tau$ together with the $X^i$, defining $X^0\equiv\tau$ and using a Greek index $X^\mu$ to refer to the collection; at other times I use just the spatial coordinates with a Latin index $X^i$.

Generally, in classical fluid dynamics this approach of following the motion is called \emph{Lagrangian}, as opposed to an Eulerian approach which considers a fixed spatial location.

\subsubsection{Example: Material coordinates for the Rindler system}
Conveniently, permanent labels for the Rindler particles are supplied already by the $(T,X)$ coordinates, where the labels $(X,y,z)$ parametrise the particles. Also $T=\frac{\tau}{X}$, so the motion is parametrised by material labels and proper time as sought.

\subsection{Proper metric}
\label{sec:propermetric}
The \emph{proper metric} describes the local spatial distance between moving particles, as measured in their local proper frames. It is formed by adapting the spatial projector to the motion, a Lagrangian approach which tracks the individual particles. It is defined by \cite[\S2]{dewitt2011}
\eqn{\gamma_{ij}\equiv P_{\mu\nu}\pdiff{x^\mu}{X^i}\pdiff{x^\nu}{X^j},}
where $P$ is the spatial projector, and the $X^i$ are material coordinates. $\gamma$ is a $3\times 3$ tensor where $i,j=1,2,3$. (It could also be defined as a $4\times 4$ tensor, but the terms $\gamma_{0\alpha}$ vanish.) I point out --- and haven't seen this point made in the literature --- that this is simply the tensor transformation formula for $P_{\mu\nu}$ under the coordinate change $x^\mu \rightarrow X^\nu$. In other words, it is the space norm expressed in the material \coordinate{X}s. This is one approach to showing $\gamma$ is a tensor, because it is a coordinate transformation of the tensor $P$, which in turn is a tensor because it is a sum and direct product of the tensors $g_{\mu\nu}$ and $u^\mu$. 

The proper metric is used like the usual spacetime metric $g$: two neighbouring particles $\vec X$ and $\vec X+d\vec X$ perceive, in their instantaneous rest frames, a separation we could denote $dL$ given by
\eqn{dL^2=\gamma_{ij}dX^i dX^j.}

The tensor is described, at least for a related ``inverse-motion'' approach, as
\begin{quote}
A natural relativistic measure of deformation... built from the inverse motion gradient... It is the canonical projection of $P^{\alpha\beta}$ onto $M^3$ by the motion. \citep[p274]{maugin2013}
\end{quote}

Though the above formula is valid for any motion, the specific case of rigid motion occurs when the comoving proper lengths remain constant, that is, \cite[\S2]{dewitt2011}
\eqn{\pdiff{\gamma_{ij}}{\tau}=0.}

Historically, the quantity was originally due to Souriau in 1958, but see his subsequent works which are more readily accessible. \citet[p280]{maugin1971} discovered it independently, and calls the $\pdiff{x^\mu}{X^i}$ terms the \emph{direct gradient} of the motion. He gives a helpful historical review within the context of continuum mechanics, and mentions analogues in various classical elasticity and strain tensors. \citet[\S2]{dewitt2011} gives the best and most thorough overview, and only he uses the term ``proper metric'', at least amongst the authors cited here. DeWitt's work is recapitulated by \cite{lyle2010,lyle2014}, whose own contributions are not insignificant. See also \citet[p284,288]{ferraresebini2008}.

The proper metric is surprisingly obscure given how important distance measures are in relativity. But note the procedure of generating a new metric from $g$ is not unfamiliar, having precedents in the spatial projector and the Landau-Lifshitz radar metric, thus casting doubt on the all-sufficiency of $g$. As for its legitimacy, see the derivation given by DeWitt, and note it works well for the Rindler and other rigid systems analysed, also it reduces to $g$ in the case of stationary observers in a diagonal metric. \citet[p275]{maugin2013} argues for the insufficiency of $P_{\mu\nu}$ in some situations, because a correction term is required in order to define a certain deformation tensor. I concur for the specific case of rigidity, which is not clear from the projector alone except for certain highly symmetric situations, but is beautifully clear in the proper metric.

\subsubsection{Example: proper metric for the Rindler chart}

\begin{wrapfigure}{r}{0.4\textwidth}
  \centering
    \includegraphics[width=0.38\textwidth]{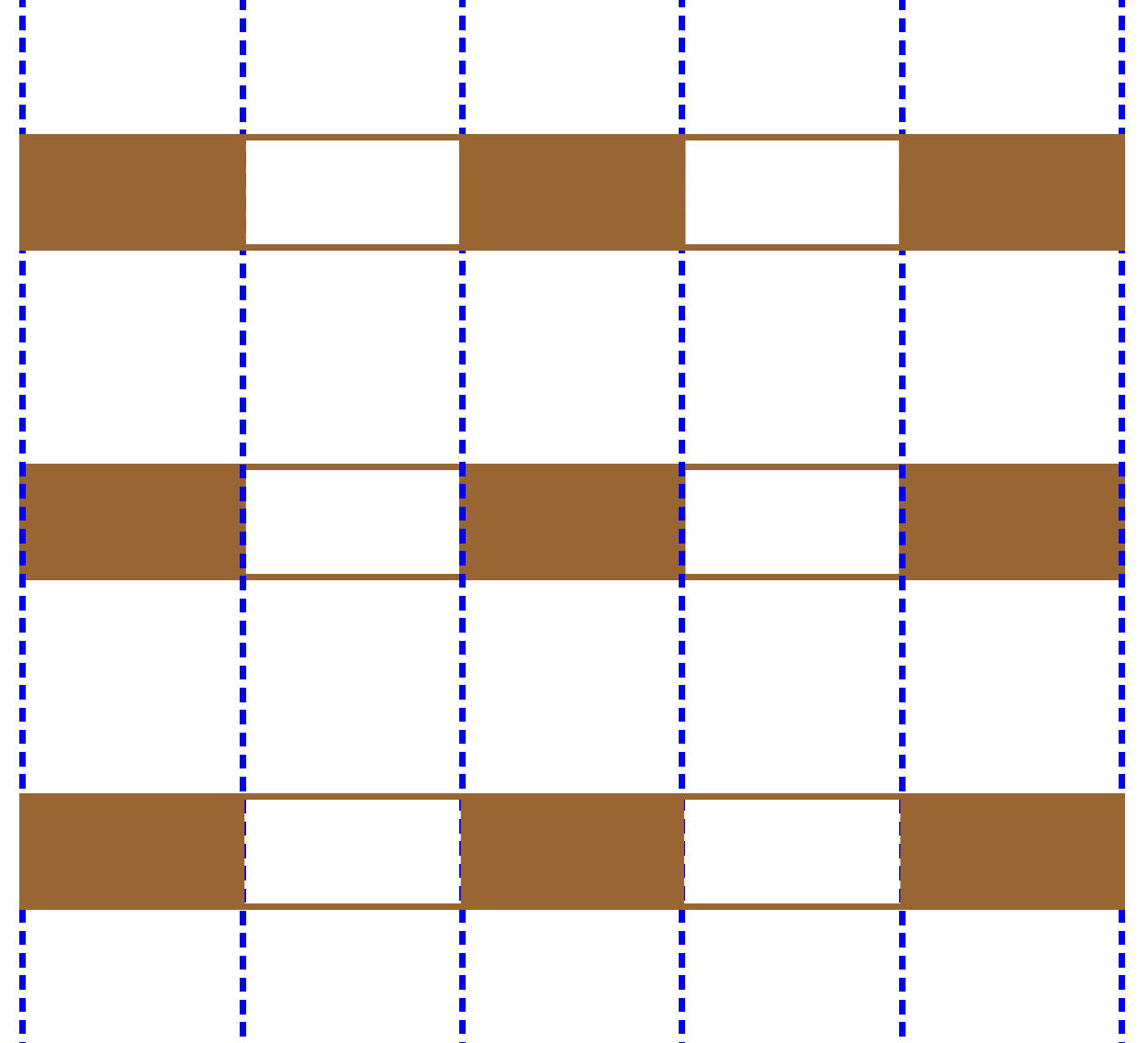}
  \caption{The Rindler rod as measured in the proper frames of the particles themselves. Here, the rod has been painted in segments for clarity. Over time, it remains unchanged. This shows the rod stays Born-rigid. Compare figure~\ref{fig:Rindlerchart}.}
\end{wrapfigure}

Finally, the proper metric for the Rindler system is given as follows, which I have done independently of any known result.
\eqn{\left(\diff{x^\alpha}{X^i}\right)=\begin{pmatrix}
 \cosh \left(\frac{\tau }{X}\right) & \sinh \left(\frac{\tau }{X}\right)-\frac{\tau  \cosh \left(\frac{\tau }{X}\right)}{X} & 0 & 0 \\
 \sinh \left(\frac{\tau }{X}\right) & \cosh \left(\frac{\tau }{X}\right)-\frac{\tau  \sinh \left(\frac{\tau }{X}\right)}{X} & 0 & 0 \\
 0 & 0 & 1 & 0 \\
 0 & 0 & 0 & 1 \\
\end{pmatrix}.}
But then the proper metric is amazingly simple:
\eqn{\gamma_{ij}=\begin{pmatrix}
 1 & 0 & 0 \\
 0 & 1 & 0 \\
 0 & 0 & 1 \\
\end{pmatrix}}
This gives the proper-frame local distances over time, as measured in the material coordinates $(X,y,z)$.  Because it is independent of $\tau$, the system is rigid. Compare \citet[\S2.4]{dewitt2011}. Writing out the bilinear form, $\gamma_{ij}dX^idX^j=dX^2+dy^2+dz^2$, so in the material coordinates the proper-frame spatial distance is simply Euclidean.

\subsection{Additional rigidity measures}
There are numerous other tensors and equations which may appear to determine rigidity, which are briefly mentioned here for completeness. The equation of continuity is \citep[p379]{hobson+2006}
\eqn{\nabla_\mu T^{\mu\nu}=0,}
however conservation of stress-energy is a much looser condition than rigidity.

The \emph{expansion tensor} describes the diverging of particle motion: \citep[p82--83]{hawkingellis1973} \citep[\S4]{bini2014} \citep[p64--65]{carterquintana1972}
\eqn{\theta_{\alpha\beta} \equiv P^\gamma_\alpha P^\delta_\beta \nabla_{(\gamma}u_{\delta)} = \frac{1}{2}\left[\mathcal L_u P\right]_{\alpha\beta},}
where parentheses refer to symmetrisation of the indices, and $\mathcal L$ is the Lie derivative. Its trace is the \emph{expansion scalar} $\theta$, which describes the volume change of a small sphere comoving with the fluid:
\eqn{\theta \equiv P^{\alpha\beta}\theta_{\alpha\beta} = P^{\alpha\beta}u_{\alpha;\beta}=u^\alpha_{;\alpha}.}
This scalar appears in the ``kinematic decomposition'' of the motion and in the Raychaudhuri equation, for $\fvec u$. However these quantities may need to be adapted to a Lagrangian approach of tracking the motion.

\emph{Killing's theorem} in geometry states rigid motion is described locally by
\eqn{\partial_\alpha u_\beta+\partial_\beta u_\alpha=0.}
However this only holds in special relativity \citep[p271]{maugin2013}, and needs to be spatially projected \citep[p37]{guilini2010}.

Other tensors include the \emph{rate of strain tensor} \citep[p92--93]{lyle2010} \citep[p23]{dewitt2011}, also $E_{\alpha\beta}$ ``a kind of relativistic `Eulerian' tensor of deformation'' \citep[p275]{maugin2013}, and canonical decompositions such as the \emph{Cauchy stress} \citep[p272]{maugin2013} and the \emph{electrogravitic} or \emph{tidal tensor} $E$.

Further brief bibliographic overviews of relativistic rigidity appear in \citet[\S15]{maugin2013}, \citet[\S1]{beigschmidt2003}, \citet[\S1,\S3]{wernig-pichler2006}, \citet[\S1]{carterquintana1972}, and \citet[\S0]{natario2014}.

\section{Flat space analogues}
In this section I point out parallels between the general relativistic FLRW model and two flat spacetime models --- the Milne model and Newtonian cosmology. For the former a coordinate transformation shows the ``empty universe'' is simply Minkowski space. In the latter, the Friedmann equations are reproduced from Newtonian gravity. In both flat spaces the expansion is normally interpreted as kinematic, not as a stretching of space. Overall, the the results in this section are not original to myself.

\subsection{Milne model --- the empty universe is flat}
The \emph{kinematic model} is a special-relativistic cosmology which explains the expansion of the universe and hence redshifts as due to motion rather than ``stretching of space''. It was proposed by Milne in 1932, whose work was passionately debated until consensus settled on general-relativistic cosmology \citep{milne1935} \citep{gale2014}. Though not a viable model of our universe, it remains useful as a comparison model \cite[\S16.3]{rindler2006} \cite[\S4]{davis2004}.

Define coordinates $(t,r,\theta,\phi)$ for Minkowski space where the spatial part is the usual spherical coordinates. The Milne model describes an ``explosion'' at the spatial origin at $t=0$, from which particles expand outwards at all possible speeds. For any given particle the speed $v\equiv\frac{r}{t}$ remains constant because the model assumes no gravity, and $\theta$ and $\phi$ also remain fixed.

This model is more robust than it seems. It appears to possess a unique centre and boundary, which would violate the \emph{cosmological principle}. (The spherical boundary is made up of the null worldlines $r=t$, and is not achieved by any particle.) However we can Lorentz boost to the global inertial frame of any given particle, which would then measure these same properties of itself --- of being stationary in these coordinates and at the site of the explosion, with the boundary receding at $r=t$ relative to it, and infinite matter in all directions.

Yet the model is not \emph{maximal}, having ``more spacetime than substratum'' \citep[p361]{rindler2006}. There is an asymmetry in causality, in that a source which is outside the boundary at a given time could have earlier emitted a photon towards the expanding ball, thus potentially influencing the explosion particles but not being influenced by them before this time. Also matter near the boundary would require literally perfect fine-tuning.

The sense of isotropy and homogeneity according to proper-frame measurements is made precise by the following coordinate transformation. Note that locally, in the instantaneous inertial frame of a given particle, the perceived separation from its neighbours is given by the rapidity $\chi$ rather than the \coordinate{r}. We have $\tanh\chi=v=\frac{r}{t}$. For a new time coordinate, the natural choice is the proper time $\tau$ of particles, which is given by
\eqn{\diff{t}{\tau}=\gamma=(1-v^2)^{-1/2}=(1-\tanh^2\chi)^{-1/2}=(\sech^2\chi)^{-1/2}=\cosh\chi,}
so $t=\tau\cosh\chi$, taking $\tau=0$ when $t=0$. Also
\eqn{r=tv=t\tanh\chi=\tau\cosh\chi\tanh\chi=\tau\sinh\chi,}
so in summary,
\eqn{t=\tau\cosh\chi, \qquad\qquad r=\tau\sinh\chi.}
The $\theta$ and $\phi$ coordinates remain unchanged. The existing metric is
\eqn{ds^2=-dt^2+dr^2+r^2d\Omega^2,}
writing $d\Omega^2\equiv d\theta^2+\sin^2\theta d\phi^2$ for the metric of the usual unit \sphere{2}. It follows
\begin{align}
dr &= \pdiff{r}{\tau}d\tau+\pdiff{r}{\chi}d\chi=\sinh\chi d\tau+\tau\cosh\chi d\chi,\\
dt &= \pdiff{t}{\tau}d\tau+\pdiff{t}{\chi}d\chi=\cosh\chi d\tau+\tau\sinh\chi d\chi.
\end{align}
After substituting into the metric and simplifying using $\cosh^2-\sinh^2=1$,
\eqn{ds^2=-d\tau^2+\tau^2(d\chi^2+\tau^2\sinh^2\chi d\Omega^2),}
so the synchronous comoving coordinates are an FLRW metric! The cosmic time is $\tau$, the scale factor $R(\tau)=\tau$ is linear, and the curvature $k=-1$ is negative as seen from the $\sinh$ term, which is expected because this eternally expanding universe has surplus energy above the critical density. The sparse matter density, which is the limit $\rho\rightarrow 0$ or ``empty universe'' could model a late stage of our universe, were it not for the requirement of cosmological constant $\Lambda=0$ which is a fatal contradiction with current observation \citep[p362, 399]{rindler2006}. The Milne model also has analogues in de Sitter and anti-de Sitter space \citep[p401]{rindler2006}.

See \citet[p63]{davis2004} for spacetime diagrams comparing the Minkowski and FLRW representations. The coordinate velocities differ between the models, as do the ``Hubble law''s based on them:
\eqn{v_{\rm Milne}=\frac{r}{t}=\tanh\chi, \qquad\qquad\qquad H_{\rm Milne}=\frac{v}{r}=\frac{1}{t},}
where the latter formula follows from $v=H_{\rm Milne}D$, where the distance is $D=r$ in this case. In FLRW coordinates, 
\eqn{v_{\rm FLRW}=H_{\rm FLRW}D=\dot R\chi=\frac{1}{\chi}, \qquad\qquad\qquad H_{\rm FLRW}=\frac{\dot R}{R}=\frac{1}{\tau},}
where the former follows from the latter.

In fact the Milne velocity even fits the special-relativistic Doppler shift formula
\eqn{1+z=\sqrt\frac{1+v_{\rm Milne}}{1-v_{\rm Milne}},}
however note this is \emph{not} the usual Hubble flow velocity \citep[p64--65]{davis2004}. \citet{chodorowski2005} shows the angular diameter distance is the same in both models, as is the luminosity distance.

One may wonder how the flat Milne model could transform into a hyperbolic FLRW model, when the coordinate transformation (diffeomorphism) implies the geometries should be the same! Indeed the Milne model has flat spacetime. However in this FLRW model it is only the \emph{spatial} slices which are hyperbolic, and this slicing is based on the coordinate representation. In fact the \emph{spacetime} of this particular FLRW model is indeed flat. In general, an FLRW geometry has only three distinct Riemann curvature tensors --- $0$, $-\frac{\ddot a}{a}$, and $\frac{k+\dot a^2}{a^2}$ --- at least in one coordinate system \citep[p547--548]{hartle2003}. In the present case $k=-1$ and $\dot a=1$, so the Riemann tensors are all zero in the cited coordinate system, hence they are zero in all coordinate systems, and the spacetime is flat.

The Milne model fails observation \citep[\S2.4]{davislineweaver2004}, although it does so less badly than may be expected: 
\begin{quote}
Though not a viable alternative to the currently favored model, the Milne model has great pedagogical value, elucidating the kinematic aspect of the universe's expansion. ... [For] supernovae Ia, it remains a useful reference model when comparing predictions of various cosmological models. \cite[\S5]{chodorowski2005} 
\end{quote}

In conclusion, the Milne model suggests the expansion of the empty universe might be interpreted not only as a stretching of space, but also as a kinematic expansion with kinematic (Doppler) redshifts.

\subsection{Newtonian cosmology --- Friedmann equations in Newton's gravity}
\emph{Newtonian cosmology} is the model of the universe as an infinite, homogeneous, and isotropic matter distribution filling all of $\mathbb R^{3}$, under Newtonian gravity. Curiously, its expansion dynamics turn out to be identical to the FLRW model, which are described by the scale factor $R(t)$ in the Friedmann equations. Hence there is nothing specifically general-relativistic about the rate of expansion, which may raise questions about the usual interpretation of $R(t)$ as expanding space rather than motion through space. (Disclaimer: though equal on this point, on many other aspects --- and overall --- general relativity is observationally superior).

\subsubsection{Friedmann equations}
It is common practice in teaching relativity to either derive the Friedmann equations from Newtonian gravity as a heuristic proof, or alternatively to take the Friedmann equations and reverse-engineer a Newtonian interpretation. The latter starts with the general-relativistic Friedmann equation,
\eqn{\dot a^2-\frac{8\pi\rho}{3}a^2=-k.}
These terms are suggestive of Newtonian kinetic energy, potential energy, and conserved total energy respectively \citep[\S18.4, \S18.7]{hartle2003} \citep[p706--708]{misner+1973}. Misner also et al recommend Bondi's 1961 book.

The former approach traditionally takes a finite sphere of matter --- a subset of the infinite universe --- and applies Newton's laws. This yields the same result, but with $\frac{-2E}{m}$ in place of $k$ \cite{tipler1996a}. The other Friedmann equation may be derived from the conservation equation.
\begin{quote}
Not only are they formally identical, the symbols in the two equations refer to exactly the same quantities. \cite[\S IV]{tipler1996b}
\end{quote}

Thus the rate of expansion is identical to a model where the fabric of space by definition or convention cannot expand. (As another disclaimer, other examples suggest the opposite.)

Historically, this Newtonian parallel was first derived by \cite{milne1934} and \cite{mccreamilne1934}, following a brief offhand mention by Lema\^itre. Friedmann's results for the FLRW model were from 1922. Newtonian cosmology in its static form began at least a few centuries earlier, with Isaac Newton. After formulating his law of gravity, Newton was challenged that the material would all collapse inwards under the attraction, so Newton posited an infinite matter extent with forces balanced in all directions. However this model contained a logical flaw for centuries.

\subsubsection{History and the longstanding flaw}
Simply, the logical inconsistency is that the force integral on any given particle in the Newtonian universe
\eqn{\vec F=\iiint \vec F' dx dy dz}
is divergent. It is tempting to argue from symmetry that the forces will cancel and leave zero net force, and many physicists have seemed content with this approach, including Newton himself. However one may just as ``logically'' derive a net force in any given direction and of any given magnitude, as shown qualitatively by \citet{norton2002}. Quantitatively, suppose that for a particle at $\vec r_1$ a given net force per mass $\frac{\vec F}{m}$ or acceleration $\vec a$ is sought. Then the ball of matter centred at $\vec r_2=\vec r_1+\frac{3}{4\pi G\rho}\vec a$ with radius $\abs{\vec r_2-\vec r_1}$ would exert the desired force, were there no other matter. But then one can use Newton's \emph{shell theorem} to ``remove'' all remaining matter by considering concentric shells outside this ball. Every shell contains $\vec r_1$, hence exerts no force on it. Thus, the total force is $\frac{\vec F}{m}$ as required. It is also possible to obtain an infinite net force in a given direction, by considering opposing hemispherical shells centred on the given point. If the shells have constant thickness, then each contributes an identical force, since their area and volume increase as $r^2$ for distance $r$ from the point, but the gravitational attraction decreases as $\frac{1}{r^2}$. Then match consecutive shells on one side with every second shell on the other. Under this summation order, every pair of shells cancels, but remaining is an infinite number of shells --- and hence infinite force --- on one side.

There is no flaw with the shell theorem. The problem is the integral is \emph{divergent}. The situation is comparable to \emph{Grandi's series} $1-1+1-1+1-1+\cdots$, which also does not converge. Invoking the shell theorem for the Newtonian universe is akin to invoking the result $1-1=0$ for Grandi's series to claim it sums to zero. There is no problem with the result $1-1=0$, rather the sum is divergent. Mainstream mathematics today simply asserts such series or integrals are divergent and have no sum. For philosophical and logical reasons it is better to conclude there is no result, rather than saying there are many results. Logical inconsistencies lead to \emph{ex contradictione quodlibet} or ``from contradiction, anything follows'' \citep{vickers2008}. Hence if one asserts that the net force is both zero and nonzero for instance, then any statement at all would follow logically from such a contradiction --- for instance that planet orbits trace the outline of the Australian coast!
\begin{quote}
What Norton presents as an argument for inconsistency is better understood as just a vivid demonstration of non-convergence... Rather than asserting that Newtonian theory makes inconsistent determinations of gravitational force... Norton should have asserted that it makes no determination \emph{at all}. \citep{malament1995}
\end{quote}

Historically, Hugo von Seeliger popularised awareness of the inconsistency in a rigorous 1895 article and others, demonstrating the problem of divergent forces. (A related issue involves removing a spherical cavity from an infinite mass.) \cite[p348]{rindler2006} But many physicists have not perceived a problem, including Newton a full two centuries earlier, who responded in 1692 to questions with a symmetry argument. In contrast, Einstein overemphasised the issue in motivating his new theory of general relativity \citep{norton1999}. With the advent of Milne and McCrea's dynamical version, new criticism arose from \citet{layzer1954} and others; see references 4 to 11 in \citet{tipler1996b}.

\subsubsection{The resolution}
The resolution to the centuries-long logical inconsistency is to recast Newtonian gravity in a different form --- either as Poisson's law, or on a \manifold{4} as explained shortly. This strengthens the comparison with the FLRW model. But this procedure does come with a philosophical consequence: acceleration becomes relative.

Poisson's equation is:
\eqn{\nabla^2\phi=4\pi G\rho,}
then the force per unit mass on a particle is $F=-\nabla\phi$, where $\phi$ is the potential. After applying a suitable additional constraint such as isotropy, a somewhat-unique ``canonical solution'' results. For any given ``centre'' point $\vec r_0$, we get $\phi(\vec r)=\frac{2}{3}\pi G\rho\abs{\vec r-\vec r_0}^2$ plus a constant. This is homogeneous in the sense that there is no observational significance of the choice of centre, which is a sort of ``gauge freedom''.

This leads to a ``relativity of acceleration'', which is new within Newtonian theory \citep{norton1995}. The acceleration between any two masses in free-fall is independent of the centre point:
\eqn{\vec a = -\frac{4}{3}\pi G\rho(\vec r_2-\vec r_1).}

The third formulation of Newtonian gravity is on a \manifold{}, called Newton-Cartan gravity and originally developed by \'Elie Cartan in 1922 for torsion, although that debate is not relevant here \citep[\S I-D, p395]{hehl+1976}. This theory is equivalent to Poisson's formulation, at least in the sense of giving identical results.

Historically, Newtonian cosmology was first put on solid logical footing by \citet{heckmannschucking1955}, according to \citet[p223]{rindler1977}, although both the problem and its resolution appear to have been rediscovered multiple times. Norton triggered another round of discussion with a paper originally published in 1992 \cite{norton2002} outlining the logical flaw. \citet{malament1995} replied with a thorough analysis of the resolutions, investigating the gravitational models. Independently of this discussion, \citet{tipler1996a,tipler1996b} provided a resolution with a cosmological focus.

One may ask whether it is justified to extend Newtonian gravity in this way, and after all the theory was not handed down in one lump-sum historical package \cite{vickers2008}. But for Newtonian gravity, there is complete equivalence between the traditional force law integral formulation and the differential formulations on their shared domain of applicability. Further, the extension is ``relatively simple and straight forward'' \cite[\S5]{malament1995}.

It is surprising that both the dynamic model and the resolution to the inconsistency were not discovered for centuries \citep[p199]{rindler2006}! Possible reasons for this, which apply at various times in history, include the assumption of a static universe, a lack of emphasis on cosmology, and a lack of mathematical tools such as rigorous description of convergence \citep[\S4]{vickers2008}.

\subsubsection{Conclusions}
Newtonian cosmology is indeed a logically self-consistent model. The equations describing the expansion are identical to the FLRW Friedmann equations. This expansion is described as motion through space in the Newtonian case.

There are ways in which Newtonian cosmology is actually more general --- the constant options are not just fixed to $-1$ \citep[p708]{misner+1973}, it can be done on any geometry \citep[\S3]{tipler1996a}, and it allows torsion. However general relativity is clearly the superior model in terms of observational evidence, as seen in local inhomogeneous patches such as Schwarzschild spacetime with its light deflection and delay. The value of Newtonian cosmology in this case is as a comparison model --- specifically, to test assumptions and the interpretation of expanding space.

\subsection{Conclusions to flat spacetime parallels}
The FLRW model has stronger parallels with the flat spacetime models of the Milne empty universe and Newtonian gravity, than one might expect. These models explain the expansion by motion, and redshift by the Doppler effect.

\section{Doppler and gravitational redshift interpretations --- introduction}
It is common to list three types of redshifts: Doppler, gravitational, and cosmological. More specific types may be categorised under these, for instance the Sunyaev-Zel'dovich effect is a blueshift from cosmic microwave background (CMB) photons colliding with energetic electrons. Hence it can be classified as a Doppler shift. Cosmological redshift is due to the expansion of universe \citep[\S29.2]{misner+1973}, but the stretching of space does not represent any new physical law. Observationally, no difference is detectable between the various interpretations, meaning there is ``cause'' imprinted on the photon.

A commonplace assertion is that the cosmological redshift is \emph{not} a Doppler shift, for instance the redshift does not match the special-relativistic Doppler shift formula using the Hubble law velocity $v=HD$. As another instance, \citet[p386]{defeliceclarke1990} state the local approximation of redshift as a Doppler effect with $v=cz$ is ``sometimes convenient, although rather fictitious''.

However various authors argue cosmological redshift can be interpreted as a series of infinitesimal Doppler shifts \citep[\S6.2a]{padmanabhan1996} \citep{rindler2006} \citep[\S2]{bunnhogg2009} \citep[A-1.5]{davis2004}. Consider an object an infinitesimal coordinate $d\chi$ from an observer, both moving with the Hubble flow. The proper distance between them is $dD=Rd\chi$, and relative velocity $dv=HdD$. A photon will take cosmic time $dt=\frac{dD}{c}$ to traverse the distance. Rearranging, $dv=\frac{\dot R}{R}\frac{dt}{c}=c\frac{dR}{R}$. Since an LIF exists, use the special-relativistic Doppler formula, $dz=\frac{dv}{c}$. Then $dz=\frac{d\lambda}{\lambda}=\frac{dR}{R}$. Finally, upon integrating, the wavelength is $\lambda\propto a$. So the cosmological redshift is entirely made up of local Doppler shifts \citep{davis2004}. \citet[\S3]{bunnhogg2009} go further, to argue this is not just local infinitesimal Doppler shifts, but an overall Doppler shift also, and that this interpretation is the most natural.

An approach with the same conclusion is to parallel-transport the \velocity{} of the emitting galaxy to the observer's location. On a curved manifold there is no absolute comparison of vectors. But if done over a hypersurface of constant cosmic time, the resulting formula fits the special-relativistic Doppler formula! Though a purist might argue this cannot be done unambiguously, this is a ``natural'' choice \citep[\S3]{bunnhogg2009} \citep[\S2]{chodorowski2011}. I suggest another natural choice would be to parallel propagate along the null worldline of the photon's path. Then the parallel transport of the galaxy \velocity{} $\fvec u$ satisfies the following differential equations. Parametrise the path by cosmic time $t$. Then
\begin{alignat}{3}
\diff{u^0}{t} &= -\Gamma^0_{11}\diff{x^1}{t}u^1 &&= -R\dot R\dot\chi u^1\\
\diff{u^1}{t} &= -\Gamma^1_{01}\diff{x^0}{t}u^1-\Gamma^1_{10}\diff{x^1}{t}u^0 &&= -\frac{\dot R}{R}(u^1-\dot\chi u^0).
\end{alignat}
This is more complicated than transporting over a surface of constant $t$, because of an extra term. Other possible equations for the system are the Friedmann equations, and normalisation $u^\mu u_\mu=-1$. I have not analysed these further.

Another example is the Pound-Rebka experiment, which determined the gravitational redshift of photons travelling outwards in Earth's gravitational field. It demonstrated not only the existence of gravitational redshift, but also the equivalence between Doppler and gravitational shifts. However from another perspective, that of a freely falling frame initially stationary and released when the photon is emitted, there is no observed redshift \citep[p115]{schutz2009}! This demonstrates a certain flexibility of interpretation. Still, the Schwarzschild observer frames are the most natural, as discussed in section~\ref{sec:naturalcanonical}, and so the standard interpretation of a purely gravitational redshift interpretation is the most natural.

However the approach investigated for the rest of this document is to push the interpretation of redshifts using the equivalence principle. This idea is from \citet{bunnhogg2009}, whose paper suggests intriguing creative possibilities but provides almost no detail. Their idea is to set up a chain of observers along the spatial path of a photon in some curved spacetime. If each observer is in free fall at the instant the photon passes, then they would naturally interpret any redshift as purely Doppler in origin, because the lack of acceleration feels equivalent to a lack of gravitational field. On the other hand, if each observer remains at constant distance from their neighbour at the instant the photon passes, they would naturally interpret any redshift as purely gravitational in origin, because the zero relative velocity would seem to imply no Doppler shift.
\begin{quote}
The two interpretations arise from different choices of coordinates, or equivalently from imagining different
families of observers along the photons path. \citep{bunnhogg2009}
\end{quote}

The start and end of the chain should coincide with the emitter and receiver of the photon, matching place, time, and velocity \citep{bunnhogg2009}. I relax this requirement slightly by adjusting the endpoints so the end observer velocities are in the same spatial direction as the photon, but still measure the same redshift, as in figure~\ref{fig:endpointadjust}.

\begin{SCfigure}
\centering
  \includegraphics[width=0.3\textwidth]{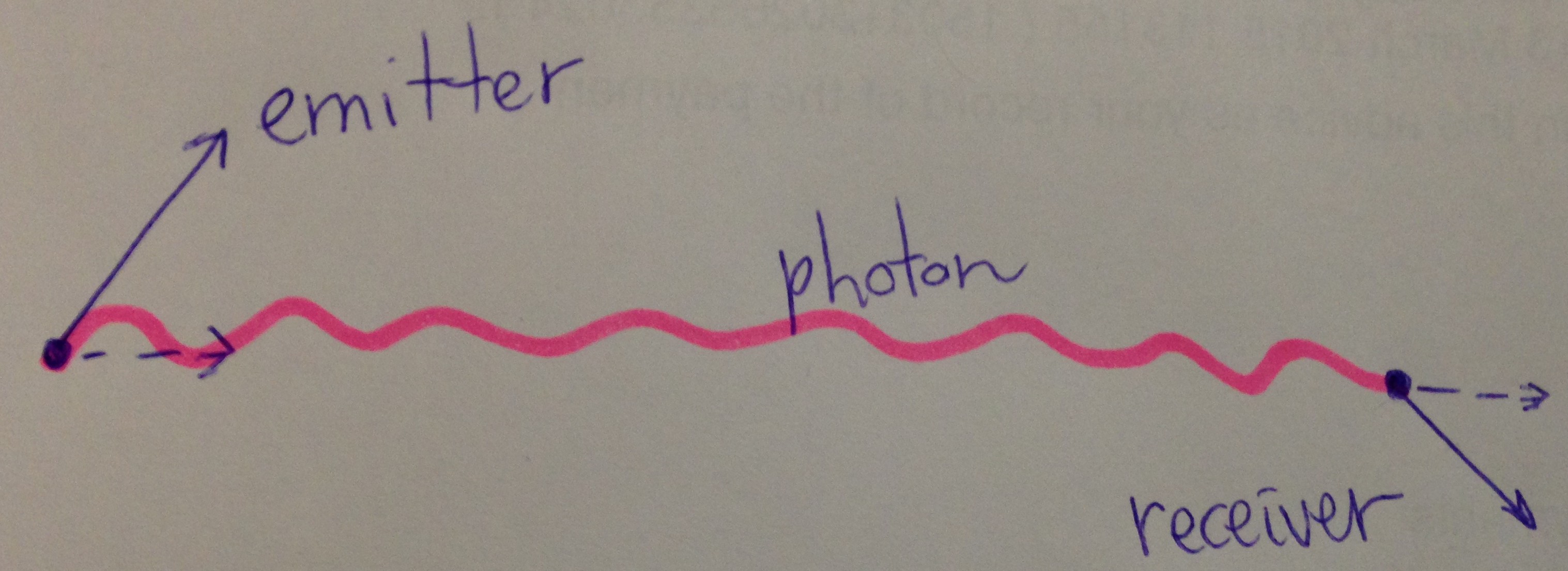}
  \caption{Adjusting the endpoints. An emitter (solid arrow, left) sends out a photon (pink wave) which is later intercepted or measured by a receiver (solid arrow, right). To simplify the situation, find new endpoint observers (dotted arrows) lined up with the photon, moving at the appropriate speed such that they detect the same redshift as the original emitter and receiver. This simplifies the task of determining rigid motions, while not ``cheating'' in any way.}
\label{fig:endpointadjust}
\end{SCfigure}

\subsection{Doppler interpretation}

\begin{quote}
To construct the Doppler family of observers, we require all observers to be in free fall. In this case, within each local inertial frame, there are no gravitational effects, and hence the infinitesimal frequency shift from each observer to the next is a Doppler shift. \citep[\S4]{bunnhogg2009}
\end{quote}

This is a straightforward process for typically used coordinate systems, and thus I move on to tackle the difficult problem of gravitational interpretation.

\subsection{Gravitational redshift interpretation}

\begin{quote}
To construct the gravitational family of observers, we require that each member be at rest relative to her neighbor at the moment the photon passes by, so that there are no Doppler shifts. Initially, it might seem impossible in general to satisfy this condition simultaneously with the condition that the first and last observers be at rest relative to the emitter and absorber, but it is always possible to do so. One way to see that it is possible is to draw a small world tube around the photon path as [earlier]. \citep[\S4]{bunnhogg2009}
\end{quote}
Earlier, they state the spacetime near a geodesic is approximately Minkowski spacetime. They propose constructing Rindler coordinates within this tube. 
\begin{quote}
Because they are not in free fall, the members of the gravitational family all feel like they are in local gravitational fields. Because each has zero velocity relative to her neighbor when the photon goes by, each observer interprets the shift in the photons frequency relative to her neighbor as a gravitational shift.\cite[\S4]{bunnhogg2009}
\end{quote}
They admit this setup is highly contrived, but their point is that redshift interpretation is flexible.

\citet[\S12.4]{rindler2006} also draws a comparison between accelerated coordinates and gravity. By the equivalence principle, the rigid accelerated rod is like a skyscraper in a gravitational field, as both experience acceleration. One limitation is the ``gravitational field'' due to acceleration cannot be constant. The metric can be expressed in a form roughly resembling the Schwarzschild metric.

Finding a gravitational family is closely related to finding rigid body motion in curved spacetimes. But this is not at all straightforward. I first develop a comprehensive foundation for rigidity, which requires constant proper-frame distances, so first step is a thorough critique of distance in relativity.

This construction require only a \dimensional{1} object, which can be visualised as a trail of observers in independently-accelerating rocket ships, or a a coordinate system, using various terminology including ``rod'', ``rope'', ``line'', or ``chain''. A \dimensional{1} rigid object will exist in most scenarios, as for the \dimensional{2} boundaries implemented by Mann and coauthors \citep{epp+2009}, however \dimensional{3} rigid kinematics are trypically far more limited for a given object and spacetime.

\section{Distances and rulers in relativity --- a critique}
\label{sec:distancecritique}
How does one measure distance in general relativity? Specifically, the interpretation of redshifts as gravitational requires the notion of \emph{constant distance}, as measured in the proper frame of observers. This seemingly simple requirement was a difficult part of the research. The technical machinery required is mostly well known to experts --- for instance tetrad frames and the spatial projector, although not the proper metric. But I know of no comprehensive and accessible overview, and thus work towards providing one. (I do ignore quantum considerations, and measurement limits experienced by accelerated observers \citep{mashhoonmuench2002}). This section is a general discussion, followed by section~\ref{sec:lengthcontractionGR} which describes length-contraction in curved spacetimes. On the latter, I did not find an accessible account in any textbook or other source. This section critically overviews various distance measures, and follows on from the introduction in section~\ref{sec:distanceintro}.

\subsection{Distance depends on the time slice}
When consulting a textbook on distance in the Schwarzschild and FLRW spacetimes, one finds overwhelming emphasis on a specific radial distance --- the ``proper distance'' --- and rightly so, as I affirm in section~\ref{sec:naturalcanonical}. But the reader could be forgiven for interpreting this as the \emph{only} radial distance, as having some invariant or absolute quality measured by all observers. For instance one undergraduate general relativity course termed $dR$ the ``physical spatial distance'' and ``a proper length''. \citet[p230]{rindler2006} calls it ``radial ruler distance'', and another book describes it as ``physically measurable distances'', ``physical distance'', and ``actual radial distance'' \cite[p106--107]{moore2012}. But from special relativity, distances are relative and length-contraction occurs between relatively moving frames. Intuitively, one might expect that for an approximately Minkowski spacetime, for instance Schwarzschild spacetime at large $r$, that the familiar special-relativistic properties would approximately hold. So why not in all curved spacetimes?

It is also well known that distance measurements depend on the simultaneity convention chosen, which is defined by the choice of time slice. (This was seen earlier for the proper distance --- in the general definition --- which depends on the traversal through time as well as space). The radial ``proper distance'' in Schwarzschild and FLRW spacetimes corresponds to a slicing by $t$ (respectively far-away time and cosmic time). But why single out these particular coordinates?

For instance Schwarzschild spacetime has a particularly rich variation of (commonly used) coordinates. In the usual choice for distance, Schwarzschild coordinates, setting $dt=d\theta=d\phi=0$ yields radial distance by stationary observers:
\eqn{ds_{\rm radial,stationary}=\Schw^{-1/2}dr,}
as mentioned. Another choice is Gullstrand-Painlev\'e coordinates, which are adapted to raindrops:
\eqn{ds^2=-\Schw dt_r^2+2\Schwroot dt_rdr+dr^2+r^2d\Omega^2.}
Setting the coordinate time interval $dt_r=0$, we have
\eqn{ds_{\rm radial,raindrop}=dr}
This is ironic since sources go to pains to emphasise --- and rightly so --- that the Schwarzschild \coordinate{r} is \emph{not} the radial distance, and yet raindrops are the exception! (Indeed, one paper interprets distance based on the proper time of falling clocks, achieving the same result. This has the added bonus of remaining valid for $r\le 2M$ \citep{gautreauhoffmann1978}.)

Another choice is the (ingoing) Eddington-Finkelstein coordinates, which are adapted to radial photons. Here
\eqn{ds^2=-\Schw dv^2+2dvdr+r^2d\Omega^2.}
Setting the timelike coordinate $dv=0$, we obtain
\eqn{ds_{\rm radial,photon}=0}
I generalise this procedure in section~\ref{sec:lengthcontractionGR}. This quantity is clearly invariant, in the sense that it is defined from the metric for a given path --- and the metric and path can be transformed into other coordinates. But it is not ``invariant'' in the sense of the method of construction --- specifically, simply setting $dt=0$ for the time coordinate in every coordinate system yields only a coordinate-dependent quantity. Thus we should naturally ask why the Schwarzschild coordinates in particular are chosen.

In particular, this is only the ruler distance measured by \emph{Schwarzschild observers} and \emph{Hubble flow comovers} respectively, as seen later. Other motions perceive different distances. In the FLRW case, some authors expressly state the comoving requirement, for instance \citet[\S29.3]{misner+1973}, \citet[p218]{rindler1977}, \citet[p35]{davis2004}, and \citet[\S4]{schmidt1996}. It may be implicit in \citet[p415]{weinberg1972} who describes ``typical galaxies'' or \citet[p367]{rindler2006} who mentions ''galaxies'', as ``galaxy'' is often a rubric for a Hubble comoving object. But lacking is a clear statement of the converse: other radial motions do not measure this distance. This is hinted at in the common statement that non-comoving observers do not perceive an isotropic universe, but it is not clear. One might argue that in the FLRW case, motions are generally close to comoving, so this is the primary frame and distance measure. Though true, there are exceptions, but either way the conceptual understanding is important.

\subsection{Rulers vs photons}
The early relativists used rulers as conceptual measuring devices. Later, they came to prefer light signals, for reasons one textbook articulates:
\begin{quote}
Since meter sticks do not really make sense in general relativity and are not used in actual astronomical measurements, one regards a radar set as the basic distance measuring device. ... [I]n general there is no rigid reference frame defined in any neighborhood of a spacetime point... [This i]s one reason why a ``rigid meter stick'' is not a very useful concept in general relativity. \cite[p135, 59]{sachswu1977}
\end{quote}
Einstein later critiqued his early ``sin'' of rulers and clocks as foundational:
\begin{quote}
But one must not legalize the mentioned sin so far as to imagine that intervals are
physical entities of a special type, intrinsically different from other physical variables
(``reducing physics to geometry'', etc.). Einstein, as cited in \citet[\S3.1, \S4]{brown2018}
\end{quote}
Pauli emphasised ruler measurement is only a secondary effect because its component particles are Lorentz covariant \citep[\S3.2]{brown2018}.

One challenge is the impossibility of an intrinsically Born-rigid ruler, however a resilient ruler which restores its shape after moderate deformation is quite plausible \citep[\S2.5]{rindler1977}. Another objection is the unknown atomic behaviour of a ruler or other object at the microscopic level, where both quantum and relativistic effects are significant. (See the classic article by John Bell showing length-contraction of a slowly accelerating ``classical'' atom, made up of an electron orbiting a nucleus under electric attraction \citep{bell1976}.) Another possible objection is that it is only everyday experience which seeks a $3+1$ splitting of space and time \citep[\S1]{bini2014}.

Since there are doubts about rulers, the ``ruler hypothesis'' is a nontrivial assumption. It states that rulers do in fact achieve their purpose, in other words that physical rigid objects do actually measure distance \citep{lyle2010}. As seen above, this is not a trivial assumption. I concur that light is an especially good measuring tool, but nonetheless I attempt to show rulers are a useful measure after all. In particular, the concept of length for extended objects is essential:
\begin{quote}
without the existence of \emph{some} rigid standard of length the constancy of the speed of light would become a mere convention. \citep[\S2.5]{rindler1977}
\end{quote}

Some question if rigidity is even important, or if a consistent notion even exists. But then what would people mean when they say the Sun has a specific diameter $\approx 1.4\times 10^6$km? One might argue this context is close to flat spacetime. But the case of neutron stars, said to have a diameter $\approx 12km$, is highly relativistic. What of ``bar detectors'' for gravitational waves, which work (hypothetically, since gravitational waves have not yet been detected) precisely because of the discrepancy between their reasonably rigid proper length and the fluctuating spacetime caused by a passing wave? These have been replaced by ``laser interferometers'' for engineering reasons, not theoretical distance issues \citep[p209--210,220--227]{schutz2009}. Similarly the SI definition of the metre is no longer a rigid bar in Paris, but based on light travel time, but again this change is for pragmatic and not distance-theoretic reasons.

\subsection{Hypersurface orthogonal --- a canonical slicing for distances}
\label{sec:orthogonalhypersurface}

\begin{wrapfigure}{r}{0.5\textwidth}
    \includegraphics[width=0.48\textwidth]{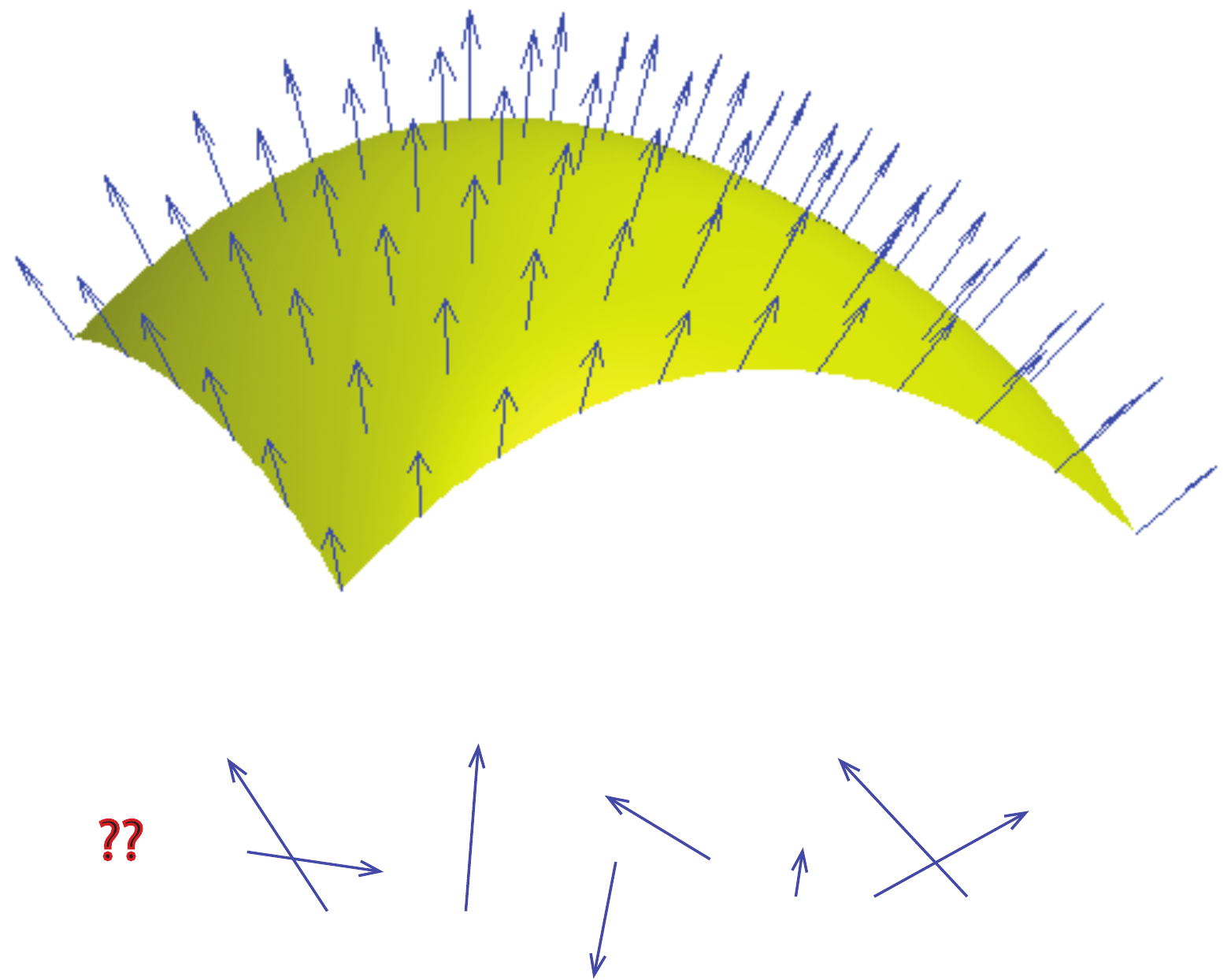}
  \caption{Top: A timelike vector field admitting a surface everywhere orthogonal to the field, and thus a natural measure of spatial distance. Bottom: A messy vector field which seems not to admit an orthogonal surface.}
\end{wrapfigure}

Distance is relative. But in the spacetimes considered in practice there is often a canonical choice of time slicing. This is the set of hypersurfaces orthogonal at every point to the worldline of particles, if such surfaces exist. Such \emph{orthogonal hypersurfaces} yield a natural distance measure \citep[\S1]{bini2014}. As the following authors explain:

\begin{quote}
...[G]enerally speaking, there is no distinguished family of sections (hypersurfaces) across the bundle of flow lines that would represent `the body in space', i.e. mutually simultaneous locations of the body's points. Distinguished cases are those exceptional ones in which $\fvec u$ is hypersurface orthogonal. Then the intersection of $\fvec u$'s flow lines with the orthogonal hypersurfaces consist of mutually \emph{Einstein synchronous} locations of the points of the body. \citep[p37]{guilini2010}
\end{quote}

\begin{quote}
If there is matter with a uniquely defined hypersurface-orthogonal time-like velocity vector... then the physical definition (spatial distance between events... measured by the length of the connecting geodesic segment... most closely reflects the idealized measuring method by use of rods being at rest with respect to the existing matter. \citep[\S5]{schmidt1996}
\end{quote}

From Frobenius' theorem, $\fvec u$ is hypersurface orthogonal if and only if \citep[\S B.3]{wald1984} \citep[p462]{carroll2004}
\eqn{u_{[\alpha}\nabla_\beta u_{\gamma]}=0.}
This holds if and only if \citep[p217]{wald1984}
\eqn{\omega_{\alpha\beta}\equiv\nabla_{[\beta}u_{\alpha]}=0,} 
where $\omega$ is the \emph{rotation}, an antisymmetric tensor. In the specific context of $\mathbb R^3$ for comparison, a fluid dynamics result is that orthogonal surfaces exist if the \emph{helicity} of the flow vanishes:
\eqn{\vec u\cdot\nabla\times\vec u=0,}
which is clearly implied by the vanishing of \emph{vorticity} $\nabla\times\vec u=0$, which holds for a conservative vector field.

For example, Rindler observers are vorticity/rotation free, and hence an orthogonal hypersurface exists. These are the surfaces of constant $T$, which appear as radial lines from the origin in the Rindler chart. See figure~\ref{fig:Rindlerorthogonal}.

\subsection{Natural and/or canonical --- a defense of tradition}
\label{sec:naturalcanonical}

\begin{quote}
``Everything should be made as simple as possible, but not simpler.'' --- Einstein, apocryphal
\end{quote}

Earlier I critiqued the ``proper distance'' as not the only radial distance measures in Schwarzschild and FLRW spacetimes. In this section, I defend it as the best choice of distance measure because the relevant coordinates satisfy numerous canonical/natural and other nice properties, as introduced in section~\ref{sec:naturalintro}.

For the Schwarzschild spacetime,
\begin{quote}
There is just one natural slicing and it turns out to give us spaces whose geometric properties remain constant with time... That is, natural in the sense that the slicing satisfies the technical condition of orthogonality with the world lines of the [central mass] and the field's natural rest states. \citep[\S2]{norton2014}
\end{quote}
Stationary observers measure a constant metric over time, which is related to the timelike Killing vector. Then the orthogonal hypersurfaces to these provide a natural distance measure. The coordinate time corresponds to the proper time for an observer ``at infinity''. Also, note as a curiosity that in the Schwarzschild coordinates many of the traditional Newtonian gravity results also apply, such as Kepler's law.

For the FLRW spacetime, the Hubble flow comovers stand out as a ``privileged Lorentz frame in which the universe appears isotropic'' \citep[\S14.1]{weinberg1972}. The comovers are free of shear and vorticity. Orthogonal to these worldlines are the homogeneous and isotropic spacelike slices, a ``natural choice'' \citep[p67--70]{griffithspodolsky2009} \citep[p780]{misner+1973}. These space slices are ``three-dimensional maximally symmetric subspaces $t=\rm{const}$'', corresponding to 6 Killing vectors which is the maximum possible in $3$ dimensions \citep[p374--375]{defeliceclarke1990} \citep[\S14.1]{weinberg1972}. The coordinates are comoving and synchronous for these observers, and $t$ and $\chi$ are Gaussian normal coordinates \cite[\S14.2]{weinberg1972}.

\citet[p69]{griffithspodolsky2009} describe:
\begin{quote}
[There are] privileged observers for whom the space appears to be isotropic... However, [for others] such an anisotropy would be considered to be due to the motion of the observers in an isotropic background and this could be evaluated.
\end{quote}
In other words, from any given frame, the Hubble frame is singled out. This is a good description of ``natural''.

To the lay person's dictum ``Everything is relative'', it could be added ``...but a few things are invariant, and a few other things are natural or canonical.''

\subsection{Conclusions}
The Schwarzschild and FLRW ``proper distance'' is only the ``radial ruler distance'' for observers with motion $\diff{r}{t}=0$ or $\diff{\chi}{t}=0$ in the respective spacetimes. The claim it is ``actual radial distance'' sounds dubious in a relativistic setting; I do not say this claim is wrong, because this is indeed the most natural choice, however it is misleading without a supplementary statement that this is not the proper-frame distance for most observers. Generally speaking, contra \citet[p171]{schutz2009}, it seems rods and clocks do not measure the metric.

In summary,
\begin{itemize}
\item The Schwarzschild and FLRW ``proper distance'' is the proper-frame distance measured by comovers
\item This is \emph{not} the proper-frame distance measured by other motions, in general
\item There are geometric reasons for preferring or singling out the above coordinate systems.
\end{itemize}

\section{Length-contraction in general relativity}
\label{sec:lengthcontractionGR}
The interpretation of redshifts as gravitational requires the measurement of constant proper-frame distance, which in turn requires analysis of length-contraction in curved spacetimes. I demonstrate how to Lorentz boost in the Schwarzschild and FLRW geometries, drawing on and generalising a very helpful example from \citet{taylorwheeler2000}. This leads to length-contraction by a factor $\gamma$ in local inertial coordinates. The spatial projector provides a complementary result, yielding an increase by a factor $\gamma$. This is simply the familiar mutually-perceived length-contraction from special relativity, interpreted in a different light, as in figure~\ref{fig:lengthcontraction}.

This section is original material, or at least done independently. Though I assume others would have done these calculations, I am not aware of any sources.

\subsection{Lorentz boost --- Schwarzschild spacetime}
This section follows the construction of \citet{taylorwheeler2000} who effectively perform a Lorentz boost from Schwarzschild coordinates to the Gullstrand-Painlev\'e raindrop coordinates. I generalise this procedure to arbitrary radial motions, add technical justifications based on local inertial frames, and later extend to FLRW spacetime.

The usual Schwarzschild coordinates, which are adapted to stationary observers, do not give rise to local inertial frames because the metric does not have the form $\diag(-1,1,1,1)$. But a coordinate transformation, a rescaling, to ``shell coordinates'' does yield this form: \citep[p2-22, 2-23]{taylorwheeler2000}
\eqn{dt_{\rm shell}\equiv\Schw^{1/2}dt, \qquad\qquad dr_{\rm shell}\equiv\Schw^{-1/2}dr.}
These describe the proper time and radial distance measured by a given Schwarzschild observer. The name refers to these stationary observers, as if they were standing on a spherical shell at some fixed $r$ such as the surface of the Earth. Note $dr_{\rm shell}$ is equivalent to the $dR$ notation used elsewhere. Under this transformation the metric becomes \citep[p2-33]{taylorwheeler2000}:
\eqn{ds^2=-dt_{\rm shell}^2+dr_{\rm shell}^2+r^2d\Omega^2.}
Locally, this is simply flat space! (Perhaps technically the angular coordinates should also be rescaled for the calculation, but as there is no motion in these directions they are unaffected by the Lorentz boost and length-contraction.) Of course a Schwarzschild observer is not freely falling, however over short-enough timescales gravity is negligible and the frame is inertial, as assumed for nuclear physics experiments on Earth for example. This is achieved by working with differentials like ``infinitesimal'' time $dt$ \citep[p2-33, 2-34]{taylorwheeler2000}. After rescaling $d\theta$ and $d\phi$ also, the metric has form $\diag(-1,1,\ldots)$ for all $r>2M$. This holds not just at a point but for all $r>2M$, hence the first partial derivatives of the metric vanish. These are not only local inertial frames, but have the especially convenient property of a \emph{coordinate basis} giving rise to them \citep[p485--486]{carroll2004}.

It is valid to perform local Lorentz boosts in these local inertial coordinates. Suppose the boost in the radial direction is by speed $V(t,r)$ and Lorentz factor $\gamma(V)$, where $V$ has the same sign as $\Delta r$ --- positive for outward motion, and negative for inward motion. Then the new boosted time coordinate $T$ follows from equation~\ref{eqn:Lorentzboosttime}, as demonstrated in \citet[pB-13]{taylorwheeler2000}:
\begin{align}
dT &= -V\gamma dr_{\rm shell}+\gamma dt_{\rm shell} \\
   &= -V\gamma\Schw^{-1/2}dr+\gamma\Schw^{1/2}dt,
\end{align}
switching back to Schwarzschild coordinates after the boost. Solving for $dt$,
\eqn{dt=\frac{1}{\gamma}\Schw^{-1/2}dT+V\Schw^{-1}dr.}
\citet[pB-13]{taylorwheeler2000} go on to substitute the raindrop frame values of $V$ and $\gamma$, then plug $dt$ into the Schwarzschild coordinates, thus deriving the Gullstrand-Painlev\'e coordinates . Instead, I preserve generality by keeping $V$ arbitrary, and substitute the above expression into the Schwarzschild coordinate metric:
\begin{align}
ds^2 &= -\Schw\left(\frac{1}{\gamma}\Schw^{-1/2}dT+V\Schw^{-1}dr\right)^2+\Schw^{-1}dr^2+r^2d\Omega^2 \\
     &= -\frac{1}{\gamma^2}dT^2-V^2\Schw^{-1}dr^2-\frac{2V}{\gamma}\Schw^{-1/2}dTdr + \Schw^{-1}dr^2+r^2d\Omega^2 \\
		 &= -\frac{1}{\gamma^2}dT^2-\frac{2V}{\gamma}\Schw^{-1/2}dTdr +\frac{1}{\gamma^2}\Schw^{-1}dr^2+r^2d\Omega^2,
\end{align}
after expanding and collecting terms, then using $1-V^2=\gamma^{-2}$. (As a check, for $V=0$, $\gamma=1$ and so $ds^2=-dT^2+\Schw^{-1}dr^2+r^2d\Omega^2$, which is equivalent to the usual Schwarzschild metric since for $V=0$, $dT^2=dt_{\rm shell}^2=\Schw dt^2$.)

These coordinates $(T,r,\theta,\phi)$ describe the same Schwarzschild geometry, but are adapted to the motion $V$. Thus the hypersurfaces of constant $T$ describe the space part adapted to the motion. Under $dT=d\theta=d\phi=0$ the metric becomes:
\eqn{ds_{\rm radial}=\gamma^{-1}\Schw^{-1/2}dr,}
after taking the square root. In terms of $dR$,
\eqn{ds_{\rm radial}=\gamma^{-1}dR.}
This is precisely the familiar length-contraction formula from special relativity! The justification of expressing this in terms of $dR$ ($=dr_{\rm shell}$) rather than $dr$ is the former is an inertial coordinate.

The length-contraction formula only holds locally, so integration is required to compute a macroscopic quantity. But this is not unfamiliar, as in special relativity an object will not always lie in one global inertial frame, as for the Rindler system in all frames but one (at any given instant).

One must not forget distances are relative, and here the comparison is between a Schwarzschild observer and an observer moving at speed $V$ relative to them, where both occur at the same place and time. The formula states that the ``moving'' (non-stationary) observer measures the Schwarzschild observer to be length-contracted by a factor $\gamma$.

Since $\gamma\ge 1$, the maximum value of $ds_{\rm radial}$ is $dR$ when $V=0$, but this is simply the familiar observation from special relativity that the maximum length of an object --- it's ``proper length'' --- is the length measured in its own frame. This demonstrates that only Schwarzschild observers measure the so-called radial ``proper distance'' $dR$ --- actually any motion with $\diff{r}{t}=0$ --- and other observers do not.




There are others who have performed a Lorentz boost on tetrads, including \citet{hamiltonlisle2008} who cite Taylor and Wheeler, and \citet[\S6]{bini2014}. But overall, information seems scant.




\subsection{Spatial projector --- length-contraction in reverse --- Schwarzschild spacetime}
The spatial projection tensor describes the distance measured by rulers in the local proper frame of an observer. It depends on both the background spacetime metric $g$ and the \velocity{} $\fvec u$ of the motion. This actually yields a \emph{greater} distance than $dR$, which I interpret as due to measuring using length-contracted rulers. This time, the relative motion is considered from the frame of the Schwarzschild observers. These measure the ``moving'' observer to be length-contracted, thus its rulers are also contracted, and so it would measure an increased distance (see figure~\ref{fig:lengthcontraction}). In other words, from the local frame of a stationary observer, the ``moving'' observer would measure a greater length for the stationary observer.

Consider an observer moving radially with speed $V(r,t)$ relative to the local Schwarzschild observer, as before. The calculation requires the \velocity{} of the motion, which is determined in Schwarzschild coordinates as follows. Let $\fvec u_{\rm shell}$ be the \velocity{} of the Schwarzschild observer. The spatial components of $\fvec u_{\rm shell}$ are $0$, hence by the normalisation requirement $g(\fvec u_{\rm shell},\fvec u_{\rm shell})=-1$ it follows $u_{\rm shell}^t=\Schw^{-1/2}$; in summary,
\eqn{u_{\rm shell}^\mu=\left(\Schw^{-1/2},0,0,0\right).}
Label by $\fvec u$ the \velocity{} of the radial motion, which has form
\eqn{u^\mu=(u^t,u^r,0,0),}
where the nonzero components are to be determined. From equation~\ref{eqn:relativespeed}, the Lorentz factor satisfies the inner product:
\eqn{\label{eqn:propergamma}\gamma=-\fvec u \cdot \fvec u_{\rm shell}.}
Only the time component of $\fvec u_{\rm shell}$ is nonzero, and so
\eqn{\gamma = (-1)(-1)\Schw u^t \Schw^{-1/2},}
or
\eqn{u^t=\gamma\Schw^{-1/2}}
after rearranging. Then $u^r$ follows from the normalisation $g(\fvec u,\fvec u)=-1$:
\eqn{-\Schw(u^t)^2+\Schw^{-1}(u^r)^2=-1,}
so
\eqn{\Schw^{-1}(u^r)^2=-1+\Schw(u^t)^2=-1+\Schw\gamma^2\Schw^{-1}=\gamma^2-1.}
But
\eqn{\gamma^2-1=\frac{1}{1-V^2}-1=\frac{1-(1-V^2)}{1-V^2}=\frac{V^2}{1-V^2}=V^2\gamma^2,}
so
\eqn{(u^r)^2=V^2\gamma^2\Schw,}
then
\eqn{u^r=\pm\abs V\gamma\Schw^{1/2}=V\gamma\Schw^{1/2},}
since $u^r$ has the same sign as $V$ and the other terms are always positive. Hence the \velocity{} of the moving observer is
\eqn{\label{eqn:velflow}u^\mu=\left(\gamma\Schw^{-1/2},V\gamma\Schw^{1/2},0,0\right).}

The projection tensor requires the \velocity{} in covariant form. Because the metric is diagonal, $u_\alpha=g_{\alpha\alpha}u^\alpha$. Hence $u_t=g_{tt}u^t=-\Schw\gamma\Schw^{-1/2}=-\gamma\Schw^{1/2}$, and $u_r=g_{rr}u^r=\Schw^{-1}V\gamma\Schw^{1/2}=V\gamma\Schw^{-1/2}$. That is,
\eqn{u_\mu=\left(-\gamma\Schw^{1/2},V\gamma\Schw^{-1/2},0,0\right).}
One term of the projector is $\fvec u\otimes\fvec u$, which is formed by multiplying the entries of $u_\mu$ with themselves:
\eqn{u_\mu u_\nu =
\begin{pmatrix}
\gamma^2\Schw & -V\gamma^2 & 0 & 0 \\
-V\gamma^2 & V^2\gamma^2\Schw^{-1} & 0 & 0 \\
0 & 0 & 0 & 0 \\
0 & 0 & 0 & 0 \\
\end{pmatrix}.}
Thus the projection tensor is
\eqn{P_{\mu\nu}=g_{\mu\nu}+u_\mu u_\nu =
\begin{pmatrix}
V^2\gamma^2\Schw & -V\gamma^2 & 0 & 0 \\
-V\gamma^2 & \gamma^2\Schw^{-1} & 0 & 0 \\
0 & 0 & r^2 & 0 \\
0 & 0 & 0 & r^2\sin^2\theta \\
\end{pmatrix}}
where the equality $\gamma^2-1=V^2\gamma^2$ was used twice.

Note the proper metric $\gamma_{ij} \equiv P_{\mu\nu}\pdiff{x^\mu}{X^i}\pdiff{x^\nu}{X^j}$ gives the same result in this case. The derivative terms can be taken as delta functions $\pdiff{x^\mu}{X^i}=\delta^\mu_i$since there is no need of tracking particles over time and an instantaneous snapshot is all that is required. Another way of seeing this is to define the material manifold by the present time slice, so the material coordinates would simply be the spacetime coordinates. Then $\gamma_{ij}=P_{ij}$, which forms the lower-right $3\times 3$ matrix of $P_{\mu\nu}$.

In the radial direction, the spatial metric gives
\eqn{dL^2=\gamma_{rr}dr^2=P_{rr}dr^2=\gamma^2\Schw^{-1}dr^2,}
so
\eqn{dL=\gamma\Schw^{-1/2}dr=\gamma dR.}

This time the length is \emph{increased} by the factor $\gamma$! Though not the usual way of thinking about length-contraction, this is nonetheless equivalent. That the alternate approaches of Lorentz boosts and the spatial projector give the same result --- in a reciprocal sense --- is an important confirmation of the result. This also affirms the previous result that the radial ``proper distance'' is only measured by Schwarzschild observers and others with $\diff{r}{t}=0$.

One application is the popular conceptual scenario of an astronaut falling feet-first into a black hole. At what $r$ would they be ``spaghettified'', meaning stretched lengthwise and compressed sideways by tidal forces? The astronaut would experience increasing length-contraction relative to Schwarzschild observers as they plummet. From the local Schwarzschild frames, the steadily advancing length-contraction would add to the tension on the astronaut's body as it is required to contract to remain rigid. But the shortened length (in Schwarzschild coordinates) leads to a reduced tidal acceleration between head and feet. The former effect seems paltry compared to the latter, so overall the astronaut would last longer. In the case of sufficiently large black holes where the astronaut would survive past the event horizon, the length-contraction formula derived here is not valid because it is relative to Schwarzschild observers which are impossible for $r\le 2M$, but one could try another coordinate system. Finally, even if the scenario of the falling astronaut is unrealistic, it is conceptually important for understanding general relativity.

\subsection{Lorentz boost --- FLRW spacetime}
I repeat the earlier Lorentz boost construction, to demonstrate length-contraction in FLRW spacetime. Observers ``comoving'' with the Hubble flow are the natural choice (see section~\ref{sec:naturalcanonical}), and these are also the obvious choice of reference frame from which to describe length-contraction. Suppose an observer moves at speed $V(t,\chi)$ with respect to the local Hubble flow. Orient the standard coordinate system $(t,\chi,\theta,\phi)$ so the motion is radial. As before, take the sign of $V$ as the sign of $\Delta\chi$, so for motion receding from the origin $V>0$.

The FLRW metric is
\eqn{ds^2 = -dt^2 + R(t)^2\left(d\chi^2 + S_k(\chi)^2d\Omega^2\right).}
The so-called ``proper distance'' is the spatial distance $R(t)\chi$ from the coordinate origin. Define a new radial coordinate $X$ by the simple rescaling:
\eqn{dX\equiv R(t)d\chi,}
under which the metric becomes
\eqn{ds^2=-dt^2+dX^2+R(t)^2S_k(\chi)^2d\Omega^2.}
Notice this is a locally inertial coordinate system, at least in the $t$ and $X$ coordinates. This is reminiscent of conformal coordinates, in that they also make spacetime appear flat. As before, the angular coordinates $\theta$ and $\phi$ could be rescaled or simply left unchanged since they are orthogonal to the motion. Define a new time coordinate $T$ by the Lorentz boost
\eqn{dT=-V\gamma dX+\gamma dt=-V\gamma R(t)d\chi+\gamma dt.}
Rearranging, $dt=\frac{1}{\gamma}dT+R(t)Vd\chi$. Substituting this into the metric,
\eqn{ds^2=-\frac{1}{\gamma^2}dT^2-\frac{2V}{\gamma}R(t)dTd\chi+(1-V^2)R(t)^2d\chi^2+R(t)^2S_k(\chi)^2d\Omega^2.}
Determining spatial sections by constant $T$, the radial distance is
\eqn{ds_{\rm radial}^2=(1-V^2)R(t)^2d\chi^2=\gamma^{-2}R(t)^2d\chi^2,}
so
\eqn{ds_{\rm radial}=\gamma^{-1}R(t)d\chi.}

As in the Schwarzschild case, this is local length-contraction by a factor of $\gamma$, in this case relative to the usual distance interval $R(t)\chi$. Expressed in words, an observer moving at speed $V$ relative to the local Hubble flow will measure the Hubble comovers to be length-contracted by a factor $\gamma(V)$.

This also demonstrates the (radial) so-called ``proper distance'' $R\chi$ is only that measured by observers moving with the Hubble flow, or more generally those with $\diff{\chi}{t}=0$. All other observers measure a different distance, locally.





There is an important application to our motion through the universe. The Milky Way and other galaxies in the Local Group are travelling at $\approx 600$km/s relative to the Hubble flow, as determined from the cosmic microwave background (CMB). This is a Lorentz factor of $\gamma\approx 1.000002$, or $\gamma^{-1}\approx 0.999998$. Hence from the frame of the Earth (more precisely, the Local Group), distances in the direction of motion (roughly towards the constellation Hydra, and also the opposite direction) are decreased by $\approx 0.0002\%$, relative to the measurement made by Hubble comovers. This calculation does not consider the Earth's orbit around the Sun, nor the Solar System's orbit around our galaxy's centre. Note the questionable approach of taking the $600$km/s figure in a Minkowski space calculation would give a correct result, since the proper distance $R\chi$ is an inertial coordinate. But in a curved spacetime, such a measurement applies only locally; it would still make sense to consider a whole chain of observers travelling at $600$km/s relative to their local Hubble flow, then the compendium of yields an overall contraction. Either way, this calculation would presumably have been done by others.

\subsection{Spatial projector --- FLRW spacetime}

With the same setup as the above, label the \velocity{} of Hubble comovers by $\fvec u_{\rm Hubble}$ and the \velocity{} of the relatively ``moving'' observer $\fvec u$. Then
\eqn{u_{\rm Hubble}^\mu=\left(1,0,0,0\right),}
and
\eqn{\fvec u=\left(u^t,u^\chi,0,0\right)}
for components to be determined. Then $\gamma=-\fvec u_{\rm Hubble}\cdot\fvec u=u^t$, so $u^t=\gamma$. Then by normalisation $u_\mu u^\mu=-1$, it follows $-\gamma^2+R(t)^2(u^r)^2=-1$, so $R(t)^2(u^r)^2=\gamma^2-1=V^2\gamma^2$, so $u^r=\frac{V\gamma}{R(t)}$, taking $u^r$ to have the same sign as $V$. Then:
\eqn{\label{eqn:FLRW4velocity}u^\mu=\left(\gamma,\frac{V\gamma}{R(t)},0,0\right).}
In covariant form, from equation~\ref{eqn:contratocov},
\eqn{u_\mu=\left(-\gamma,V\gamma R(t),0,0\right).}
Then the spatial projector tensor is
\eqn{P_{\mu\nu}=g_{\mu\nu}+u_\mu u_\nu=\begin{pmatrix}
V^2\gamma^2 & -V\gamma^2 R(t) & 0 & 0 \\
-V\gamma^2 R(t) & \gamma^2 R(t)^2 & 0 & 0 \\
0 & 0 & R(t)^2S_k(\chi)^2 & 0 \\
0 & 0 & 0 & R(t)^2S_k(\chi)^2\sin^2\theta \\
\end{pmatrix}.}

As before, the proper metric yields the same result in this case, because it is an instantaneous snapshot, not tracking particles over time. From the component $\gamma_{\chi\chi}=P_{\chi\chi}=\gamma^2 R(t)^2$, the local ``radial'' ($\chi$) distance is
\eqn{\label{eqn:FLRWlengthincrease}dL=\gamma R(t)d\chi,}
which is the usual ``proper distance'' $R(t)d\chi$ expanded by the Lorentz factor $\gamma$.

Hence the observer's rulers, as perceived in the frames of the Hubble flow, would measure an \emph{increased} length by a factor $\gamma$ in the direction of motion. This concurs with the previous result that only Hubble comovers measure a distance $R\chi$. To be precise, all motions which coincide with the Hubble flow in at least one direction measure a local distance $Rd\chi$ in that direction.

\subsection{Theorem --- Spatial metric for stationary/comoving observers}
In this section I prove an original theorem which leads to a host of example rigid systems, and hence a selection of possible gravitational redshift interpretations. Though it is done independently and I do not know of any source for it, presumably others would have derived it previously, because it is a straightforward computation.

\begin{thrm}
For a given coordinate system $x^\mu$ with \emph{diagonal} metric $g_{\mu\nu}$, the spatial metric of particles stationary with respect to the coordinates is simply $g_{ij}$, which is the spatial part of the spacetime metric.
\end{thrm}

\begin{proof}
Each observer/particle remains at constant spatial position with respect to the specified coordinates. Then by normalisation,
\eqn{u^\mu=\left(\sqrt{-g_{tt}^{-1}},0,0,0\right),}
noting $-g_{tt}$ is positive with the $-+++$ convention. In covariant form, $u_\mu=g_{\mu\nu}u^\nu=g_{\mu t}u^t=g_{tt}u^t\delta^t_\mu$, so since the metric is diagonal,
\eqn{u_\mu=\left(g_{tt}u^t,0,0,0\right)=\left(\sqrt{-g_{tt}},0,0,0\right).}
Then $u_\mu u_\nu=0$ except for $(\mu,\nu)=(t,t)$ for which $(u_t)^2=-g_{tt}$. Thus
\eqn{P_{\mu\nu}=g_{\mu\nu}+u_\mu u_\nu=\diag(0,g_{11},g_{22},g_{33}).}
The proper metric $\gamma_{ij}=P_{\mu\nu}\pdiff{x^\mu}{X^i}\pdiff{x^\nu}{X^j}$ gives the same result in this case. Since the particles remain at fixed spatial positions, the $x^\mu$ are also material coordinates, and so $\pdiff{x^\alpha}{X^i}=\pdiff{x^\alpha}{x^i}=\delta_i^\alpha$. Hence
\eqn{\gamma_{ij}=P_{ij}=\begin{pmatrix}
g_{11} & 0 & 0 \\
0 & g_{22} & 0 \\
0 & 0 & g_{33} \\
\end{pmatrix}.}
\end{proof}

Note the Landau-Lifshitz radar distance concurs with this result, as for a diagonal metric, $\gamma^{\rm LL}_{ij}=g_{ij}$.

\emph{Corollary}: If in addition the metric is time-independent, then the stationary observers form a rigid system.


Thus Schwarzschild observers form a rigid system. In FLRW spacetime, the theorem shows Hubble comovers do measure the ``proper distance'' $Rd\chi$. But since the metric evolves over time, in general, they do not form a rigid system. The theorem concurs with the previous result that Schwarzschild observers and Hubble comovers measure the ``proper distance'' as defined in their respective spacetimes.

The theorem does not apply to Gullstrand-Painlev\'e coordinates for raindrops, because the metric is not diagonal in these coordinates. This fits with intuition that such a system would not be rigid. Neither are the Eddington-Finkelstein coordinates diagonal. Another situation which fails to meet the criteria is $r<2M$ for a black hole in Schwarzschild coordinates, because no particle can remain stationary there.



\section{Gravitational redshifts --- flat spacetime}
There is no gravitational field in Minkowski spacetime, and so any redshifts must be Doppler shifts. At least this would be the most natural interpretation of an all-knowing ``bookkeeper''. But for a local observer who experiences acceleration and may or may not have knowledge of the global spacetime, this acceleration is locally indistinguishable from gravity according to the equivalence principle.

This section demonstrates matching a photon transit to the Rindler chart, then investigates arbitrary accelerations in special relativity, before considering one possible extension to FLRW spacetime.

\subsection{Gravitational interpretation --- Rindler observers}

\begin{SCfigure}\centering
  \includegraphics[width=0.4\textwidth]{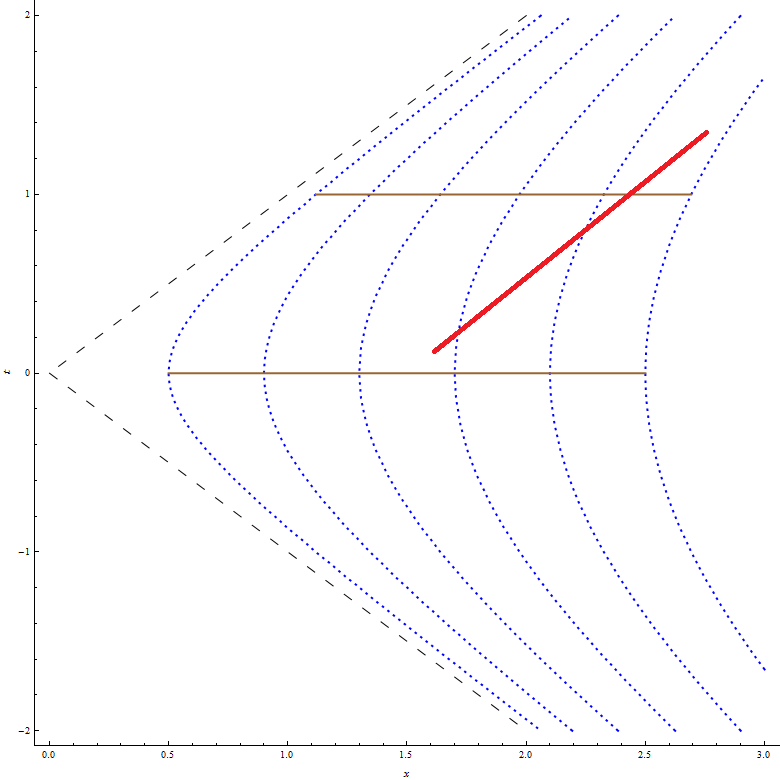}
  \caption{The photon cuts a $45^\circ$ path (red) through the Rindler wedge. The endpoints can be matched up with the desired path. This yields a gravitational redshift interpretation in Minkowski space.}
\end{SCfigure}

In Minkowski space, the Rindler chart makes it straightforward to find a family of observers to interpret a given frequency shift as gravitational. Suppose the \velocity{} of a given emitter and receiver are known, as well as the energy of a photon at emission.

The first step is to convert these endpoint motions to motions collinear (spatially) with the photon. So long as the adjusted emitter and receiver measure the same initial and final energies of the photon, observers would accept this adjustment without suspecting any ``cheating''. Then the velocities are all collinear, so a Rindler chart can be defined with this direction as the \coordinate{x}.

Suppose the \velocity{} of the observer to be adjusted is $\fvec u$, and photon \momentum{} $\fvec p$. Align Minkowski coordinates so that $\fvec p$ is in the \direction{x}. The observer measures photon energy $E=-\fvec p\cdot\fvec u=p^tu^t-p^xu^x$. The adjusted \velocity{} $\fvec u_{\rm new}$ will have components $u_{\rm new}^\mu=(u^t,u^x,0,0)$, and satisfy the requirement $E=-\fvec p\cdot\fvec u_{\rm new}=p^tu_{\rm new}^t-p^xu_{\rm new}^x$. Normalisation of $\fvec u_{\rm new}$ is $-(u_{\rm new}^t)^2+(u_{\rm new}^x)^2=-1$. $E+p^xu_{\rm new}^x=p^tu_{\rm new}^t$ so squaring, $E^2+2Ep^xu_{\rm new}^x+(p^xu_{\rm new}^x)^2=(p^tu_{\rm new}^t)^2$. The right hand side is $(p^t)^2(u_{\rm new}^t)^2=(p^t)^2(1+(u_{\rm new}^x)^2)$. Expanding out, this gives a quadratic in $u_{\rm new}^x$ which can be solved given the input variables, and thus $\fvec u_{\rm new}$ is determined.


Now assume the endpoint observers --- the emitter and receiver --- have been adjusted. Suppose the emitter is at $(t_1,x_1)$ in the chosen coordinates, where the other coordinates are zero and suppressed. A null worldline connects the observers, so the reception event is at $(t_1+d,x_1\pm d)$ for some coordinate interval $d>0$.

From straightforward algebra, it can be shown
\eqn{x_1=d\frac{1-v_2}{v_2-v_1}, \qquad\qquad t_1=v_1d\frac{1-v_2}{v_2-v_1},}
where $v_i\equiv\frac{t_i}{x_i}$ is the coordinate speed of the Rindler observers matching the emitter and receiver. This yields a chain of observers, with the endpoints coincident with the emitter and receiver observers, who remain at constant proper-frame distance. Thus each would locally interpret the infinitesimal redshift as a purely gravitational shift. And thus it is an overall gravitational shift, in this interpretation.


The photon takes nonzero coordinate time $t$ to transit, which allows more flexibility. For instance, there is no frame in which part of the Rindler rod travels left, and part travels right. But due to the finite photon travel time, a photon path which traversed from some position at $t<0$ to another position at $t>0$ could include this property.

\subsection{Arbitrary accelerations in special relativity}
The case of arbitrary accelerations in general curved spacetimes is difficult, so a natural precursor is arbitrary accelerations in Minkowski space. Because it is just in $1+1$ dimensions, the situation is considerably simplified.

Noether and Herglotz both showed the motion of a rigid body in special relativity is determined by a single particle. But one may wonder questions like, does the local length-contraction depend only on the leading particle's instantaneous speed, or its full worldline history? From the particular case of Rindler coordinates, the velocity of a single particle is not sufficient to determine the behaviour of the rod, as inspection of the chart shows. Along a radial line, $T$ and hence the velocity is constant, but all these particles have different accelerations.

In the Rindler chart, the following properties hold, and one might wonder if they hold for accelerated rigid motion more generally:
\begin{itemize}
\item to stay rigid individual proper accelerations differ
\item to stay rigid individual speeds differ
\item speed depends on both speed and acceleration of adjoining particle
\item there is a limit on the extent, could possibly call a ``rigidity horizon'' or similar. Can't extend rigid structure indefinitely through space under the given requirements.
\end{itemize}

Back to arbitrary accelerations (but still \dimensional{1} and collinear), suppose $x^\mu$ are inertial coordinates corresponding to an observer. Let $s$ be a material (``Lagrangian'') coordinate, with the additional requirement that it be scaled to proper length. Label by $v\equiv \diff{x}{t}$ the coordinate speed, and $a\equiv\diff{v}{t}$ for the coordinate (not proper) acceleration. $x$, $v$, and $a$ are all functions of both $s$ and $t$. The local or infinitesimal length-contraction is: \citep[p50--51]{rindler1977}
\eqnboxed{\label{eqn:dxds}\pdiff{x}{s}=\gamma^{-1}=\sqrt{1-v^2}}
The higher the speed, the more material fits into a given coordinate interval, as measured from the observer's frame. Differentiate both sides with respect to $t$. Then $\frac{d^2x}{dsdt}=\pdiff{(\gamma^{-1})}{t}$. The right hand side is
\eqn{\pdiff{}{t}\gamma^{-1}=\pdiff{}{t}\sqrt{1-v^2}=\frac{1}{2}(1-v^2)^{-1/2}\cdot(-2)v\cdot\pdiff{v}{t}=-v\gamma\pdiff{v}{t}=-v\gamma a,}
and the left hand side is $\diff{}{s}\left(\diff{x}{t}\right)=\diff{v}{s}$. Hence
\eqnboxed{\pdiff{v}{s}=\pdiff{\gamma^{-1}}{t}=-v\gamma a}
This describes the propagation of the speed along the material, in a snapshot of constant coordinate time. Differentiate again with respect to $t$, and after a few lines one can show:
\eqnboxed{\pdiff{a}{s}=-\pdiff{(v\gamma a)}{t}=-\gamma^3a^2-\gamma v\dot a.}
This describes the propagation of acceleration along the material. The $\dot a$ term is unappealing, but there is no point continuing this process ad-infinitum. I derived the latter two independently, but they would certainly have been derived before. Note they all resulted simply from the length-contraction formula. As a check, the Rindler chart satisfies these equations.

From equation~\ref{eqn:dxds} and the definition of $v$ it follows that \citep[\S2]{eriksen+1982}
\eqn{\label{eqn:rigid1D}\pdiffsquared{x}{s}+\pdiffsquared{x}{t}=1.}
This equation is not explicitly covariant (tensorial), but it can be put into a form which is: \citep{soderholm1982}
\eqn{\pdiffsquared{s}{x}-\pdiffsquared{s}{t}=1.}

In fact the following features of the Rindler chart do carry through to generalised \dimensional{1} accelerated rigid motion in special relativity. There is a single frame in which the whole \dimensional{1} rigid body is at rest \citep[\S4]{eriksen+1982}. If a given point has proper acceleration $\frac{1}{X}$, then the maximum extent of the rigid body is $X$ in the direction opposite to the acceleration. For such a limit point $X_0$, the proper acceleration of other points is $(X-X_0)^{-1}$ \citep[\S5]{eriksen+1982}. The distribution of the acceleration across space is the same as for the Rindler chart \citep[\S6]{eriksen+1982}. Contrast the approach in \citet[\S2.4]{dewitt2011}.

This analysis of rigid body dynamics in $1+1$ dimensions in flat spacetime is linked to gravitational redshift interpretations, and is a preliminary step to application in curved spacetime.

\section{Gravitational redshifts --- the Schwarzschild fishing line}
\label{sec:fishingline}

\begin{wrapfigure}{l}{0.5\textwidth}
\includegraphics[width=0.48\textwidth]{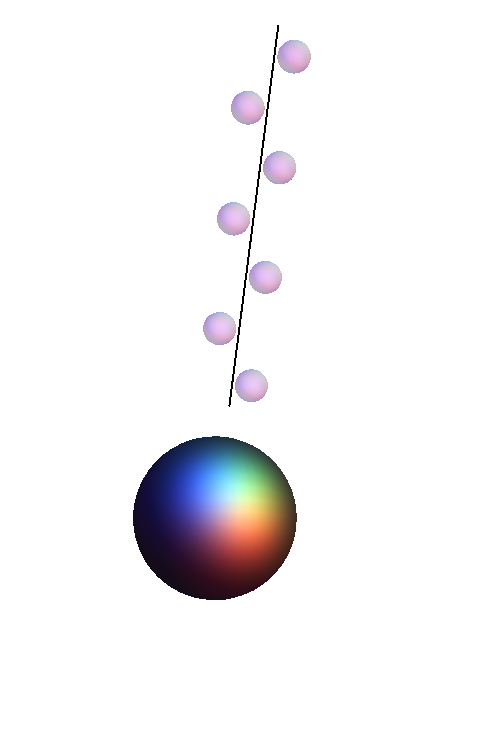}\label{fig:fishinglineconveyor}
\caption{A fishing line dangling over an unusually colourful black hole. The line is moving radially, and is hence length-contracted as measured by Schwarzschild observers. It is kept perfectly rigid by a large number of fixed pulleys or winches (in pink), each rotating at a constant rate depending on their position. This ensures the overall motion is Born-rigid.}
\end{wrapfigure}

This section describes the relativistic kinematics of a \dimensional{1} rigid object in Schwarzschild spacetime, or put more simply, a ``fishing line'' dangling outside a black hole and which is being reeled in or out. This rigid body solution leads to gravitational redshift interpretations for radial photons. This is an original result. It was claimed by \citet{brotas2006} who is demonstrably incorrect. 

The object could be visualised as an idealised fishing line, bar, rope, string, etc. Technically, it is only the motion which is rigid and not the object itself, as discussed in section~\ref{sec:bornrigid}. Hence, one could visualise a host of miniature rocket boosters affixed at all points of the line, or a system of pulleys which feed the line through like a conveyor belt. Historically, the idea sprang out of discussion about ``mining'' black holes, for which one proposal is to lower a box on a rope, fill it with energy, then raise it again. This leads to inquiries about the tensilve strength and kinematics of the rope.

One might assume the ``proper distance'' $dR=\Schw^{-1/2}dr$ is the only important radial distance, and hence that a rigid system must be described by a constant $R$\nobreakdash-spacing of particles in an object. However despite the ubiquity of this length measure, it does not account for length-contraction effects due to motion, as shown previously.

\subsection{Initial calculations}

The goal is to determine quantitative values such as the \velocity{} of particles on the line, and ultimately the proper metric. Firstly, I start with insight from special relativity, to be later applied to local inertial frames in Schwarzschild spacetime.

Let $K$ be a constant related to the speed of the line, defined as the rate of proper length passing per unit coordinate time $t$, at any given point. This definition is made by Brotas, who correctly asserts this leads to rigid motion. See the pipe diagram for one justification, and the proper metric calculation later. It seems to be the static nature of the geometry which supports this feature. Brotas defines $0<K<\infty$, I use $-\infty<K<\infty$, where as before positive $K$ refers to outward motion, and vice versa. This unboundedness does not lead to unbounded velocities, as the following lemma reassures. Let $L$ measure proper length along the rope, then
\eqn{\diff{L}{t}=K.}

\begin{figure}
\centering
\includegraphics[width=0.7\textwidth]{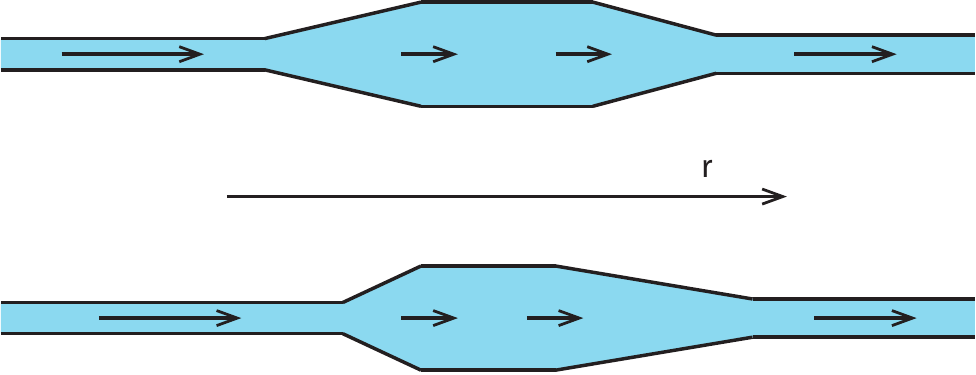}
\caption{Water flowing through oddly shaped pipes in Newtonian mechanics, an analogy to the length-contracted fishing line in Schwarzschild geometry. Assume the water flows into the pipe at a constant rate, and the flow is streamlined, frictionless, and incompressible. Its speed through the thinner sections is greater, and slower through the thicker sections. But the equation of continuity says the mass flow rate is the same everywhere, i.e. the amount of material or number of atoms passing any given $r$ remains constant. The varying thickness of the pipe is analogous to the varying time and length distortion of the fishing line in Schwarzschild spacetime. Stretching the pipe (bottom diagram) in the lengthwise does not change the flow rate (arrows), which is dependent only on the cross-section.}
\end{figure}

\begin{lemma}\label{lem:ropeminkowski}
In Minkowski space, suppose a rigid rope lies in a straight line and moves along its length, at constant speed $V$ relative to an inertial observer. Then the amount of material (proper length) of rope passing the observer is given by the rate
\eqn{\diff{L}{\tau}=V\gamma,}
where $L$ is the proper length of rope which has passed the observer, $\tau$ is the proper time of the observer, and $\gamma\equiv(1-V^2)^{-1/2}$ is the Lorentz factor.
\end{lemma}

\begin{proof}
Recall ``proper length'' means the length of the rope as measured in its own frame. The rope is length-contracted by a factor of $\gamma$, as measured by the observer. In the observer's frame, the rope will appear denser. Though the speed is $V$, due to the increased density from length-contraction the proper length will pass at the increased rate of $V\gamma$.
\end{proof}

\begin{lemma}
\label{lem:vgamma}
Suppose $V\gamma=q$ for some number $q$ at a given point, where $\gamma(V)$ is the Lorentz factor, and $q$ and $V$ may be positive or negative. Then
\eqn{\gamma=\sqrt{1+q^2}, \qquad\qquad V=\frac{q}{\sqrt{1+q^2}}.}
\end{lemma}

\begin{proof}
This follows from elementary algebra. Now $V\gamma=\frac{V}{\sqrt{1-V^2}}=q$, so $\frac{V^2}{1-V^2}=q^2$. Now $V=0 \Longleftrightarrow q=0$, otherwise invert both sides so $\frac{1-V^2}{V^2}=q^{-2}$, so $\frac{1}{V^2}=1+q^{-2}=1+\frac{1}{q^2}=\frac{1+q^2}{q^2}$. Then $V^2=\frac{q^2}{1+q^2}$, so $V=\frac{q}{\sqrt{1+q^2}}$, taking $V$ to have the same sign as $q$. Then from the starting assumption $V\gamma=q$, it follows $\gamma=\frac{q}{V}=\sqrt{1+q^2}$.
\end{proof}

Back to Schwarzschild geometry, the fishing line passes any stationary observer at a proper length rate of $\diff{L}{t}=K$. The use of coordinate time gives a convenient universal time with which to coordinate the motion. Then this is related to the proper time $t_{\rm shell}$ of the local Schwarzschild observer by
\eqn{\diff{L}{t_{\rm shell}}=\diff{L}{t}\diff{t}{t_{\rm shell}}=K\Schw^{-1/2}.}
(The label $\tau$ will be is reserved for particles on the line.) Since $t_{\rm shell}$ is a local inertial coordinate, the special-relativistic result of lemma~\ref{lem:vgamma} is valid, taking $q=K\Schw^{-1/2}$. It follows:
\eqn{\label{eqn:fishVgamma}\gamma=\sqrt{1+K^2\Schw^{-1}}=W(r)\Schw^{-1/2} \qquad\qquad V=\frac{K}{W(r)},}
after simplifying. Here $W(r)$ is defined by Brotas and has the following definition, and I also add a function $f(r)$ used shortly:
\eqn{W(r)\equiv\sqrtfish, \qquad\qquad f(r)\equiv\ffish=\Schw^{-1}W(r).}
One can show equation~\ref{eqn:fishVgamma} implies $\gamma\ge 1$ and $\abs{V}<1$ as required.

The \velocity{} in Schwarzschild coordinates follows from equation~\ref{eqn:velflow},
\eqnboxed{u^\mu=\left(f(r),K,0,0\right)}
after simplifying. In particular, $u^r=\diff{r}{\tau}=K$, so integrating gives
\eqnboxed{\label{eqn:fishr}r=K\tau+X,}
where $X$ is the constant of integration and is clearly the value of $r$ when $\tau=0$. 

Similarly, $u^t=\diff{t}{\tau}=f(r)$. A relation between $t$, $\tau$, and $r$ is needed. But $\diff{r}{\tau}=K$, so $f(r)=\diff{t}{\tau}=\diff{t}{r}\diff{r}{\tau}=K\diff{t}{r}$, so $\diff{t}{r}=\frac{1}{K}f(r)$. Integrate to get:
\eqnboxed{\label{eqn:fisht}t=\frac{1}{K}\int_X^r f(r')dr'}
where the lower bound follows from choosing $\tau=0$ to coincide with $t=0$, when the \coordinate{r} is $X$. A primitive function for the integral is messy but provided in Appendix B. It is not required for the calculations which can be done from the above expression.

Note that since the upper limit in the integral for $t$ is $r=X+K\tau$, the new variables $(\tau,X,\ldots)$ are expressed in terms of the Schwarzschild coordinates $(t,r,\ldots)$, where $\theta$ and $\phi$ are shared by both systems. The coordinate velocity
\eqn{\diff{r}{t}=\frac{K}{f(r)}}
follows directly from the \velocity{}.

The proper length of rope contained in interval $dr$ is, from equation~\ref{eqn:fishVgamma},
\eqn{\gamma dR=\frac{W(r)}{\Schw^{1/2}}\Schw^{-1/2}dr={W(r)\Schw^{-1}}dr=f(r)dr.}



When $\abs K$ is large, the interval tends to $\abs Kdr$, so is unbounded. When $K=0$ the interval is $dr$ as expected. The \acceleration is $\fvec a\equiv\nabla_{\fvec u}\fvec u$, which evaluates to
\eqn{\fvec a=\left(KM(r-2M)^{-2}f(r)^{-1},\frac{M}{r^2},0,0\right).}
The first term may also be expressed $KMr^{-2}\Schw^{-2}f(r)^{-1}$ or $\frac{KM}{r(r-2M)\sqrt{1-\frac{2M}{r}+K^2}}$. For $K=0$ this reduces to the expression for Schwarzschild observers $a^\mu=(0,\frac{M}{r^2},0,0)$ as expected \citep[p459]{hartle2003}. The magnitude describes the ``felt'' acceleration:
\eqn{\sqrt{\fvec a\cdot\fvec a} = \frac{M}{r^2}\fish^{-1/2} = \frac{M}{r^2W(r)}.}
The form is reassuringly similar to that of Schwarzschild observers, for whom $\sqrt{\fvec a\cdot\fvec a}=\frac{M}{r^2}\Schw^{-1/2}$ \citep[p459]{hartle2003}, and as $K\rightarrow 0$ this limit is obtained. Note the acceleration is always less for fishing line particles than Schwarzschild observers, for $K\ne 0$.

\subsection{Proper metric}

The projection tensor is given by equation~\ref{eqn:projectorcoord}, with $V$ and $\gamma$ from equation~\ref{eqn:fishVgamma}. Then,
\eqn{P_{\mu\nu}=\begin{pmatrix}
K^2 & -Kf(r) & 0 & 0 \\
-Kf(r) & f(r)^2 & 0 & 0 \\
0 & 0 & r^2 & 0 \\
0 & 0 & 0 & r^2\sin^2\theta \\
\end{pmatrix},}
after simplifying.

\begin{wrapfigure}{r}{0.5\textwidth}
  \centering
  \includegraphics[width=0.48\textwidth]{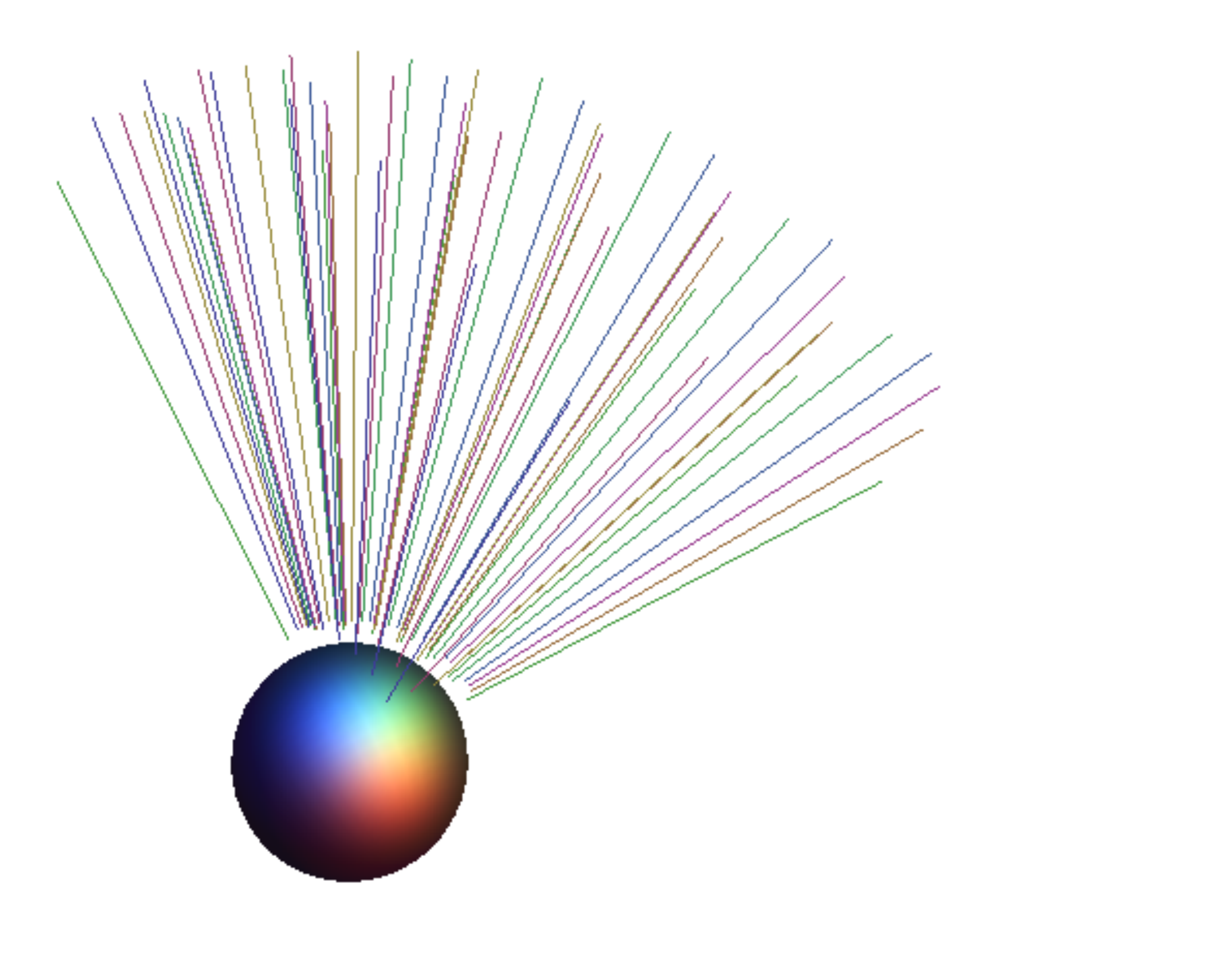}
  \caption{Many fishing lines dangling over a black hole. The picture is basically the same in the Schwarzschild frames, the fishing line frames, and $\mathbb R^3$, as the only difference is a radial length-contraction. Clearly this is not a rigid system in the angular coordinates - if the lines were connected sideways, they would rip apart when reeled in from above. But for a thin line which is not glued to others, only the radial direction matters.}
\end{wrapfigure}

The proper metric (see section~\ref{sec:propermetric}) requires the Schwarzschild coordinate position $x^\mu$ of particles to be expressed in terms of material coordinates $X^\nu$ including the proper time $\tau$. Now $\theta$ and $\phi$ are fixed for any given particle, so they can double as material coordinates: $X^\theta\equiv x^\theta$ and $X^\phi\equiv x^\phi$. Also equation~\ref{eqn:fishr} is $r=K\tau+X$, so choose $X$ as a radial material coordinate: $X^r\equiv X$. Equation~\ref{eqn:fisht} gives $t$ as a function of the $X^\nu$. Thus the Schwarzschild coordinates $x^\mu$ are expressed in terms of the material coordinates $X^\nu$, as required.

The derivatives $\pdiff{x^\mu}{X^i}$ are straightforward, apart from $\pdiff{t}{X}$. To repeat, $t=\frac{1}{K}\int_X^r f(r')dr'$. It is easy to miss that the upper limit $r=K\tau+X$ is also dependent on $X$. Hence, by a double application of the First Fundamental Theorem of Calculus which states $\diff{}{x}\int_c^x h(s)ds=h(x)$,
\eqn{\pdiff{t}{X}=\frac{f(r)-f(X)}{K},}
since $\diff{}{X}(I)=\diff{}{r}(I)\diff{r}{X}$ by the chain rule, writing $I$ for the integral expression. This is confirmed using the closed form expression for $t$.

It follows the proper metric is, after simplifying,
\eqn{\gamma_{ij}=\begin{pmatrix}
f(X)^2 & 0 & 0 \\
0 & r^2 & 0 \\
0 & 0 & r^2 \sin^2\theta \\
\end{pmatrix}.}

This gives the proper-frame distance as perceived by particles in the fishing line, and is interpreted as follows. In the radial direction the perceived distance between particles separated by $dX$ in the material manifold, is $f(X)dX$. This distance is measured in the orthogonal hypersurface to the motion. This formula applies over all time, so the proper-frame distance between these particles remains context, and thus they remain Born-rigid. $X$ is also the initial \coordinate{r} at $t=0$, so a Schwarzschild radial distance of $dX$ is perceived as $f(X)dX$ by the particles.

The angular components $\gamma_{\theta\theta}$ and $\gamma_{\phi\phi}$ are unchanged from those in the spacetime metric $g_{\mu\nu}$. This is an illustration of the familiar principle that there are no length-contraction (or expansion) effects orthogonal to the motion. However, $r$ is not constant over time for a given particle, for $K\ne 0$. Hence the radial distances change over time, and there is not rigidity in this direction. A Born-rigid object extended in \dimension{3}s will rip apart under fishing line motion (if $K\ne 0$). However the idealised \dimensional{1} fishing line remains Born-rigid. An alternate approach is to accept the line width as nonzero but small, and then allow for slight deviations from Born-rigidity, where the latter remains an extremely good approximation.





Limiting cases provide helpful confirmation. As $K\rightarrow 0$, $g_{11}\rightarrow \Schw^{-1/2}=dR$ as expected, which is the proper distance for Schwarzschild observers. As $M\rightarrow 0$ and/or $r\rightarrow\infty$, $\gamma_{11}\rightarrow \sqrt{1+K^2}$. This implies Minkowski space, and indeed this fits lemma~\ref{lem:vgamma} with $q=K$. If in addition $K\rightarrow 0$, $\sqrt{1+K^2}\rightarrow 1$ as expected.

\citet[\S6]{natario2014} also discusses the fishing line.

\subsection{Coordinate transform}
There are many coordinate systems for the Schwarzschild geometry already in common use. Each is chosen for a particular application, for instance geometric reasons, or following light paths or free-falling particles. I provide a new coordinate system $(\tau,X,\theta,\phi)$ which is adapted to the fishing line kinematics. $X$ is constant for a particle on the line and corresponds to the Schwarzschild \coordinate{r} when $\tau=0$. $\tau$ is the proper time for particles on the line. $\theta$ and $\phi$ are the usual angular spherical coordinates. They are synchronous comoving coordinates for a Born-rigid fishing line.

After substitution and simplification, one ultimately obtains
\eqn{ds^2 = -d\tau^2 + \frac{2}{K}\left(f(X)\sqrtfish-1\right)d\tau dX + \left(\Schw^{-1}-\frac{1}{K^2}\Schw\left(f(r)-f(X)\right)^2\right)dX^2 + r^2d\Omega^2}
which is valid for $K\ne 0$. Here $d\Omega^2=d\theta^2+\sin^2\theta d\phi^2$, and $r$ is the Schwarzschild radial coordinate but interpreted as a function of the coordinates, $r=X+K\tau$, which I continue to use for brevity of expression. Note the cross-term $d\tau dX$, due to motion in the Schwarzschild case.

This metric has vanishing Ricci tensors $R_{\mu\nu}$ as required for any coordinatisation of the Schwarzschild geometry, as \citet{brotas2006} mentions. This follows from the tensor transformation law $R'_{\alpha\beta}=\pdiff{x^\gamma}{x'^\alpha}\pdiff{x^\delta}{x'^\beta}R_{\gamma\delta}$, so since the $R_{\gamma\delta}$ all vanish in Schwarzschild coordinates, they must vanish in all coordinate systems for this geometry. Similarly, the curvature scalar $R=0$ everywhere for this geometry.



\subsection{Infinite length of line}

\begin{wrapfigure}{r}{0.5\textwidth}
  \centering
    \includegraphics[width=0.48\textwidth]{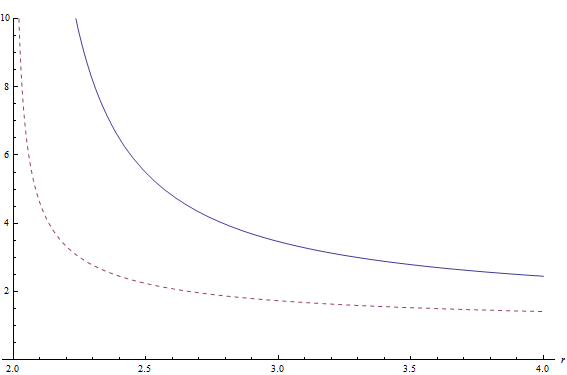}
  \caption{The function $f(r)$ (blue line) compared with the usual ``proper distance'' interval $\Schw^{-1/2}$, for $M=1$ and $K=1$. Visually, it is quite plausible the integral of $f$ --- the area under the curve --- is infinite near $r=2M$.}
\end{wrapfigure}

The length-contracted fishing line has the curious property than an infinite proper length will fit between the event horizon $r=2M$ and any point $r>2M$. This follows from the integral of $f(r)dr$, which describes the proper length, as given in Appendix B. One might guess a realistic material would break under the extreme conditions near the event horizon. But as $r\rightarrow 2M$, the \acceleration magnitude approaches $\frac{1}{4M\abs{K}}$, which curiously is finite. Perhaps a real line could even withstand the tension, aided by time-dilation, as \citet[\S2]{brown2013} expresses ``the weight of the rope redshifts away.''

The infinite length is \emph{not} simply due to $f(r)$ diverging to infinity as $r\rightarrow 2M$. Though this is true, this property does not imply the integral also diverges --- compare \citet[p2-49]{taylorwheeler2000}. Consider the familiar setting of functions from $\mathbb R$ to $\mathbb R$, e.g. $\frac{1}{x}$ and $\frac{1}{\sqrt{x}}$. Both are ``infinite'' at $x=0$ (more rigorously, both are undefined at $x=0$ but the limit is $+\infty$ when approaching from the positive side). And yet one integral diverges when the other converges. for $\int_0^1 \cdot dx$.

The point $r=2M$ is not included in the integral. But the integration is still the same $\int_{2M}{\cdot}$, whether or not this endpoint is included. (Technically, under Lebesgue integration, the ``measure'' of a single point within an interval of $\mathbb R$ is $0$. The convention is $0$ times $\infty$ is $0$, under Lebesgue integration. Thus one point makes no difference. The divergence is due to the behaviour \emph{near} $r=2M$, not the point itself.) Another way of expressing this is that the set of lengths $\int_{r_1}^{3M}f(r')dr'$ for $2M<r_1<3M$ is unbounded. The upper limit is an arbitrary choice, but close (in Schwarzschild coordinates) to $2M$.

\begin{figure}\centering
\includegraphics[width=0.7\textwidth]{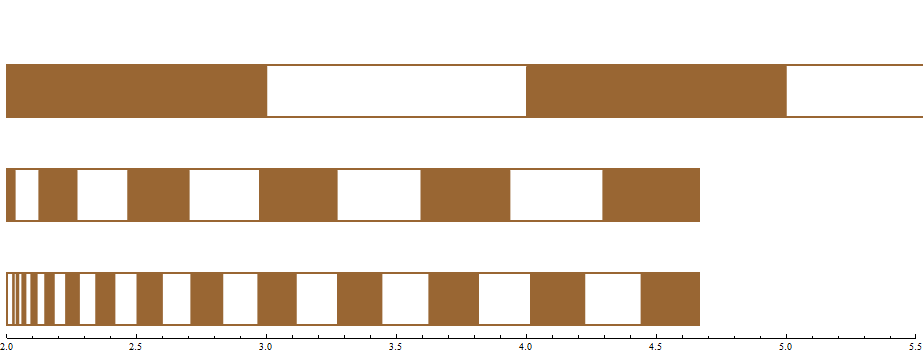}
\caption{A rigid Schwarzschild fishing line or rope as measured in various coordinate systems. In the top image, the rope is stationary in an inertial frame in flat spacetime. The scale shows the rope's regularly spaced one-unit markings, which serve as a visual ruler. The bottom pair is Schwarzschild spacetime with $M=1$, and here the scale represents the $r$-coordinate. In the middle image the rope is stationary and almost touching the event horizon. A proper length of $11$ units fits into a Schwarzschild $r$-coordinate interval of just over $2\frac{1}{2}$ units, due to the different geometry. In the lower image the rope is moving with $K=1$. Just $28$ segments are depicted of its infinite proper length.}
\end{figure}

The various radial length measures compare as follows, for $K\ne 0$:
\eqn{dL>dR>dr,}
since
\eqn{\diff{L}{r}=f(r)>\Schw^{-1}\sqrtSchw=\Schw^{-1/2}=\diff{R}{r}.}
See the diagram. If $K=0$ then $dL=dR>dr$.

\subsection{Application}
The kinematics of the rigid fishing line can be copied by a chain of observers to interpret redshifts as gravitational, or at least adjust the usual interpretation.

\emph{Example 1}. Suppose a pair of observers are stationary on the same radial line, at $r_1$ and $r_2$. Suppose a photon passes between them. Interpret the redshift as purely gravitational.

This example is chosen for familiarity, and to show consistency with the usual interpretation. A fishing line with $K=0$ has particle motion coinciding with the observers. Because this is a rigid system, every particle/observer on the line perceives constant proper-frame distance to its neighbours, and thus the photon redshift observed must be purely gravitational in origin since it can't be a Doppler shift. So the overall redshift is purely gravitational.

\emph{Example 2}. Suppose the \velocity[3] of an emitter at location $r_1$ is $V_1$ with respect to the local Schwarzschild observer there, and speed $V_2$ of a receiver at location $r_2$ relative to its local Schwarzschild observer. Use the convention $V_i>0$ for outward motion, as before. Suppose the emitter sends out a photon, which the emitter measures to have initial energy $E_1$. Firstly, the usual convention is to quantify redshift around what stationary Schwarzschild observers measure, a choice I affirm as the most ``natural''. Thus we convert the velocities to a stationary observer perspective.

At the emission event, the local Schwarzschild observer would measure a Doppler shift between themselves and the emitter. If $V_1>0$, the stationary observer would measure increased energy, which is $\sqrt\frac{1+V_1}{1-V_1}E_1$ from equation~\ref{eqn:DopplerSR}. Later, the photon arrives at $r_2$, experiencing an overall gravitational redshift by an energy factor $\sqrt\frac{1-\frac{2M}{r_1}}{1-\frac{2M}{r_2}}<1$. Then another Doppler correction is needed in converting to the receiver's frame, a factor of $\sqrt\frac{1-V_2}{1+V_2}$ which is a decrease in energy if $V_2>0$. Thus the breakdown is:
\begin{align}
\rm{Doppler} \qquad & \sqrt\frac{1+V_1}{1-V_1}\sqrt\frac{1-V_2}{1+V_2}\\
\rm{Gravitational} \qquad & \sqrt\frac{1-\frac{2M}{r_1}}{1-\frac{2M}{r_2}}
\end{align}
This standard interpretation is the most natural, based on the symmetries of the spacetime. Nevertheless, one can manipulate the interpretation using the equivalence principle.

Suppose for the same setup that a rigid fishing line moves concurrently with the emitter. One can solve for $K=V_1\gamma_1\Schw[r_1]^{1/2}$. There is no initial conversion Doppler shift in this case, as the emitter and other observer at the event have the same motion. Then all along the photon trajectory, frequency shifts are interpreted as purely gravitational by observers moving with the line. At the receiver location $r_2$ the line has speed $V'=\frac{K}{W(r_2)}$ relative to the local stationary observer by equation~\ref{eqn:fishVgamma}. In general this is not the same speed as the receiver, so an endpoint Doppler adjustment is needed, which can be deduced from a combined shift ``back'' to the local Schwarzschild observer then ``forward'' to the receiver. These factors are $\sqrt\frac{1+\frac{K}{W(r_2)}}{1-\frac{K}{W(r_2)}}$ and $\sqrt\frac{1-V_2}{1+V_2}$ respectively. The gravitational component must be whatever energy shift is unaccounted for. Comparing to the above calculations, overall:
\begin{align}
\rm{Doppler} \qquad & \sqrt\frac{1+\frac{K}{W(r_2)}}{1-\frac{K}{W(r_2)}} \sqrt\frac{1-V_2}{1+V_2}\\
\rm{Gravitational} \qquad & \sqrt\frac{1+V_1}{1-V_1} \sqrt\frac{1-\frac{2M}{r_1}}{1-\frac{2M}{r_2}} \sqrt\frac{1-\frac{K}{W(r_2)}}{1+\frac{K}{W(r_2)}}\\
\end{align}

Let's put in some numbers. Suppose $r_1=3M$ and $r_2=4M$, $V_1=\frac{1}{2}$ and $V_2=\frac{1}{5}$. Then the usual gravitational redshift factor is $\sqrt\frac{1-\frac{2M}{r_1}}{1-\frac{2M}{r_2}}=\sqrt\frac{2}{3}$. The $V_1$ special-relativistic Doppler factor is $\sqrt{3}$, and $\gamma(V_1)=\frac{2\sqrt{3}}{3}$. The $V_2$ Doppler factor is $\sqrt\frac{2}{3}$. Then the line matching the emitter has $K=\frac{1}{2}\frac{2}{\sqrt{3}}\frac{1}{\sqrt{3}}=\frac{1}{3}$. Then $W(r_2)=\sqrt\frac{11}{18}$, so $\frac{K}{W(r_2)}=\sqrt\frac{2}{11}$, and after simplifying the $\sqrt\frac{1+\frac{K}{W(r_2)}}{1-\frac{K}{W(r_2)}}=\frac{\sqrt{11}+\sqrt{2}}{3}\approx 1.58$.

The usual natural interpretation is of a Doppler energy shift by $\sqrt{3}\sqrt\frac{2}{3}=\sqrt{2}\approx 1.41$, and a gravitational shift by $\sqrt\frac{2}{3}\approx 0.82$ as above. The overall factor is an energy shift of $\frac{2}{\sqrt{3}}\approx 1.15$, a slight increase or blueshift.

Using the fishing line interpretation, the Doppler component is $\frac{\sqrt{11}+\sqrt{2}}{3} \sqrt\frac{2}{3}=\frac{\sqrt{66}+2\sqrt{3}}{9}\approx 1.29$, and the gravitational component $\frac{\sqrt{22}-2}{3}\approx 0.90$. The overall shift is the same, but there is a significant difference in its interpretation and makeup!

A possible extension to this analysis would be to use an arbitrary winding rate $K$ which is not necessarily matched to either emitter nor receiver.

\subsection{Alternate L coordinates}
The previous coordinate system used a radial coordinate which is comoving for particles on the line. Brotas makes a different choice of radial coordinates, using the proper length along the line, which he labels $\overline x$ and I label $L$. The purpose here is to repeat Brotas' construction and to compare the results.

Let $r_a$ be the Schwarzschild radial coordinate of a winch or fishing reel, as in Brotas. Natural choices such as $r_a=2M$ or $r_a=\infty$ do not work, as the proper length from these to any point within $2M<r<\infty$ would be infinite. One must pick some arbitrary $2M<r_a<\infty$.

For a given particle on the line, $\diff{L}{t}=K$, so $L=Kt+Y$ where $Y$ is the initial $L$\nobreakdash-value at $t=0$. Given $X$ as the starting \coordinate{r} used previously, then
\eqn{Y=\int_{r_a}^X f(r')dr'.}
As for relating $L$ to $r$, $\diff{L}{r}=f(r)$, so
\eqn{L=\int_{r_a}^r f(r')dr'.}
Then
\eqn{t=\frac{L-Y}{K}=\frac{1}{K}\int_X^r f(r')dr'.}
By comparison, Brotas has $t=\frac{L}{K}-\frac{1}{K}\int_{r_a}^rf(r')dr'$ which differs only in the swapping of $X$ and $r_a$. It appears Brotas' setup implies all particles must start at $r=r_a$ at $t=0$. This system has limited flexibility, and certainly his coordinates would not cover all of Schwarzschild spacetime. Yet in fairness these formulae are extremely similar, and perhaps Brotas intends the winch position $r_a$ to be unique for every fishing line, in which case the formulae would coincide. \citet[equation 4]{brotas2006} also has $r=2M+KL-K\tau$, where it remains to be further checked. These formulae are given without proof nor even a hint of their derivation, but merely the statement ``[t]he derivation/justification of these formula[e] is somewhat complicated'' \cite{brotas2006}.

As Brotas states, the Ricci tensors $R_{\alpha\beta}$ must vanish for every coordinate system for the Schwarzschild geometry. However his own coordinate system fails this. Also the curvature scalar $R$ must be identically $0$ for any Schwarzschild coordinate system, but this is not the case for Brotas' metric. 

$\diff{L}{\tau}=\diff{L}{t}\diff{t}{\tau}=Kf(r)$. So in the \coordinate{L}s, a given particle has \velocity{}
\eqn{u^\mu=(1,Kf(r),0,0).}

Brotas claims the $r=2M$ pathology of the Schwarzschild coordinates is moved to $r=\frac{2M}{1+K^2}$, but I do not find this result.

\subsection{Conclusion}
The description of a Born-rigid fishing line supports flexible interpretations of redshifts in the Schwarzschild geometry.

\section{Rigid rotation in Schwarzschild spacetime}
I derive the dynamics of a light rigid object rotating on its axis, in Schwarzschild geometry. The object is a spinning solid ball, with a black hole inside it at the centre, with all of the ball's matter outside the event horizon. In general, the particles do not follow geodesics. This research is original, and demonstrates that accelerated \dimensional{3} rigid motion does exist in curved spacetimes. This may also be used to interpret redshifts as gravitational, within the Schwarzschild geometry.

Orient Schwarzschild coordinates so that the rotation axis corresponds to the $\theta=0$ direction. Then the ``latitude'' or \coordinate{\theta} will remain fixed for a given particle, whereas the ``longitude'' or \coordinate{\phi} will change. Suppose the rotation rate is
\eqn{\diff{\phi}{t}=\Omega}
for all particles, employing the universality of coordinate time in an attempt to orchestrate the motion. Particles travel on circular paths, however these are not geodesics unless (i) the particle is on the ``equator'' and (ii) $\Omega$ matches the freefall rate, and in particular $r>3M$. Thus imagine the structure held together by either (i) its large spatial extent, extending over $180^\circ$, and balanced so that the whole structure moves in ``freefall'' even though the individual particles do not follow geodesics, or (ii) a rocket booster attached to every small clump of particles.

\subsection{Proper metric}
Assume the object is small enough or light enough not to affect the metric. A given particle's \position{} is
\eqn{x^\mu=(t,r,\theta,\phi_0+\Omega t),}
where $\phi_0$ is the \coordinate{\phi} at $t=0$. One can compute the \velocity{} as
\eqn{u^\mu=u^t\left(1,0,0,\Omega\right),}
where the time component is
\eqn{u^t=\left(1-\frac{2M}{r}-r^2\Omega^2\sin^2\theta\right)^{-1/2}.}
The proper metric is, \emph{in Schwarzschild coordinates}:
\eqn{\gamma_{ij}=\begin{pmatrix}
\Schw^{-1} & 0 & 0 \\
0 & r^2 & 0 \\
0 & 0 & \Schw(u^t)^2 r^2\sin^2\theta \\
\end{pmatrix}}

Note that $r$, $\theta$, $\Omega$, and $u^t$ are all constant for a given particle. Hence the proper metric is constant over time, and the system is rigid in all \dimension{3+1}s. The $\gamma_{rr}$ and $\gamma_{\theta\theta}$ components are unchanged from the Schwarzschild metric. But the motion is in the \direction{\phi}, and so the proper ruler distance is increased, relative to $g_{\phi\phi}=r^2\sin^2\theta$.

There is a dependency on $\Omega$. This means we cannot change the rotation rate and maintain Born-rigidity. The ``proper geometry'' is not Schwarzschild for $\Omega\ne 0$, but nonetheless we can in principle assemble the object while in rotation. These properties are identical for the spinning disc in Minkowski space.

\subsection{Limited extent}
There is a limit to the spatial extent of the rotating object. We require
\eqn{u^t=1-\frac{2M}{r}-r^2\Omega^2\sin^2\theta>0.}
A ``sufficient'' (but not ``necessary'') condition is $1-\frac{2M}{r}-r^2\Omega^2>0$, since $0\le\sin^2\theta\le 1$. If the structure doesn't extend to the equator, there will be a looser bound. In the extreme case, particles at the poles don't move at all. Rearranging the equation,
\eqn{\Omega<\frac{1}{r}\sqrtSchw.}
Differentiating, the maximum of this function occurs at $r=3M$, which has value $\frac{1}{3\sqrt{3}M}\approx 0.19M^{-1}$. Hence there is a global limit $\Omega<\frac{1}{3\sqrt{3}M}$, but the body will be limited by its weakest link.

\begin{wrapfigure}{r}{0.5\textwidth}
  \centering
    \includegraphics[width=0.48\textwidth]{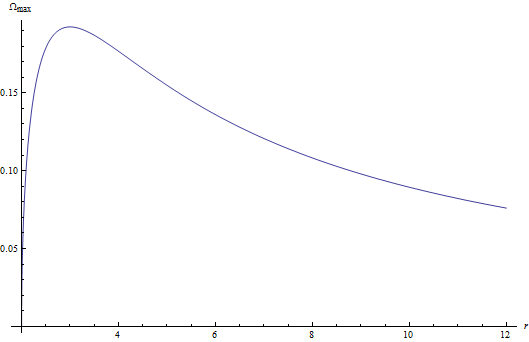}
  \caption{The maximum (``supremum'') spin rate $\Omega$ for particles at the equator, which is $\frac{1}{r}\sqrtSchw$ for $M=1$. At large $r$, a particle must move quickly through space to keep up --- think classically where we would have $v=r\Omega$. At small $r>2M$, time dilation is extreme and so even a low coordinate velocity seems much faster to a particle there. Also the acceleration must be increasingly angled outwards from the black hole, so there is less leeway for revolution.}
\end{wrapfigure}

\subsubsection{Application: Orbiting planets}
An orbiting body such as a planet will measure an increased orbit length due to length-contraction of its rulers, at least from the perspective of Schwarzschild observers. Assume the orbit is circular, then $\theta=\frac{\pi}{2}$ and the rotation rate is given by Kepler's law \cite[p200]{hartle2003}
\eqn{\Omega^2=\frac{M}{r^3}.}
(Recall freefall was not assumed prior to this point.) This also adds the restriction $r>3M$. Then $u^t$ becomes $(1-\frac{3M}{r})^{-1/2}$ as in \cite[p201]{hartle2003}. Then $\gamma_{\phi\phi}=r^2\frac{1-\frac{2M}{r}}{1-\frac{3M}{r}}$. The ruler length interval is $\sqrt{\gamma_{\phi\phi}}$, with total orbit length
\eqn{\sqrt\frac{1-\frac{2M}{r}}{1-\frac{3M}{r}}\cdot 2\pi r.}

For the Earth, $r\approx 1.5\times 10^8km$, and the Sun's mass is $M\approx 1.5km$. The expansion factor is $\approx 1+10^{-8}$, an increase of $\approx 9km$.

For a millisecond pulsar, which is both extremely dense and rotating extremely quickly, the length-contraction effects are likely extreme --- relative to, say, the Kerr metric. Does the literature account for the increased density due to length-contraction?

\subsection{Redshift}
This rigid system gives the framework or coordinate system for interpreting a certain class of redshifts as gravitational shifts. One must simply choose the value of $\Omega$, then plot a null worldline through the system. The particles encountered on this worldline all measure constant distances to their neighbours, and thus interpret the redshift as purely gravitational. Due to our setup, there are no restrictions on the path as being radial etc., however limits due to endpoint motions.

\subsection{Application to Ehrenfest paradox --- the spinning disc}
The Ehrenfest paradox of the rigidly spinning $2$D disc in Minkowski space is a special case with $M=0$ and $\theta=\frac{\pi}{2}$. Then the proper metric reduces to

\eqn{\gamma_{ij}=\begin{pmatrix}
1 & 0 & 0 \\
0 & r^2 & 0 \\
0 & 0 & (u^t)^2 r^2 \\
\end{pmatrix}}
where $u^t$ simplifies to
\eqn{u^t=\left(1-r^2\Omega^2\right)^{-1/2},}
so can reexpress $\gamma_{\phi\phi}=\frac{r^2}{1-r^2\Omega^2}$.

The \component{\theta} can be ignored, since there is no extent in this direction. An alternative is to use cylindrical coordinates $(r,\phi,z)$ along with $t$. This concurs with the result $\diag(1,\frac{r^2}{1-\omega^2 r^2},1)$ of \citet[\S2.5]{dewitt2011}, after rearranging the coordinate order.

Einstein realised the spinning disc would have non-Euclidean geometry, and for him this was a heuristic motivation in developing general relativity.


\section{Gravitational redshifts --- the FLRW galactic cable}
\label{sec:galaxycable}

\begin{wrapfigure}{l}{0.3\textwidth}
  \includegraphics[width=0.28\textwidth]{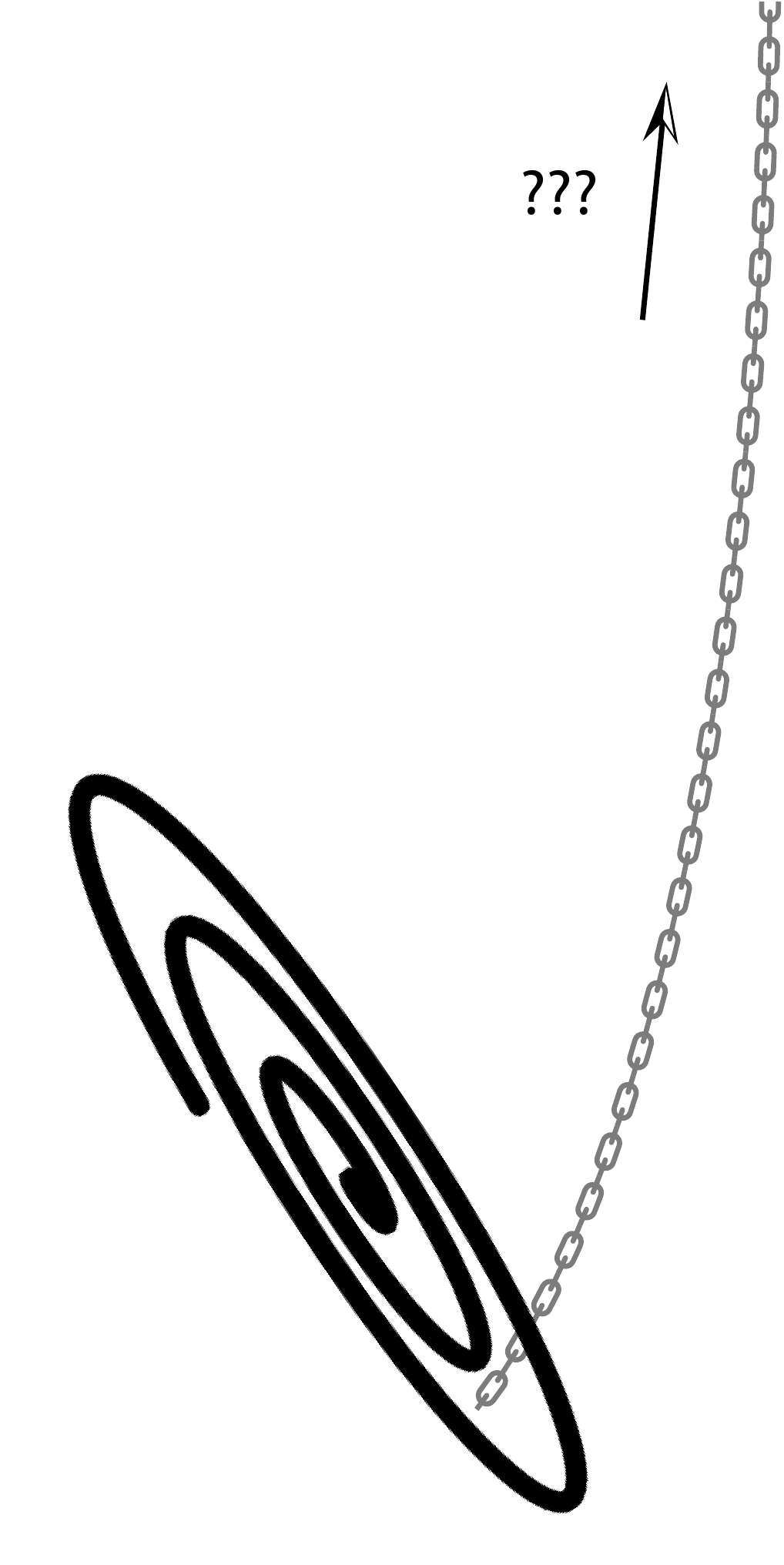}
  \caption{A rigid cable is tethered to the Milky Way galaxy. How far could it reach?}
\end{wrapfigure}

An interesting question is how far a rigid cable could reach through space before expanding space destroys it. This is interesting for two reasons. Firstly, using rigidity to define a fixed distance allows us to interpret cosmological redshifts as gravitational in origin. Secondly, when considering whether space ``expands'' it is pertinent to ask whether it is possible to make an infinitely long rigid rod. If not, it would seem that space really does expand, because something has to be pulling the rod apart, and there is no way to avoid that interpretation. This would be similar to the \emph{ergosphere} of a spinning black hole, in which a stationary observer is impossible due to ``spinning'' of space around the spinning black hole.

We investigate how far a rigid cable could reach through space before expanding space destroys it. As discussed previously, rigidity allows gravitational redshift interpretations. Also illustrative of expansion of space, in fact this example adds evidence to that interpretation. This section includes original results.

To begin with, I assume local proper velocities $v=-HD$ exactly counteract the Hubble speed. While this is insufficient in general, it is a good starting place.

\begin{lemma}
The changing ``proper distance'' $\dot D$ from the origin equals the local peculiar velocity of the distant object.
\end{lemma}
\begin{proof}
Rigorously, we seek $\dot D=0$, where dots refer to differentiation by cosmic time. $\dot D=\dot R\chi+R\dot\chi$ \cite[\S2]{davis+2003}, so $\dot\chi=-\frac{\dot R}{R}\chi=-H\chi$. This is a coordinate speed, but what speed would a local Hubble comover measure? From equation~\ref{eqn:FLRW4velocity}, $\diff{\chi}{t}=\frac{v}{R(t)}$, so $v=-\dot R\chi=-HD$ is the speed relative to the Hubble LIF.
\end{proof}

\subsection{Cable length for constant Hubble parameter}
Suppose $H(t)$ is constant over time. This implies either $H(t)=0$, a static universe; or a universe with cosmological constant and no other matter-energy, which grows exponentially and is known as \emph{de Sitter space}. The former is trivial, so we assume the latter. This could be a model of inflation, or of the distant future of our universe.

Suppose the cable has one end anchored with the local Hubble flow, say to the Milky Way galaxy. Suppose it lies in a straight line (geodesic), is Born-rigid, and can withstand any tension so long as its particles remain subluminal. (To be precise, to achieve Born-rigidity we assume adjustments to its displacement are not causally connected, but caused by a series of pre-programmed rocket boosters all along the cable. See section~\ref{sec:bornrigid}.)

Now, the change in proper distance $v\equiv\dot D$ is given by the Hubble flow $v=HD=\dot R\chi$. Hence for the cable to remain rigid its particles must move against the flow as $v=-HD$ to maintain equilibrium. Note, this will \emph{only} result in a rigid system if $H(t)$ is constant.

Now denote the proper-length along the rope by $s$, where $s=0$ at the anchor. This is a Lagrangian- or material coordinate, which would elsewhere be denoted $X$, and is fixed for any given particle of the rope. In the local Hubble frames, the rope is length-contracted, and so more of it will fit into a given ``proper distance'' $D$. From equation~\ref{eqn:FLRWlengthincrease},
\begin{align}
\diff{D}{s} &= \gamma^{-1}=\sqrt{1-v^2}=\sqrt{1-H^2D^2}\\
\int ds &= \int\frac{dD}{\sqrt{1-H^2D^2}}\\
s &= \frac{1}{H}\asin(HD),
\end{align}
using the initial condition $s=0$ when $D=0$. Locally, particle proper motion cannot exceed light speed, so we have the following:
\eqn{\label{eqn:Hubblesphere}\abs{v}<1 \qquad\qquad HD<1 \qquad\qquad D<H^{-1} \qquad\qquad \chi<\dot R^{-1}.}
At this maximum (technically ``supremum''), we have
\eqn{s_{\rm max}=H^{-1}\asin(1)=\frac{\pi}{2}H^{-1}.}
$H^{-1}$ is the Hubble radius. This means the longest possible cable, if one end is anchored with the Hubble flow, is $H^{-1}$ in ``proper distance'', but longer than this by a factor of $\frac{\pi}{2}\approx 1.6$ in terms of its proper length --- that is, the amount of cable.

\emph{Example:} In our universe, roughly $70\%$ of the matter-energy content is dark energy, so the above calculations (which are only valid for $100\%$ dark energy) seem a reasonable approximation. At the present, $H=H_0$ which in geometrical units is $\frac{H_0}{c}$; writing this in SI units gives $\approx 10^{-26}m^{-1}$. Thus the Hubble radius is $H_0^{-1}\approx 10^{26}m\approx 14Glyr$. A cable stretching to both horizons would have ``proper distance'' $D=2H_0^{-1}$, and proper length $\pi H_0^{-1}\approx 4\times 10^{26}m\approx 40Glyr$. Towards the endpoints, even this idealised cable can barely resist the expansion. This test scenario gives a sense that space really ``is'' expanding. In Newtonian cosmology by contrast, such a cable could have any extent.

This is not just a failure of our metric. (Incidentally, one extension to this problem would be the effect of the cable on the metric --- which here assumed inconsequential). What is supplying the force? 
In flat spacetime e.g. Rindler case, acceleration is the force ripping the system apart (if were extended too far). Here, endpoint has no acceleration, however due to expansion of space, i.e. curvature of spacetime --- the change in how the LIFs are ``joined'' together.

We can rearrange to get:
\eqn{\label{eqn:cableD}D=\frac{1}{H}\sin(sH) \qquad\qquad\qquad \chi=\dot R^{-1}\sin(sH).}
Note that for $sH\ll 1$, $D\approx s$ as expected since the counteracting Hubble velocities are small and thus length-contraction is minimal. Also in the limit $H\rightarrow 0$ we have $D\rightarrow s$ (using L'H\^opital's rule), as expected. Even though $\gamma$ increases without bound as $D\rightarrow H^{-1}$, the total length remains finite, in contrast to the ``fishing line'' case.

\subsection{Application}
The galactic cable can be used to interpret redshifts as a certain composite of Doppler and gravitational effects, by setting up observers or a coordinate system with the same kinematics. Firstly, perform a Doppler shift from a comoving galaxy frame to the frame corresponding to the endpoint of a rigid cable at the same event. Then, the infinitesimal redshifts all along the cable are interpreted as gravitational. There is no Doppler shift adjustment at home, the coordinate origin, because the cable is stationary with respect to the origin.

Suppose we have a dark energy dominated universe, and a comoving galaxy at $\chi$ emits a photon of wavelength $\lambda_0$ towards the coordinate origin. If $\chi<\dot R^{-1}$ then the galaxy is within the Hubble radius (equation~\ref{eqn:Hubblesphere}), and we could set up a rigid cable stretching from the origin to this point, or else a matching coordinate system or chain of observers. The endpoint at $\chi$ has local speed $\dot R\chi$ relative to the galaxy.

The local cable frame will observe a frequency $(1+z_{\rm Doppler})\lambda_0$ where the factor is
\eqn{1+z_{\rm Doppler}=\sqrt\frac{1+\dot R\chi}{1-\dot R\chi}}
by equation~\ref{eqn:DopplerSR}. Then, for the entirety of the photon's journey, the cable particles interpret the infinitesimal redshifts as gravitational, because they maintain constant proper-frame distance from each-other. There is no Doppler shift adjustment at the home galaxy (coordinate origin), because the cable is comoving with the flow there. Thus, this observer or coordinate system interprets the overall redshift as made up of a mixed $1+z_{\rm Doppler}$ component, and a gravitational component, which is nonzero in general \cite[\S3]{davis+2003}.

This shows the redshift interpretation is flexible. For an emitter galaxy beyond the Hubble sphere, light could never reach the home galaxy in a de Sitter universe.

For faster than exponential expansion, the event horizon is smaller than the Hubble sphere. So if this equation still held, you could make a rigid rod whose end we would never see! However this is a more general case and must be solved from the system of differential equations which follows.

\subsection{General case}
If $H$ is decreasing, the particles will require less peculiar velocity to counteract the expansion, so their length-contraction effects will diminish, and so the cable will expand outwards (as measured by the ``proper distance'' of Hubble flow comovers). The competing requirements of moving against the flow, and partial acquiescence to compensate for changing $H$, result in a feedback mechanism. This would be sensitive to the initial setup, although quite possibly might form a stable equilibrium also.

Using the same variables as previously, but with $D$ and $v$ now each functions of $s$ and $t$:
\eqnboxed{\pdiff{D}{s}=\gamma^{-1}=\sqrt{1-v^2},}
which is nothing more than local length-contraction. We also have
\eqnboxed{\pdiff{D}{t}=HD+v}
This follows from $\dot D=\dot R\chi+R\dot\chi=HD+v$ since $\dot\chi=\diff{\chi}{t}=\frac{v}{R(t)}$ which follows from the \velocity{} in equation~\ref{eqn:FLRW4velocity}. We also have initial conditions including $D(t,0)=0$ and $v(t,0)=0$. In the varying $H$ case, $v\ne-HD$ in general, although we would expect it to remain close to this value if $H$ changes only gently.

We could eliminate $v$ to give
\eqnboxed{\pdiffsquared{D}{s}+\left(\pdiff{D}{t}-HD\right)^2=1.}
This complements equation~\ref{eqn:rigid1D} for Minkowski space which does not have the ``$-HD$'' term.

I have not analysed these equations further. We likely need to specify further initial conditions, such as the initial velocity of the cable as it was unreeled. For instance, it could be assumed to unroll from Earth at a constant rate $K$, analogous to the variable used for the Schwarzschild fishing line. Assuming there is flexibility in the system, we could go further and consider acceleration of an endpoint, analogous to the Rindler chart. This would achieve great flexibility in the interpretation of cosmological redshifts.

\subsection{Further exploration}
One could compute the forces on the cable, or even consider elasticity and the propagation of waves along it. Could one set up a ``galaxy hook'', by which I dub a cable which would ``hook'' to the Hubble flow and tow a spacecraft? This is meant as an analogy to ``skyhooks'', which are a speculative concept for escaping Earth's gravitational pull by a long rotating cable.

\section{Conclusions}
There are lines of evidence to interpret the expansion of the universe as both ``stretching of space'' and motion through space. Parallels with the Milne model and Newtonian gravity might encourage the kinematic view, as does the natural interpretation of cosmological redshifts as purely Doppler. On the other hand, the superluminal recession velocity of distant galaxies, and the limited extent of a galactic cable suggest the expansion cannot be kinematic. Overall, the stretching of space is indeed a helpful metaphor, though it has benefited from critiques of misinterpretations. It does not describe any novel physical law, but is simply a manifestation of curved spacetime. It might be said the expansion is locally kinematic but not globablly so, at least there is no global inertial frame.

The nature of redshifts is highly flexible in interpretation, however in the Schwarzschild and FLRW cases symmetries of the matter and spacetime single out the standard interpretation as most natural. Distances measured by an observer are not given simply by the spacetime metric in general, but depend also on the observer's motion. Length-contraction does occur in general relativity, which has major conceptual consequences including the need to reevaluate certain astronomical calculations. Rigid acceleration of extended \dimensional{3} bodies is possible in curved spacetimes, in specific cases.





\newpage
\pagenumbering{Alph}
\renewcommand{\thepage}{- \Alph{page} -}

\section{Appendix A: Rigidity distance}
I propose a new distance measure ``rigidity distance'' or ``ruler distance'' as follows. Given the worldline of a particle as the starting event, extend a Born-rigid ruler out to the specified end event, if possible. Then the rigidity distance is the proper length of the ruler. The required inputs are:
\begin{itemize}
\item motion of start particle (instantaneous position, velocity, and acceleration)
\item a time slicing $t=const$ (default assumed in FLRW, Schwarzschild)
\item position of end particle, at least at same time $t$.
\end{itemize}

This assumes the ruler follows a spacetime geodesic. It is highly dependent on the start particle worldline. Distance is relative to the frame it is measured in, but there may be an orthogonal or other natural frame as discussed previously in section~\ref{sec:naturalcanonical}, or such a frame may be specified directly. There are limits to the extent of a rigid body, depending on the acceleration of the start particle and the spacetime geometry. In some cases such an extension is impossible, for instance from inside a black hole event horizon to outside.

The motion of the far end of the ruler will not match the motion of the destination location or particle, in general. Thus the measure is asymmetric in the sense that if is between two particles, it depends on the motion of the chosen start particle, but not the motion of the other. For nice properties, this measure has both existence and uniqueness.

One application would be to the ``tethered galaxy'' scenario, where in this case the tether is taken more literally.

\section{Appendix B: Schwarzschild distance integrals}
The Schwarzschild radial ``proper distance'' $dR=\Schw^{-1/2}dr$ may be integrated via the substitution $r=z^2$ \citep[p2-28]{taylorwheeler2000} or computer algebra to obtain a primitive function
\eqn{\sqrt{r(r-2M)}+2M\ln\left(\sqrt{r}+\sqrt{r-2M}\right).}
Then the proper distance is given by evaluating between some $r_1$ and $r_2$.

The proper length interval of fishing line is $f(r)dr=\ffish dr$, as measured by the local Schwarzschild observer. The following closed-form primitive function $F(r)$ satisfying $\diff{F(r)}{r}=f(r)$ follows from computer algebra plus subsequent manual simplifying:
\eqn{F(r)=rW(r)-2KM\ln\left(1+2Kf(r)+2K^2\Schw^{-1}\right)+\frac{M(1+2K^2)}{\sqrt{1+K^2}}\ln\left(r\left(\frac{M}{r}+W(r)^2+\sqrt{1+K^2}W(r)\right)\right).}
As above, the length is given by evaluation between some $r_1$ and $r_2$.

\section{Appendix C: Counteracting coordinates}
For an expanding FLRW universe, one could define a new coordinate system which attempts to ``remove'' the expansion in which comoving coordinates remain at constant ``proper distance'' $D$. The particles could be conceived to be tied by rigid ropes, or affixed with rocket boosters, to maintain constant $D$, as they are not freefalling in general. These coordinates are inspired in part by an online physics forum post, for which I have lost the source. But the derivation is done independently.

Actually this coordinate system fails to achieve its purpose, except in the specific case of a Hubble constant $H$ fixed over time. The model fails to consider length-contraction (as measured by Hubble comovers) due to motion against the Hubble flow. It is included primarily as a negative result, to argue against its validity should it be proposed by others.

The ``proper distance'' from $\chi$ to the origin is $R\chi$. Define a new radial coordinate by $X=R(t)\chi$, and keep the other coordinates the same. Taking the total derivative,
\eqn{dX=\dot R\chi dt+Rd\chi.}
Rearranging, $Rd\chi=dX-\dot R\chi dt$, and so
\eqn{R^2d\chi^2=dX^2+\dot R^2\chi^2 dt^2-2\dot R\chi dXdt.}
Substituting the Hubble parameter $H\equiv\dot R/R$,
\eqn{R^2d\chi^2=dX^2+H^2X^2dt^2-2HXdXdt.}
Then $\chi$ is eliminated from the right-hand side. Now substitute this expression into the FLRW metric:
\begin{align}
ds^2 &= -dt^2+dX^2+H^2X^2dt^2-2HXdtdX+R^2S_k^2d\Omega^2\\
&= -(1-H^2X^2)dt^2-2HXdtdX+dX^2+R^2S_k^2d\Omega^2.
\end{align}

Note the coordinate degeneracy when $g_{tt}=-(1-H^2X^2)=0$, which occurs when $X=\frac{1}{H}$, that is at the Hubble radius. This is to be expected, because the Hubble radius is where the Hubble recession velocity is $c$, so a particle attempting to stay at constant proper distance from the origin would require a proper velocity of $c$ inwards. So this coordinate singularity corresponds to a true physical limitation of the setup.

The energy-momentum tensor $T$ will transform accordingly.

Also the particles experience increasing time dilation with distance from the origin, as measured by Hubble flow observers who proper time is $t$. For a particle stationary in the new coordinates, $dX=d\theta=d\phi=0$, so its \fourvel $\fvec u$ satisfies
\eqn{-1=g(u,u)=-(1-H^2X^2)(u^t)^2,}
so
\eqn{u^t=(1-H^2X^2)^{-1/2}.}
This relates the coordinate and proper times via $u^t\equiv\diff{t}{\tau}$. So the harder they attempt to resist the expansion, the more time-dilated they are relative to comoving objects

Na\"ively it may seem this achieves the goal of setting up observers at constant distance. However this metric distance does not consider length-contraction effects, in contrast to the ``galaxy cable'' of section~\ref{sec:galaxycable}.




\newpage
\bibliographystyle{hapj}
\bibliography{biblio}

\begin{thebibliography}{72}
\expandafter\ifx\csname natexlab\endcsname\relax\def\natexlab#1{#1}\fi

\bibitem[{{Beig}(2004)}]{beig2004}
{Beig}, R. 2004, in Proceedings of the 7th Hungarian Relativity Workshop, ed.
  I.~{R{\'a}cz}

\bibitem[{{Beig} \& {Schmidt}(2003)}]{beigschmidt2003}
{Beig}, R., \& {Schmidt}, B. 2003, Classical and Quantum Gravity, 20, 889,
  gr-qc/0211054

\bibitem[{{Bell}(1976)}]{bell1976}
{Bell}, J. 1976, Progress in Scientific Culture, 1

\bibitem[{{Bini}(2014)}]{bini2014}
{Bini}, D. 2014, in General Relativity, Cosmology and Astrophysics, ed.
  J.~{Bi{\v c}{\'a}k} \& T.~{Ledvinka} (Springer), 67--90

\bibitem[{Born(1909)}]{born1909}
Born, M. 1909, Annalen der Physik, 335, 1,
  \href{http://en.wikisource.org/wiki/Translation:The_Theory_of_the_Rigid_Electron_in_the_Kinematics_of_the_Principle_of_Relativity}{English
  translation} by WikiSource

\bibitem[{{Brotas}(2006)}]{brotas2006}
{Brotas}, A. 2006, ArXiv e-prints, gr-qc/0609005

\bibitem[{{Brotas} \& {Fernandes}(2003)}]{brotasfernandes2003}
{Brotas}, A., \& {Fernandes}, J. 2003, ArXiv e-prints, physics/0307019

\bibitem[{{Brown}(2013)}]{brown2013}
{Brown}, A. 2013, Physical Review Letters, 111, 211301

\bibitem[{Brown(2018)}]{brown2018}
Brown, H. 2018, in Beyond Einstein: Perspectives on geometry, gravitation, and
  cosmology in the twentieth century, ed. D.~Rowe, T.~Sauer, \& S.~Walter
  (Springer), 51--66

\bibitem[{{Bunn} \& {Hogg}(2009)}]{bunnhogg2009}
{Bunn}, E., \& {Hogg}, D. 2009, American Journal of Physics, 77, 688

\bibitem[{{Carroll}(2004)}]{carroll2004}
{Carroll}, S. 2004, {Spacetime and geometry: An introduction to general
  relativity} (Addison Wesley)

\bibitem[{{Carter} \& {Quintana}(1972)}]{carterquintana1972}
{Carter}, B., \& {Quintana}, H. 1972, Royal Society of London Proceedings
  Series A, 331, 57

\bibitem[{Catt\`aneo(1958)}]{cattaneo1958}
Catt\`aneo, C. 1958, Il Nuovo Cimento, 10, 318

\bibitem[{{Chodorowski}(2005)}]{chodorowski2005}
{Chodorowski}, M. 2005, Publications of the Astronomical Society of Australia,
  22, 287, astro-ph/0503690

\bibitem[{{Chodorowski}(2011)}]{chodorowski2011}
------. 2011, Monthly Notices of the Royal Astronomical Society, 413, 585,
  0911.3536

\bibitem[{{Davis}(2004)}]{davis2004}
{Davis}, T. 2004, PhD thesis, University of New South Wales

\bibitem[{{Davis} \& {Lineweaver}(2004)}]{davislineweaver2004}
{Davis}, T., \& {Lineweaver}, C. 2004, Publications of the Astronomical Society
  of Australia, 21, 97, astro-ph/0310808

\bibitem[{{Davis} {et~al.}(2003){Davis}, {Lineweaver}, \& {Webb}}]{davis+2003}
{Davis}, T., {Lineweaver}, C., \& {Webb}, J. 2003, American Journal of Physics,
  71, 358

\bibitem[{{de Felice} \& {Clarke}(1990)}]{defeliceclarke1990}
{de Felice}, F., \& {Clarke}, C. 1990, {Relativity on curved manifolds}
  (Cambridge)

\bibitem[{{DeWitt}(2011)}]{dewitt2011}
{DeWitt}, B. 2011, Bryce DeWitt's lectures on gravitation, ed. by S.
  Christensen (Springer)

\bibitem[{{Epp} {et~al.}(2009){Epp}, {Mann}, \& {McGrath}}]{epp+2009}
{Epp}, R., {Mann}, R., \& {McGrath}, P. 2009, Classical and Quantum Gravity,
  26, 035015, 0810.0072

\bibitem[{{Eriksen} {et~al.}(1982){Eriksen}, {Mehlen}, \&
  {Leinaas}}]{eriksen+1982}
{Eriksen}, E., {Mehlen}, M., \& {Leinaas}, J. 1982, Physica Scripta, 25, 905

\bibitem[{{Fayngold}(2010)}]{fayngold2010}
{Fayngold}, M. 2010, ArXiv e-prints, 1001.0088

\bibitem[{{Ferrarese} \& {Bini}(2008)}]{ferraresebini2008}
{Ferrarese}, G., \& {Bini}, D. 2008, {Introduction to relativistic continuum
  mechanics} (Springer)

\bibitem[{{Gale}(2014)}]{gale2014}
{Gale}, G. 2014, in The Stanford Encyclopedia of Philosophy, ed. E.~{Zalta}
  (Stanford)

\bibitem[{{Gautreau} \& {Hoffmann}(1978)}]{gautreauhoffmann1978}
{Gautreau}, R., \& {Hoffmann}, B. 1978, Physical Review D, 17, 2552

\bibitem[{{Giulini}(2010)}]{guilini2010}
{Giulini}, D. 2010, in Minkowski spacetime: A hundred years later, ed.
  V.~Petkov (Springer), 83--132

\bibitem[{{Griffiths} \& {Podolsk{\'y}}(2009)}]{griffithspodolsky2009}
{Griffiths}, J., \& {Podolsk{\'y}}, J. 2009, {Exact space-times in Einstein's
  general relativity} (Cambridge University Press)

\bibitem[{{Gr{\o}n}(2004)}]{gron2004}
{Gr{\o}n}, {\O}. 2004, in Relativity in Rotating Frames, ed. G.~{Rizzi} \&
  M.~{Ruggiero} (Springer), 285--333

\bibitem[{{Gr{\o}n} \& {Elgar{\o}y}(2007)}]{gronelgaroy2007}
{Gr{\o}n}, {\O}., \& {Elgar{\o}y}, {\O}. 2007, American Journal of Physics, 75,
  151, astro-ph/0603162

\bibitem[{{Hamilton} \& {Lisle}(2008)}]{hamiltonlisle2008}
{Hamilton}, A., \& {Lisle}, J. 2008, American Journal of Physics, 76, 519,
  gr-qc/0411060

\bibitem[{{Hartle}(2003)}]{hartle2003}
{Hartle}, J. 2003, {Gravity: An introduction to Einstein's general relativity}
  (Addison Wesley)

\bibitem[{{Hawking} \& {Ellis}(1973)}]{hawkingellis1973}
{Hawking}, S., \& {Ellis}, G. 1973, {The large-scale structure of space-time}
  (Cambridge University Press)

\bibitem[{{Heckmann} \& {Sch{\"u}cking}(1955)}]{heckmannschucking1955}
{Heckmann}, O., \& {Sch{\"u}cking}, E. 1955, Zeitschrift f{\"u}r Astrophysik,
  38, 95

\bibitem[{{Hehl} {et~al.}(1976){Hehl}, von~der Heyde, {Kerlick}, \&
  Nester}]{hehl+1976}
{Hehl}, F., von~der Heyde, P., {Kerlick}, G., \& Nester, J. 1976, Rev. Mod.
  Phys., 48, 393

\bibitem[{{Hobson} {et~al.}(2006){Hobson}, {Efstathiou}, \&
  {Lasenby}}]{hobson+2006}
{Hobson}, M., {Efstathiou}, G., \& {Lasenby}, A. 2006, {General Relativity}
  (Cambridge)

\bibitem[{{Hogg}(1999)}]{hogg1999}
{Hogg}, D.~W. 1999, ArXiv e-prints, astro-ph/9905116

\bibitem[{{Landau} \& {Lifshitz}(1971)}]{landaulifshitz1971}
{Landau}, L., \& {Lifshitz}, E. 1971, {The classical theory of fields}
  (Pergamon)

\bibitem[{{Layzer}(1954)}]{layzer1954}
{Layzer}, D. 1954, The Astronomical Journal, 59, 268

\bibitem[{{Lyle}(2010)}]{lyle2010}
{Lyle}, S. 2010, in Space, time, and spacetime, ed. V.~{Petkov}
  (Springer-Verlag), 61--106

\bibitem[{{Lyle}(2014)}]{lyle2014}
------. 2014, in Springer handbook of spacetime, ed. A.~{Ashtekar} \&
  V.~{Petkov} (Springer), 117--139

\bibitem[{{Malament}(1995)}]{malament1995}
{Malament}, D.~B. 1995, Philosophy of Science, 62, 489

\bibitem[{{Mashhoon} \& {Muench}(2002)}]{mashhoonmuench2002}
{Mashhoon}, B., \& {Muench}, U. 2002, Annalen der Physik, 514, 532,
  gr-qc/0206082

\bibitem[{{Massa}(1974)}]{massa1974}
{Massa}, E. 1974, General Relativity and Gravitation, 5, 555

\bibitem[{{Maugin}(1971)}]{maugin1971}
{Maugin}, G. 1971, Annales de l'institut Henri Poincar{\'e} A: Physique
  th{\'e}orique, 15, 275

\bibitem[{{Maugin}(2013)}]{maugin2013}
------. 2013, Continuum mechanics through the twentieth century (Springer)

\bibitem[{{McCrea} \& {Milne}(1934)}]{mccreamilne1934}
{McCrea}, W., \& {Milne}, E. 1934, The Quarterly Journal of Mathematics, 5, 73

\bibitem[{{Milne}(1934)}]{milne1934}
{Milne}, E. 1934, The Quarterly Journal of Mathematics, 5, 64

\bibitem[{{Milne}(1935)}]{milne1935}
------. 1935, {Relativity, gravitation and world-structure} (Oxford)

\bibitem[{{Misner} {et~al.}(1973){Misner}, {Thorne}, \&
  {Wheeler}}]{misner+1973}
{Misner}, C., {Thorne}, K., \& {Wheeler}, J. 1973, {Gravitation} (W.H.~Freeman)

\bibitem[{Moore(2012)}]{moore2012}
Moore, T. 2012, A general relativity workbook (University Science Books)

\bibitem[{{Nat{\'a}rio, J.}(2014)}]{natario2014}
{Nat{\'a}rio, J.} 2014, General Relativity and Gravitation, 46, 1816

\bibitem[{{Norton}(1995)}]{norton1995}
{Norton}, J. 1995, Philosophy of Science, 62, 511

\bibitem[{{Norton}(1999)}]{norton1999}
------. 1999, in The Expanding Worlds of General Relativity, ed. H.~{Goenner},
  J.~{Renn}, J.~{Ritter}, \& T.~{Sauer}, Einstein Studies (Birkh{\"a}user),
  271--322

\bibitem[{{Norton}(2002)}]{norton2002}
------. 2002, in Inconsistency in Science, ed. J.~{Meheus} (Kluwer), 185--95

\bibitem[{{Norton}(2014)}]{norton2014}
------. 2014, in Physical theory: Method and interpretation, ed. L.~{Sklar}
  (Oxford), 185--228

\bibitem[{{Padmanabhan}(1996)}]{padmanabhan1996}
{Padmanabhan}, T. 1996, {Cosmology and astrophysics through problems}
  (Cambridge)

\bibitem[{{Peacock}(2008)}]{peacock2001}
{Peacock}, J.~A. 2008, ArXiv e-prints, 0809.4573, first online 2001

\bibitem[{{Rindler}(1977)}]{rindler1977}
{Rindler}, W. 1977, {Essential relativity: Special, general, and cosmological}
  (Springer-Verlag)

\bibitem[{{Rindler}(2006)}]{rindler2006}
------. 2006, {Relativity: Special, general and cosmological, 2nd edn} (Oxford)

\bibitem[{{Sachs} \& {Wu}(1977)}]{sachswu1977}
{Sachs}, R., \& {Wu}, H. 1977, {General relativity for mathematicians}
  (Springer-Verlag)

\bibitem[{{Schmidt}(1996)}]{schmidt1996}
{Schmidt}, H. 1996, General Relativity and Gravitation, 28, 899, gr-qc/9512006

\bibitem[{{Schutz}(2009)}]{schutz2009}
{Schutz}, B. 2009, {A first course in general relativity} (Cambridge)

\bibitem[{{S{\"o}derholm}(1982)}]{soderholm1982}
{S{\"o}derholm}, L. 1982, Physica Scripta, 26, 65

\bibitem[{{Taylor} \& {Wheeler}(1992)}]{taylorwheeler1992}
{Taylor}, E., \& {Wheeler}, J. 1992, {Spacetime physics: Introduction to
  special relativity, 2nd edn} (W.H.~Freeman)

\bibitem[{{Taylor} \& {Wheeler}(2000)}]{taylorwheeler2000}
------. 2000, {Exploring black holes: Introduction to general relativity} (New
  York: Addison-Wesley)

\bibitem[{{Tipler}(1996{\natexlab{a}})}]{tipler1996a}
{Tipler}, F. 1996{\natexlab{a}}, Monthly Notices of the Royal Astronomical
  Society, 282, 206

\bibitem[{{Tipler}(1996{\natexlab{b}})}]{tipler1996b}
------. 1996{\natexlab{b}}, American Journal of Physics, 64, 1311

\bibitem[{{Vickers}(2009)}]{vickers2008}
{Vickers}, P. 2009, Studies in History and Philosophy of Modern Physics, 40,
  197

\bibitem[{{Wald}(1984)}]{wald1984}
{Wald}, R. 1984, {General relativity} (University of Chicago Press)

\bibitem[{{Weinberg}(1972)}]{weinberg1972}
{Weinberg}, S. 1972, {Gravitation and cosmology: Principles and applications of
  the general theory of relativity} (Wiley)

\bibitem[{{Wernig-Pichler}(2006)}]{wernig-pichler2006}
{Wernig-Pichler}, M. 2006, PhD thesis, Universit{\"a}t Wien, gr-qc/0605025

\end{thebibliography}

I gratefully acknowledge use of the following free images: the \href{http://all-free-download.com/free-vector/download/r_is_for_rocket_144008.html}{rocket} by Andreas Preuss, \href{http://all-free-download.com/free-vector/download/learning_stationery_pencil_ruler_vectors_294665.html}{ruler} and \href{http://all-free-download.com/free-vector/download/vector_gear_chains_289029.html}{chain} by 3lian, \href{http://www.clker.com/clipart-spiral-7.html}{spiral} by hmime, and the \href{http://commons.wikimedia.org/wiki/File:Surface_normal.png}{surface with normal vectors} by Oleg Alexandrov. Hyperlinks are present here and elsewhere in the electronic version of this document.

\end{document}